\documentclass{article}

\usepackage{amsmath,amssymb,tikz} 
\usepackage{pict2e}
\usepackage{todonotes}
\usetikzlibrary{decorations.markings}

\advance\textheight by \topskip

\numberwithin{equation}{section}

\def\beq{\begin{equation}}
\def\eeq{\end{equation}}

\def\bit{\begin{itemize}}
\def\eit{\end{itemize}}

\makeatletter
\def\eqalign#1{\null\vcenter{\def\\{\cr}\openup\jot\m@th
\ialign{\strut$\displaystyle{##}$\hfil&$\displaystyle{{}##}$\hfil
\crcr#1\crcr}}\,}
\makeatother

\def\bigO{{\cal O}}
\newenvironment{proof}%
{\rm \trivlist \item[\hskip \labelsep{\bf Proof. }]}%
{\hspace*{\fill}$\Box$\endtrivlist}

\begin{document}
\tikzset{middlearrow/.style={
        decoration={markings,
            mark= at position 0.6 with {\arrow{#1}} ,
        },
        postaction={decorate}
    }
}

\newtheorem{theorem}{Theorem}
\newtheorem{acknowledgement}[theorem]{Acknowledgement}
\newtheorem{remark}[theorem]{Remark}
\newtheorem{lemma}[theorem]{Lemma}
\newtheorem{proposition}[theorem]{Proposition}

\numberwithin{equation}{section}
\numberwithin{theorem}{section}

\def\Xint#1{\mathchoice
{\XXint\displaystyle\textstyle{#1}}%
{\XXint\textstyle\scriptstyle{#1}}%
{\XXint\scriptstyle\scriptscriptstyle{#1}}%
{\XXint\scriptscriptstyle\scriptscriptstyle{#1}}%
\!\int}
\def\XXint#1#2#3{{\setbox0=\hbox{$#1{#2#3}{\int}$ }
\vcenter{\hbox{$#2#3$ }}\kern-.59\wd0}}
\def\ddashint{\Xint=}
\def\dashint{\Xint-}

\title{Asymptotics of Hankel determinants with a one-cut regular potential and Fisher-Hartwig singularities}
\author{Christophe Charlier\footnote{Institut de Recherche en Math\'ematique et Physique,  Universit\'e
catholique de Louvain, Chemin du Cyclotron 2, B-1348
Louvain-La-Neuve, BELGIUM. E-mail: christophe.charlier@hotmail.com}}

\maketitle


\begin{abstract}
We obtain asymptotics of large Hankel determinants whose weight depends on a one-cut regular potential and any number of Fisher-Hartwig singularities. This generalises two results: 1) a result of Berestycki, Webb and Wong \cite{BerWebbWong} for root-type singularities, and 2) a result of Its and Krasovsky \cite{ItsKrasovsky} for a Gaussian weight with a single jump-type singularity. We show that when we apply a piecewise constant thinning on the eigenvalues of a random Hermitian matrix drawn from a one-cut regular ensemble, the gap probability in the thinned spectrum, as well as correlations of the characteristic polynomial of the associated conditional point process, can be expressed in terms of these determinants.
\end{abstract}

\section{Introduction and main results}
We are interested in the large $n$ asymptotics of the Hankel determinant
\begin{equation}\label{Hankel determinants 1}
D_{n}(\vec{\alpha},\vec{\beta},V,W) = \det \left( \int_{\mathbb{R}} x^{j+k-2}w(x)dx \right)_{j,k=1,...,n},
\end{equation}
where the weight $w$ is defined by
\begin{equation}\label{weight_general}
w(x) = e^{-nV(x)}e^{W(x)}\omega(x), \qquad \omega(x) =  \prod_{j=1}^{m} \omega_{\alpha_{j}}(x)\omega_{\beta_{j}}(x), \qquad m \in \mathbb{N},
\end{equation}
and for each $k \in \{1,...,m\}$, we have
\begin{equation}\label{first def of omega}
\omega_{\alpha_{k}}(x) = |x-t_{k}|^{\alpha_{k}}, \qquad \omega_{\beta_{k}}(x) = \left\{ \begin{array}{l l}
e^{i\pi\beta_{k}}, & \mbox{ if } x < t_{k}, \\
e^{-i \pi \beta_{k}}, & \mbox{ if } x > t_{k}.
\end{array}  \right.
\end{equation}
The number of Fisher-Hartwig (FH) singularities is denoted by $m \in \mathbb{N}$, and $t_{1},...,t_{m}\in \mathbb{R}$ are the locations of these singularities. For each $k \in \{1,...,m\}$, the root-type singularity and the jump-type singularity located at $t_{k}$ are parametrized by $\alpha_{k} \in \{ z \in \mathbb{C} : \Re z > -1 \}$ and $\beta_{k} \in \mathbb{C}$ respectively. Since for any $n_{0}\in \mathbb{Z}$, we have $\omega_{\beta_{k}+n_{0}} = (-1)^{n_{0}} \omega_{\beta_{k}}$, we can assume without loss of generality that $\Re\beta_{k} \in (\frac{-1}{2},\frac{1}{2}]$. The potential $V: \mathbb{R} \to \mathbb{R}$ is real analytic, the function $W : \mathbb{R} \to \mathbb{R}$ is continuous on $\mathbb{R}$ and they satisfy
\begin{equation*}
\lim_{x\to\pm\infty} V(x)/\log|x| = + \infty \qquad \mbox{ and } \qquad W(x) = \bigO(V(x)), \mbox{ as } |x| \to \infty.
\end{equation*}
The $\bigO$ notation means that there exist two constants $c_{1},c_{2} >0$ such that $|W(x)| \leq c_{1} |V(x)|$ for all $|x| > c_{2}$. Thus $w$ is integrable on $\mathbb{R}$ and has finite moments for $n > c_{1}$. By Heine's formula, $D_{n}(\vec{\alpha},\vec{\beta},V,W)$ admits the following $n$-fold integral representation:
\begin{equation}\label{Hankel determinants}
D_{n}(\vec{\alpha},\vec{\beta},V,W) = \frac{1}{n!}\int_{\mathbb{R}^{n}} \prod_{1\leq j < k \leq n}(x_{k}-x_{j})^{2} \prod_{j=1}^{n}w(x_{j})dx_{j}.
\end{equation}
This determinant depends on several parameters. To simplify the notation, we denoted it by $D_{n}(\vec{\alpha},\vec{\beta},V,W)$, where the dependence in $t_{1}$,...,$t_{m}$ is omitted and where we write $\vec{\alpha}$ and $\vec{\beta}$ for $(\alpha_{1},...,\alpha_{m})$ and $(\beta_{1},...,\beta_{m})$ respectively. 

\vspace{0.2cm}The analogous situation for Toeplitz determinants (i.e. determinants constant along the diagonals) with a weight defined on the circle with FH singularities has been studied by many authors (e.g. \cite{FH, Widom2, Basor, Basor2, BS, Ehrhardt}), and the most general results can be found in \cite{DIK,DeiftItsKrasovsky}. Asymptotics for large Hankel determinants associated to a weight of the form
\begin{equation}\label{weight Deift Its Krasovsky}
e^{W(x)}\omega(x) \chi_{[-1,1]}(x),
\end{equation}
where $\chi_{[-1,1]}$ is the characteristic function of $[-1,1]$, were also obtained in \cite{DIK} by using a relation between these Hankel determinants and Toeplitz determinants with FH singularities. Note that the weight \eqref{weight Deift Its Krasovsky} is not a particular case of \eqref{weight_general} (and vice-versa).

\vspace{0.2cm}Hankel determinants of the form \eqref{Hankel determinants} appear naturally in random matrix theory. Consider the set of $n \times n$ Hermitian matrices $M$ endowed with the probability distribution
\begin{equation}\label{random matrix ensemble}
\frac{1}{\widehat{Z}_n}e^{-n \mathrm{Tr} V(M)}dM, \qquad dM = \prod_{i=1}^{n} dM_{ii} \prod_{1 \leq i<j \leq n } d \Re M_{ij}d \Im M_{ij},
\end{equation}
where $\widehat{Z}_n$ is the normalisation constant, also called the partition function. This distribution of matrices is invariant under unitary conjugations and induces a probability distribution on the eigenvalues $x_{1},...,x_{n}$ of $M$ which is of the form
\begin{equation}\label{distribution eigenvalues}
\frac{1}{n! Z_{n}} \prod_{1\leq j < k \leq n}(x_{k}-x_{j})^{2} \prod_{j=1}^{n} e^{-nV(x_{j})}dx_{j}, \qquad (x_{1},...,x_{n}) \in \mathbb{R}^{n},
\end{equation}
where the partition function $Z_{n}$ can be rewritten as
\begin{equation}
Z_{n} = D_{n}(\vec{0},\vec{0},V,0).
\end{equation}
Probably the most studied random matrix ensemble is the Gaussian Unitary Ensemble (GUE) \cite{AkeBaiFra, AndGuiZei, Mehta}, and corresponds to the case $V(x) = 2x^{2}$. In that instance, the partition function is explicitly known (see e.g. \cite{Mehta})
\begin{equation}\label{partition GUE}
\begin{array}{r c l}
\displaystyle D_{n}(\vec{0},\vec{0},2x^{2},0) & = & \displaystyle (2\pi)^{\frac{n}{2}} 2^{-n^{2}} n^{-\frac{n^{2}}{2}} \prod_{j=1}^{n-1} j! \\
& = & 2^{-n^{2}}e^{-\frac{3}{4}n^{2}} (2\pi)^{n} n^{-\frac{1}{12}} e^{\zeta^{\prime}(-1)}\left( 1+\bigO(n^{-1}) \right), \quad \mbox{ as } n \to \infty,
\end{array}
\end{equation}
where $\zeta$ is Riemann's zeta-function.

\vspace{0.2cm}The potentials $V$ we are interested in are described in terms of properties of the equilibrium measure $\mu_{V}$, which is the unique minimizer of the functional
\begin{equation}
\iint \log |x-y|^{-1} d\mu(x)d\mu(y) + \int V(x)d\mu(x)
\end{equation}
among all Borel probability measures $\mu$ on $\mathbb{R}$. It is known (see e.g. \cite{SaTo}) that this measure and its support (denoted $\mathcal{S}$) are completely characterized by the Euler-Lagrange variational conditions
\begin{align}
2 \int_{\mathcal{S}} \log |x-s| d\mu_{V}(s) = V(x) - \ell, & & \mbox{ for } x \in \mathcal{S}, \label{var equality} \\
2 \int_{\mathcal{S}} \log |x-s| d\mu_{V}(s) \leq V(x) - \ell, & & \mbox{ for } x \in \mathbb{R}\setminus \mathcal{S}. \label{var inequality}
\end{align}
\newpage
\hspace{-0.55cm}A potential $V$ is called one-cut regular if it satisfies the following conditions (see e.g. \cite{BleIts, ClaeysGravaMcLaughlin, BerWebbWong}):
\begin{enumerate}
\item $V : \mathbb{R}\to\mathbb{R}$ is analytic.
\item \vspace{-0.25cm} $\lim_{x\to\pm\infty}V(x)/\log|x| = +\infty$.
\item \vspace{-0.2cm} The inequality \eqref{var inequality} is strict.
\item \vspace{-0.2cm} The equilibrium measure is supported on $\mathcal{S} = [a,b]$ and is of the form \\ $d\mu_{V}(x) = \psi(x)\sqrt{(b-x)(x-a)}dx$, where $\psi$ is positive on $[a,b]$. 
\end{enumerate}

The above conditions imply that $\psi$ is also real analytic. This is a consequence of \cite[Theorem 1.38]{DeiKriMcL} (for more details, see the beginning of Section \ref{Section:steepest}). In the present work, we restrict ourselves to the class of one-cut regular potentials whose equilibrium measure is supported on $[-1,1]$ instead of $[a,b]$. This is without loss of generality, as can be easily seen from a change of variables in \eqref{Hankel determinants}, \eqref{var equality} and \eqref{var inequality}. For the reader's convenience, we show that explicitly in Remark \ref{Remark: support} below. An example of a one-cut regular potential is the Gaussian potential $V(x) = 2x^{2}$. In this case, it is well-known \cite{SaTo} that $\ell = 1+2\log 2$ and $\psi(x) = \frac{2}{\pi}$. 

\vspace{0.2cm}A lot of results are available in the literature for large $n$ asymptotics of $D_{n}(\vec{\alpha},\vec{\beta},V,W)$ for particular values of the parameters, and we briefly discuss them here. 

\vspace{0.2cm}We start with Hankel determinants without singularities. If $V$ is a polynomial, is one-cut regular and such that all zeros of $\psi(x)$ are nonreal, Johansson \cite{Johansson2} obtained rigorously the  asymptotics of $\frac{D_{n}(\vec{0},\vec{0},V,W)}{D_{n}(\vec{0},\vec{0},V,0)}$ (see \cite[Theorem 2.4]{Johansson2} for the precise assumptions on $W$). In particular, this implies a central limit theorem for the linear statistics, i.e. it gives information about the global fluctuation properties of the spectrum around the equilibrium measure. For a polynomial one-cut regular potential $V$, large $n$ asymptotics for the partition function $D_{n}(\vec{0},\vec{0},V,0)$ have been obtained via the Riemann-Hilbert method in \cite{ErcMcL} (under certain assumptions on the coefficients of $V$, see \cite[Theorem 1.1 and above]{ErcMcL}), and via deformation equations in \cite{BleIts} (under further technical assumptions on $V$, see \cite[equation (1.5) and remark (2) below]{BleIts}). There are also results available for partition functions in more general situations, see e.g. \cite{ClaeysGravaMcLaughlin} when the equilibrium measure is supported on two intervals.

\vspace{0.2cm}Hankel determinants with root-type singularities are related to the statistical properties of the characteristic polynomial $p_{n}(t) = \prod_{j=1}^{n}(t-x_{j})$ of \eqref{distribution eigenvalues}. There is numerical evidence and conjectures of links between $p_{n}(t)$ and the behaviour of the Riemann $\zeta$-function along the critical line (see e.g. \cite{KeaSna}). In \cite{Krasovsky}, Krasovsky studied the correlations between $|p_{n}(t_{1})|$,...,$|p_{n}(t_{m})|$, and powers of these quantities, when $-1<t_{1}<...<t_{m}<1$ are fixed in the case of GUE. The large $n$ asymptotics of these correlations are given by
\begin{multline}\label{asymptotics alpha Krasovsky}
\mathbb{E}_{\mathrm{GUE}}\bigg( \prod_{k=1}^{m} |p_{n}(t_{k})|^{\alpha_{k}} \bigg) = \frac{D_{n}(\vec{\alpha},\vec{0},2x^{2},0)}{D_{n}(\vec{0},\vec{0},2x^{2},0)} = 2^{-\mathcal{A}n} \prod_{1 \leq j < k \leq m} (2|t_{j}-t_{k}|)^{-\frac{\alpha_{j}\alpha_{k}}{2}} \\ 
\times  \prod_{j=1}^{m} \frac{G\left( 1+\frac{\alpha_{j}}{2} \right)^{2}}{G(1+\alpha_{j})} \Big(2\sqrt{1-t_{j}^{2}}n\Big)^{\frac{\alpha_{j}^{2}}{4}} \exp \left( \frac{\alpha_{j} n}{2}(2t_{j}^{2}-1) \right)  \left( 1 + \bigO \left( \frac{\log n}{n} \right) \right),
\end{multline}
where $\Re \alpha_{k} >-1$ for $k = 1,...,m$, $\mathcal{A} = \sum_{j=1}^{m} \alpha_{j}$ and $G$ is Barnes' $G$-function. 

\vspace{0.2cm}This result was recently generalized for the class of one-cut regular potentials by Berestycki, Webb and Wong. In \cite[Theorem 1.1]{BerWebbWong}, they proved that a sufficiently small power of the absolute value of the characteristic polynomial $p_{n}(t)$ of a one-cut regular ensemble converges weakly in distribution to a Gaussian multiplicative chaos measure. It was crucial in their analysis to obtain the large $n$ asymptotics of 
\begin{equation}\label{gaussian multiplicative chaos}
\mathbb{E}_{V}\bigg( \prod_{j=1}^{n} e^{W(x_{j})} \prod_{k=1}^{m} |p_{n}(t_{k})|^{\alpha_{k}} \bigg) = \frac{D_{n}(\vec{\alpha},\vec{0},V,W)}{D_{n}(\vec{0},\vec{0},V,0)},
\end{equation}
where $\alpha_{k} \in (-1,\infty)$ for $k=1,...,m$ and $W \in C^{\infty}(\mathbb{R})$ is a compactly supported function, analytic in a neighbourhood of $[-1,1]$. In particular, they proved two conjectures of Forrester and Frankel \cite{ForFra} concerning asymptotics of Hankel determinants with root-type singularities.

\vspace{0.2cm}Only limited results are available concerning Hankel determinants with jump discontinuities. Such determinants allow to generalize the correlations \eqref{gaussian multiplicative chaos}. Let us define the argument of the characteristic polynomial as follows:
\begin{equation}
\arg p_{n}(t) = \sum_{j=1}^{n} \arg(t-x_{j}), \qquad \mbox{ where } \quad \arg(t-x_{j}) = \left\{ \begin{array}{l l}
0, & \mbox{ if } x_{j}<t, \\
-\pi, & \mbox{ if } x_{j} >t.
\end{array} \right.
\end{equation}
We have
\begin{equation}\label{generalise correlation referee}
\mathbb{E}_{V}\bigg( \prod_{j=1}^{n} e^{W(x_{j})} \prod_{k=1}^{m} |p_{n}(t_{k})|^{\alpha_{k}} e^{2 i \beta_{k}\arg p_{n}(t_{k})} \bigg) =  \frac{D_{n}(\vec{\alpha},\vec{\beta},V,W)}{D_{n}(\vec{0},\vec{0},V,0)}\prod_{k=1}^{m} e^{-i\pi \beta_{k}}.
\end{equation}
Hankel determinants with jump singularities appear also when we thin the eigenvalues of a random matrix. Consider the point process of the $n$ eigenvalues of an $n\times n$ random GUE matrix. An independent and constant thinning consists of removing each of these eigenvalues with a certain probability $s \in [0,1]$. The remaining (thinned) eigenvalues are denoted $y_{1},...,y_{N}$, where $N$ is a random variable following the binomial distribution $\mathrm{Bin}(n,1-s)$. The gap probability in the thinned spectrum can be expressed in terms of a Hankel determinant with a single jump discontinuity (i.e. $m=1$) as follows: let $t_{1} \in \mathbb{R}$ and consider the partition of $\mathbb{R}^{n}$
\begin{equation}
A_{k} = \big\{ (x_{1},...,x_{n}) \in \mathbb{R}^{n} : \sharp \{ x_{i}:x_{i}<t_{1} \}=k \big\}, \qquad \bigsqcup_{k=0}^{n} A_{k} = \mathbb{R}^{n}.
\end{equation}
Since we thin the eigenvalues independently of each other and independently of their positions, we have 
\begin{equation}
\mathbb{P}_{\mathrm{GUE}}\big(\sharp \{ y_{i}:y_{i}<t_{1} \}=0|\sharp \{ x_{i}:x_{i}<t_{1} \}=k\big) = s^{k}.
\end{equation}
Therefore, using the integral representation for Hankel determinants \eqref{Hankel determinants}, one has
\begin{align}\label{gap prob m=1}
& \mathbb{P}_{\mathrm{GUE}}\big(\sharp\{ y_{i} : y_{i} < t_{1} \} = 0 \big) = \sum_{k=0}^{n}s^{k}\mathbb{P}_{\mathrm{GUE}}\big(\sharp \{ x_{i}:x_{i}<t_{1} \}=k\big) \\ 
& \hspace{0.2cm} = \frac{1}{n!Z_{n}} \sum_{k=0}^{n} \int_{A_{k}} \prod_{1\leq j < k \leq n}(x_{k}-x_{j})^{2} \prod_{j=1}^{n} e^{-2nx_{j}^{2}} \left\{ \begin{array}{l l}
\hspace{-0.06cm}s, & \hspace{-0.14cm}\mbox{ if } x_{j} < t_{1} \\
\hspace{-0.06cm}1, & \hspace{-0.14cm}\mbox{ if } x_{j}>t_{1}
\end{array} \right\} dx_{j} = e^{in\pi \beta_{1}} \frac{D_{n}(\vec{0},\vec{\beta},2x^{2},0)}{D_{n}(\vec{0},\vec{0},2x^{2},0)}, \nonumber
\end{align}
where in the last equality $\vec{\beta} = \beta_{1} = \frac{\log s}{2\pi i} \in i \mathbb{R}^{+}$ if $s > 0$. The large $n$ asymptotics of this ratio of Hankel determinants were studied by Its and Krasovsky in \cite{ItsKrasovsky} for a fixed $t_{1} \in (-1,1)$. Their asymptotic formula is valid not only for purely imaginary $\beta_{1}$, but as long as $\Re \beta_{1} \in \big( \frac{-1}{4},\frac{1}{4}\big)$. Note that the quantity in \eqref{gap prob m=1} is a polynomial in $s$, but admits an expression in terms of a Hankel determinant with a Fisher-Hartwig singularity only if $s \neq 0$. If $s = 0$, there is no thinning, and the probability $\mathbb{P}_{\mathrm{GUE}}\big(\sharp\{ y_{i} : y_{i} < t_{1} \} = 0 \big) = \mathbb{P}_{\mathrm{GUE}}\big(\sharp \{ x_{i}:x_{i}<t_{1} \}=0\big)$ can be expressed in terms of a Hankel determinant associated to the weight $\chi_{[t_{1},\infty)}(x)e^{-2nx^{2}}$. We do not study such Hankel determinants in the present paper. Their asymptotics as $n \to \infty$ deserve a separate analysis, see e.g. \cite{DS,CharlierDeano} (and see \cite{Widom,ChCl,ChCl2} for analogous situations for Toeplitz determinants).

\newpage
The operation of thinning is general and can be applied on any point process \cite{Daley}. It was introduced in random matrix theory by Bohigas and Pato \cite{BohigasPato1, BohigasPato2}. The point process \eqref{distribution eigenvalues} can be rewritten as a determinant, and is called determinantal (see e.g. \cite{Johansson, Soshnikov} for the definition and properties of determinantal point processes). If we apply an independent and constant thinning to a determinantal point process, the thinned point process is still determinantal \cite{Kallenberg}. This provides an efficient way to interpolate between the original point process (when $s = 0$) and an uncorrelated Poisson process (when $s \to 1$ at a certain speed). Several such transitions have been studied, see e.g. \cite{BotDeiItsKra1, BotDeiItsKra2, ChCl, ChCl2, BerDui, Lambert, BotBuc}.
Hankel determinants with several jump discontinuities are related to piecewise constant (and independent) thinning. We start from the point process of \eqref{distribution eigenvalues} for a general potential $V$, and consider $\mathcal{K} \subseteq \{1,...,m+1\}$, possibly empty. For $k \in \mathcal{K}$, each eigenvalue on $(t_{k-1},t_{k})$ is removed with a probability $s_{k} \in (0,1]$ (in this setting $t_{0} = -\infty$ and $t_{m+1} = +\infty$). The thinned point process is again denoted $y_{1},...,y_{N}$, but $N$ is no longer binomial. In the same spirit as \eqref{gap prob m=1}, one can show that
\begin{equation}\label{lol 19}
\mathbb{P}_{V}\Big(\sharp \{ y_{j} \hspace{-0.05cm} \in \hspace{-0.05cm} \bigcup_{k\in \mathcal{K}} (t_{k-1},t_{k}) \} = 0\Big) \hspace{-0.08cm} = \hspace{-0.08cm} \frac{1}{n!Z_{n}}\hspace{-0.08cm} \int_{\mathbb{R}^{n}} \hspace{-0.1cm} \prod_{1\leq j < k \leq n} \hspace{-0.3cm}(x_{k}-x_{j})^{2}  \prod_{j=1}^{n} e^{-nV(x_{j})} \hspace{-0.1cm} \prod_{k \in \mathcal{K}} \hspace{-0.1cm}\left\{ \hspace{-0.2cm} \begin{array}{l l}
s_{k}, & \hspace{-0.28cm} \mbox{ if } x_{j} \in (t_{k-1},t_{k}) \\
1, & \hspace{-0.28cm} \mbox{ if } x_{j} \notin (t_{k-1},t_{k}) \\
\end{array} \right. \hspace{-0.2cm} dx_{j}.
\end{equation}
The piecewise constant factor in \eqref{lol 19} can be rewritten in the form
\begin{equation}\label{lol 30}
\prod_{k\in \mathcal{K}}\left\{ \begin{array}{l l}
s_{k}, & \mbox{ if } x \in (t_{k-1},t_{k}) \\
1, & \mbox{ if } x \notin (t_{k-1},t_{k}) \\
\end{array} \right\} = \prod_{k \in \mathcal{K}} s_{k}^{1/2} \prod_{j=1}^{m} \omega_{\widetilde{\beta}_{j}}(x), \qquad 2i \pi \widetilde{\beta}_{j} = \log \left( \frac{\widetilde{s}_{j}}{\widetilde{s}_{j+1}} \right),
\end{equation}
where $\widetilde{s}_{j} = s_{j}$ if $j \in \mathcal{K}$, and $\widetilde{s}_{j} = 1$ if $j \notin \mathcal{K}$. Thus, the gap probability in the thinned point process is given by
\begin{equation}\label{ratio beta m geq 2}
\mathbb{P}_{V}\Big(\sharp \{ y_{j} \hspace{-0.05cm} \in \hspace{-0.05cm} \bigcup_{k\in \mathcal{K}} (t_{k-1},t_{k}) \} = 0\Big) =  \frac{D_{n}(\vec{0},\vec{\beta}_{\star},V,0)}{D_{n}(\vec{0},\vec{0},V,0)}\prod_{k \in \mathcal{K}} s_{k}^{1/2},
\end{equation}
and $\vec{\beta}_{\star} = (\widetilde{\beta}_{1},...,\widetilde{\beta}_{m})$ is defined as in \eqref{lol 30}.

\vspace{0.2cm}Following \cite{ChCl2}, instead of studying the thinned spectrum, we can consider a situation where we have information about the thinned spectrum (in the language of probability, we observe an event $A$), and from there, we try to obtain information about the initial/complete spectrum. The initial point process $x_{1},...,x_{n}$ is called \textit{conditional} to $A$. Let us assume that $A$ is the event of $\sharp \{ y_{j} \hspace{-0.05cm} \in \hspace{-0.05cm} \bigcup_{k\in \mathcal{K}} (t_{k-1},t_{k}) \} = 0$, i.e. we observe a gap in the thinned spectrum. Let us also assume that $n$ is known, as well as the thinning parameters $t_{1},...,t_{m}$,$\widetilde{s}_{1},...,\widetilde{s}_{m}$. From Bayes' formula for conditional probabilities, the conditional point process with respect to $A$ follows the distribution \vspace*{-0.08cm}
\begin{equation}
\frac{1}{n!\widetilde{Z}_{n}}\prod_{1\leq i < j \leq n} (x_{j}-x_{i})^{2}\prod_{j=1}^{n}\widetilde{w}(x_{j})dx_{j}, \qquad \widetilde{Z}_{n} = D_{n}(\vec{0},\vec{\beta}_{\star},V,0),
\end{equation}
where $\widetilde{w}(x) = e^{-nV(x)}\prod_{j=1}^{m}\omega_{\widetilde{\beta}_{j}}(x)$.
Thus, the generalized correlations of the characteristic polynomial of the conditional point process with respect to $A$ is expressed as a ratio of Hankel determinants with general FH singularities \vspace*{-0.08cm}
\begin{equation}\label{ratio full FH}
\mathbb{E}_{V,\mathrm{cond}}\Bigg( \prod_{j=1}^{n}e^{W(x_{j})}\prod_{k=1}^{m}|p(t_{k})|^{\alpha_{k}}e^{2 i \beta_{k} \arg p_{n}(t_{k})} \Bigg) =  \frac{D_{n}(\vec{\alpha},\vec{\beta} + \vec{\beta}_{\star},V,W)}{D_{n}(\vec{0},\vec{\beta}_{\star},V,0)}\prod_{k=1}^{m}e^{-i\pi \beta_{k}}.
\end{equation}
The contribution of this paper is to obtain large $n$ asymptotics for $D_{n}(\vec{\alpha},\vec{\beta},V,W)$ for one-cut regular potentials, up to the constant term. The case of general FH singularities (even for $m=1$) and of pure jump-type singularities (for $m \geq 2$) have not been obtained rigorously, even for the Gaussian potential $V(x) = 2x^{2}$. In particular, asymptotics for the partition functions $\widetilde{Z}_{n}$ and for the correlations \eqref{generalise correlation referee} and \eqref{ratio full FH} are not known.
\newpage
\begin{theorem}\label{Main theorem}
Let $m \in \mathbb{N}$, and let $t_{j}$, $\alpha_{j}$ and $\beta_{j}$ be such that $t_{j} \in (-1,1)$, $t_{j} \neq t_{k}$ for $1\leq j\neq k \leq m$, $\Re \alpha_{j} > -1$ and $\Re \beta_{j} \in (\frac{-1}{4},\frac{1}{4})$, for $j=1,...,m$. Let $V$ be a one-cut regular potential whose equilibrium measure is supported on $[-1,1]$ with density $\psi(x)\sqrt{1-x^{2}}$, and let $W: \mathbb{R}\to\mathbb{R}$ be analytic in a neighbourhood of $[-1,1]$, locally H\"{o}lder-continuous on $\mathbb{R}$ and such that $W(x) = \bigO(V(x)), \mbox{ as } |x| \to \infty$. As $n \to \infty$, we have
\begin{equation}\label{asymp thm}
D_{n}(\vec{\alpha},\vec{\beta},V,W) = \exp\left(C_{1} n^{2} + C_{2} n + C_{3} \log n + C_{4} + \bigO \Big( \frac{\log n}{n^{1-4\beta_{\max}}} \Big)\right),
\end{equation}
with $\beta_{\max} = \max \{ |\Re \beta_{1}|,...,|\Re \beta_{m}| \}$ and 
\begin{align}
& C_{1} = - \log 2 - \frac{3}{4} - \frac{1}{2} \int_{-1}^{1} \sqrt{1-x^{2}}(V(x)-2x^{2})\left( \frac{2}{\pi} + \psi(x) \right)dx, \\
& C_{2} = \log(2\pi) - \mathcal{A}\log 2 - \frac{\mathcal{A}}{2\pi}\int_{-1}^{1} \frac{V(x)-2x^{2}}{\sqrt{1-x^{2}}}dx + \int_{-1}^{1}\psi(x)\sqrt{1-x^{2}}W(x)dx  \\
& \hspace{1cm} + \sum_{j=1}^{m} \left( \frac{\alpha_{j}}{2}(V(t_{j})-1) + \pi i \beta_{j} \left( 1 - 2 \int_{t_{j}}^{1}\psi(x)\sqrt{1-x^{2}}dx \right) \right), \nonumber \\
& C_{3} = - \frac{1}{12} + \sum_{j=1}^{m} \bigg( \frac{\alpha_{j}^{2}}{4} - \beta_{j}^{2} \bigg), \\
& C_{4} = \zeta^{\prime}(-1)  + \frac{\mathcal{A}}{2\pi} \int_{-1}^{1} \frac{W(x)}{\sqrt{1-x^{2}}}dx - \frac{1}{4\pi^{2}}\int_{-1}^{1}  \frac{W(y)}{\sqrt{1-y^{2}}} \bigg(\dashint_{-1}^{1} \frac{W^{\prime}(x)\sqrt{1-x^{2}}}{x-y}dx \bigg) dy \\
& \hspace{1cm} - \frac{1}{24}\log \left( \frac{\pi^{2}}{4}\psi(1)\psi(-1) \right) + \sum_{1\leq j < k \leq m} \log  \Bigg(  \frac{\big(1-t_{j}t_{k}-\sqrt{(1-t_{j}^{2})(1-t_{k}^{2})}\big)^{2\beta_{j}\beta{k}}}{2^{\frac{\alpha_{j}\alpha_{k}}{2}}|t_{j}-t_{k}|^{\frac{\alpha_{j}\alpha_{k}}{2} + 2\beta_{j}\beta_{k}}} \Bigg) \nonumber \\
& \hspace{1cm} +\sum_{j=1}^{m} \Bigg(  i\mathcal{A}\beta_{j}\arcsin t_{j} - \frac{i\pi}{2}\beta_{j}\mathcal{A}_{j} + \log \frac{G(1+\frac{\alpha_{j}}{2}+\beta_{j})G(1+\frac{\alpha_{j}}{2}-\beta_{j})}{G(1+\alpha_{j})} \Bigg) \nonumber \\
&  \hspace{1cm}+\sum_{j=1}^{m} \Bigg(  \bigg(\frac{\alpha_{j}^{2}}{4}-\beta_{j}^{2}\bigg)\log\left(\frac{\pi}{2}\psi(t_{j})\right) - \frac{\alpha_{j}}{2}W(t_{j}) + i\frac{\beta_{j}}{\pi} \sqrt{1-t_{j}^{2}} \dashint_{-1}^{1} \frac{W(x)}{\sqrt{1-x^{2}}(t_{j}-x)}dx \Bigg) \nonumber  \\
& \hspace{1cm}+\sum_{j=1}^{m} \bigg( \frac{\alpha_{j}^{2}}{4}-3\beta_{j}^{2} \bigg)\log\left(2\sqrt{1-t_{j}^{2}}\right), \nonumber
\end{align}
where $G$ is Barnes' $G$-function, $\zeta$ is Riemann's zeta-function and where we use the notations
\begin{equation}
\mathcal{A} = \sum_{j=1}^{m} \alpha_{j}, \qquad \mathcal{A}_{j} = \sum_{l=1}^{j-1} \alpha_{l} - \sum_{l=j+1}^{m} \alpha_{l}.
\end{equation}
Furthermore, the error term in \eqref{asymp thm} is uniform for all $\alpha_{k}$ in compact subsets of \\$\{ z \in \mathbb{C}: \Re z >-1 \}$, for all $\beta_{k}$ in compact subsets of  $\{ z \in \mathbb{C}: \Re z \in \big( \frac{-1}{4},\frac{1}{4} \big) \}$, and uniform in $t_{1},...,t_{m}$, as long as there exists $\delta > 0$ independent of $n$ such that
\begin{equation}\label{def of delta}
\min_{j\neq k}\{ |t_{j}-t_{k}|,|t_{j}-1|,|t_{j}+1|\} \geq \delta.
\end{equation}
\end{theorem}
\begin{remark}
The notation $\dashint$ stands for the Cauchy principal value integral.
\end{remark}
\begin{remark}\label{remark: beta 1/4 1/2}
If some parameters among $\alpha_{1},...,\alpha_{m},\beta_{1},...,\beta_{m},t_{1},...,t_{m}$ tend to the boundaries of their domains at a sufficiently slow speed as $n \to \infty$, the asymptotic formula \eqref{asymp thm} is still expected to hold, but with a worse error term. If the speed is faster, critical transitions should take place and \eqref{asymp thm} is no longer valid. We don't analyse these cases here. Such transitions have been studied in \cite{ChCl,ChCl2} for Toeplitz determinants with a single jump-type singularity with $\beta_{1} \to \pm i \infty$, in \cite{ClaeysFahs} for Hankel determinants with two merging root-type singularities, in \cite{ClKr} for Toeplitz determinants with two merging FH singularities and in \cite{WuXuZhao} for Hankel determinants with a FH singularity approaching the edge (i.e. approaching $1$ or $-1$).

\vspace{0.2cm}\hspace{-0.55cm}The assumption $\Re \beta_{j} \in (\frac{-1}{4},\frac{1}{4})$, $j = 1,...,m$ makes it easier to handle some technicalities in our analysis. The methods presented here allow with extra efforts to treat the more general situation $\Re \beta_{j} \in (\frac{-1}{2},\frac{1}{2})$ and to compute more terms in the asymptotics. A similar technical problem was first observed in \cite{ItsKrasovsky} for large $n$ asymptotics of $D_{n}(\vec{0},\vec{\beta},2x^{2},0)$ with $m=1$ (i.e. $\vec{\beta}=\beta_{1}$). For the reader interested to extend Theorem \ref{Main theorem} to the most general situation $\Re \beta_{j} \in (\tfrac{-1}{2},\tfrac{1}{2}]$, we recommend to read \cite[Section 5.2]{DeiftItsKrasovsky} and \cite[Remark 1.22]{DIK}, where analogous situations have been studied for Toeplitz determinants with FH singularities and for Hankel determinants associated to a weight of the form \eqref{weight Deift Its Krasovsky}. 
\end{remark}
\begin{remark}\label{Remark: support}
The assumption that the support of the equilibrium measure is $[-1,1]$ is without loss of generality. Suppose that $\widetilde{V}$ is a one-cut regular potential whose equilibrium measure is supported on $[a,b]$. Then, the change of variables $y_{j} = \left( x_{j} - \frac{a+b}{2} \right)/\left( \frac{b-a}{2} \right)$, $j = 1,...,n$ in \eqref{Hankel determinants} shows that
\begin{equation}
D_{n}(\vec{\alpha},\vec{\beta},\widetilde{V},\widetilde{W};\vec{t}_{\star}) = \left( \frac{b-a}{2} \right)^{n^{2}+n\mathcal{A}}D_{n}(\vec{\alpha},\vec{\beta},V,W;\vec{t}),
\end{equation}
where we have explicitly written the dependence in $\vec{t}_{\star} = (\widetilde{t}_{1},...,\widetilde{t}_{m})$, where $\vec{t} = (t_{1},...,t_{m})$ is given by $t_{j} = \left( \widetilde{t}_{j} - \frac{a+b}{2} \right)/\left( \frac{b-a}{2} \right)$, $j = 1,...,m$, and where
\begin{equation*}
V(x) = \widetilde{V}\left( \frac{a+b}{2}+ x \frac{b-a}{2} \right), \quad W(x) = \widetilde{W}\left( \frac{a+b}{2}+ x \frac{b-a}{2} \right).
\end{equation*}
Also, if $\widetilde{\ell}$ and $\widetilde{\psi}(x) \sqrt{(b-x)(x-a)}$ are respectively the Euler-Lagrange constant and the density of the equilibrium measure of $\widetilde{V}$, the change of variables $u = \left( s - \frac{a+b}{2} \right)/\left( \frac{b-a}{2} \right)$ in \eqref{var equality} and \eqref{var inequality} implies that the equilibrium measure of $V$ is supported on $[-1,1]$ and can be written as $\psi(x)\sqrt{1-x^{2}}$. Furthermore, $\psi$ and the Euler-Lagrange constant $\ell$ of $V$ are given by
\begin{equation}
\ell = \widetilde{\ell} + 2\log\left(\frac{b-a}{2}\right), \qquad \psi(x) = \left( \frac{b-a}{2} \right)^{2} \widetilde{\psi}\left( \frac{a+b}{2} + x \frac{b-a}{2} \right).
\end{equation}
\end{remark}

\subsection{Outline}
We will compute the asymptotics for $D_{n}(\vec{\alpha},\vec{\beta},V,W)$ in three steps which can be schematized as
\begin{equation}
D_{n}(\vec{\alpha},\vec{0},2x^{2},0) \mapsto D_{n}(\vec{\alpha},\vec{\beta},2x^{2},0) \mapsto D_{n}(\vec{\alpha},\vec{\beta},V,0) \mapsto D_{n}(\vec{\alpha},\vec{\beta},V,W).
\end{equation}
Each of these steps is subdivided into three parts: 1) a differential identity for $\log D_{n}(\vec{\alpha},\vec{\beta},V,W)$, 2) an asymptotic analysis of a Riemann-Hilbert (RH) problem, and 3) the integration of the differential identity.

\vspace{0.2cm}\hspace{-0.52cm}As mentioned in the introduction, asymptotics for $D_{n}(\vec{\alpha},\vec{0},2x^{2},0)$ are already known (see \eqref{partition GUE} and \eqref{asymptotics alpha Krasovsky}). In the first step, for each $\beta_{k}$, $k \in \{1,...,m\}$, we will obtain an identity which expresses $\partial_{\beta_{k}} \log D_{n}(\vec{\alpha},\vec{\beta},2x^{2},0)$ in terms of orthogonal polynomials (OPs). These differential identities will be similar to the ones presented in $\cite{ItsKrasovsky, Krasovsky}$, but some extra care due to the presence of general FH singularities is needed. We integrate these differential identities in Section \ref{Section: beta}. 

\vspace{0.2cm}\hspace{-0.52cm}In the second step, we follow \cite{BerWebbWong} and introduce a deformation parameter $s\in [0,1]$ for the potential, 
\begin{equation}\label{deformation parameter potential}
V_{s}(x) = (1-s)2x^{2} + sV(x).
\end{equation}
This potential interpolates between $2x^{2}$ (for $s=0$) and $V(x)$ (for $s=1$). From \eqref{var equality} and \eqref{var inequality}, it is easily checked that $V_{s}$ is a one-cut regular potential for every $s \in [0,1]$, and the associated density $\psi_{s}(x)$ and Euler-Lagrange constant $\ell_{s}$ are simply given by
\begin{align}\label{d_s and ell_s}
& \psi_{s}(x) = (1-s)\frac{2}{\pi} + s \psi(x), \qquad \ell_{s} = (1-s)(1+2\log 2) + s \ell.
\end{align} 
We present an identity for $\partial_{s} \log D_{n}(\vec{\alpha},\vec{\beta},V_{s},0)$ in Section \ref{Section:diff identities}, and its integration in Section \ref{Section: integrating V}.

\vspace{0.2cm}\hspace{-0.52cm}In the last step, the potential $V$ is fixed and $W$ is analytic in a neighbourhood of $[-1,1]$ and locally H\"{o}lder continuous on $\mathbb{R}$. For $t \in [0,1]$, we define
\begin{equation}\label{deformation parameter Tau}
W_{t}(z) = \log\big(1-t+t e^{W(z)}\big),
\end{equation}
where the principal branch of the log is taken. For every $t \in [0,1]$, $W_{t}$ is analytic on a neighbourhood of $[-1,1]$ independent of $t$ and is still H\"{o}lder continuous on $\mathbb{R}$. This deformation is the same as the one used in \cite{DeiftItsKrasovsky, BerWebbWong}. Then, we proceed similarly as in the two previous steps, with an identity for $\partial_{t} \log D_{n}(\vec{\alpha},\vec{\beta},V,W_{t})$ in Section \ref{Section:diff identities}, and its integration in Section \ref{Section: Tau}.

\vspace{0.2cm}\hspace{-0.52cm}The differential identities for $\log D_{n}(\vec{\alpha},\vec{\beta},V,W)$ are expressed in terms of OPs, which are orthogonal with respect to a weight which depends on the step:
\begin{align*}
& w(x) = e^{-2nx^{2}}\omega(x), & & \mbox{ in step 1}, \\
& w_{s}(x) = e^{-nV_{s}(x)} \omega(x), & & \mbox{ in step 2}, \\
& w_{t}(x) = e^{-nV(x)}e^{W_{t}(x)} \omega(x), & & \mbox{ in step 3}.
\end{align*}
All these weights belong to the class of weights presented in \eqref{weight_general} (which we also denote by $w$). In Section \ref{Section:OP and Y}, we introduce  the family of OPs with respect to $w$ and the associated RH problem found by Fokas, Its and Kitaev \cite{FokasItsKitaev}. We perform an asymptotic analysis of this RH problem using the Deift/Zhou \cite{DeiftZhou1992, DeiftZhou, DKMVZ2, DKMVZ1} steepest descent method in Section \ref{Section: steepest descent}.

\section{\hspace{-0.34cm} Orthogonal polynomials and a Riemann-Hilbert problem}\label{Section:OP and Y}

The orthonormal polynomials $p_{k}(x) = p_{k}^{(n)}(x;\vec{\alpha},\vec{\beta},V,W)$ of degree $k$ associated to the weight $w$ given in \eqref{weight_general} are defined through the following orthogonality conditions:
\begin{equation}\label{ortho conditions}
\int_{\mathbb{R}} p_{k}(x) p_{j}(x) w(x) dx = \delta_{jk}, \qquad j = 0,1,...,k.
\end{equation}
We consider $D_{k}^{(n)}(\vec{\alpha},\vec{\beta},V,W) = \det(w_{j+l-2})_{j,l=1,...,k}$, where $w_{j} = w_{j}^{(n)}(\vec{\alpha},\vec{\beta},V,W) = \int_{-\infty}^{\infty}x^{j}w(x)dx$ is the $j$-th moment of $w$. When there is no confusion, we will simply write $D_{k}^{(n)}$ instead of $D_{k}^{(n)}(\vec{\alpha},\vec{\beta},V,W)$, and we set $D_{0}^{(n)} = 1$. With this notation, we have $D_{n}^{(n)} = D_{n}$, where $D_{n}$ is defined in \eqref{Hankel determinants 1}. The existence of the system of orthogonal polynomials depends if one (or several) of these determinants vanishes. It is well-known (see e.g. \cite{Szego OP, Deift}) that if
\begin{equation}\label{condition determinant non zero}
D_{k}^{(n)}(\vec{\alpha},\vec{\beta},V,W) \neq 0 \quad \mbox{ and } \quad D_{k+1}^{(n)}(\vec{\alpha},\vec{\beta},V,W) \neq 0, 
\end{equation}
then the orthogonal polynomial $p_{k}$ exists and is given by
\begin{equation}\label{lol 33}
p_{k}(x) = \frac{\left| \begin{array}{c c c c}
w_{0} & \cdots & w_{k-1} & w_{k} \\
\vdots & \ddots & \vdots & \vdots \\
w_{k-1} & \cdots & w_{2k-2} & w_{2k-1} \\
1 & \cdots & x^{k-1} & x^{k} \\
\end{array} \right|}{\sqrt{D_{k+1}^{(n)}(\vec{\alpha},\vec{\beta},V,W)}\sqrt{D_{k}^{(n)}(\vec{\alpha},\vec{\beta},V,W)}} = \kappa_{k} x^{k} + ...,
\end{equation}
where $\kappa_{k} = \kappa_{k}^{(n)}(\vec{\alpha},\vec{\beta},V,W)$ is the leading coefficient of $p_{k}$, and is given by
\begin{equation}\label{lol 32}
\kappa_{k} = \frac{\sqrt{D_{k}^{(n)}(\vec{\alpha},\vec{\beta},V,W)}}{\sqrt{D_{k+1}^{(n)}(\vec{\alpha},\vec{\beta},V,W)}} \neq 0, \qquad \mbox{ and where } \quad \arg \sqrt{D_{k}^{(n)}},\arg \sqrt{D_{k+1}^{(n)}} \in \left(-\frac{\pi}{2},\frac{\pi}{2}\right].
\end{equation}
Note that \eqref{ortho conditions} alone does not uniquely define the orthonormal polynomials $p_{k}$ and the leading coefficients $\kappa_{k}$ (they are unique up to multiplicative factors of $-1$). We fixed this with the (arbitrary) choice of the arguments in \eqref{lol 32}. However, in the present paper we will only use $\kappa_{k}^{2}$, $\kappa_{k}p_{k}$, as well as the monic orthogonal polynomials $\kappa_{k}^{-1}p_{k}$, which are always uniquely defined by \eqref{ortho conditions} (if they exist), i.e. they are independent of the choice of the arguments in \eqref{lol 32}. We will also use the subleading coefficients of the monic orthogonal polynomials:
\begin{equation}\label{def coeff OPs}
\kappa_{k}^{-1}p_{k}(x) = x^{k} + \eta_{k}x^{k-1} + \gamma_{k} x^{k-2} + ..., \qquad k \in \mathbb{N}, \mbox{ } k \geq 2.
\end{equation}
If $D_{k}^{(n)} \neq 0$ for every $k = 0,1,...,n$, \eqref{lol 32} implies the following well-known formula
\begin{equation}\label{lol 34}
D_{n}(\vec{\alpha},\vec{\beta},V,W) = \prod_{j=0}^{n-1} \kappa_{k}^{-2}.
\end{equation}
Note that if $\alpha_{k} \in (-1,\infty)$ and $\beta_{k} \in i\mathbb{R}$ for all $k =1,...,m$, then $w$ is positive and it follows from Heine's integral representation for Hankel determinants that $D_{k}^{(n)}(\vec{\alpha},\vec{\beta},V,W) > 0$ for all $k$. As a consequence, the sequence of orthogonal polynomials $(p_{k})_{k\in \mathbb{N}}$ exists. Nevertheless, for general values of the parameters $\vec{\alpha}$ and $\vec{\beta}$, the weight $w$ is complex and there is no guarantee of existence for the orthogonal polynomials (one of several determinants $D_{k}^{(n)}$ may vanish). This will cause some extra difficulties in the analysis.  We discuss this in more detail at the beginning of Section \ref{Section:diff identities}. 

\vspace{0.2cm}\hspace{-0.525cm}For sufficiently large $n$, we will prove existence and obtain explicit asymptotics of $p_{n}$ via the Deift-Zhou steepest descent method applied on a Riemann-Hilbert problem. Consider the matrix valued function $Y(z) = Y^{(n)}(z;\vec{\alpha},\vec{\beta},V,W)$, defined by
\begin{equation}
Y(z) = \begin{pmatrix}\label{Y definition}
\kappa_{n}^{-1}p_{n}(z) & \frac{\kappa_{n}^{-1}}{2\pi i} \int_{\mathbb{R}} \frac{p_{n}(x)w(x)}{x-z}dx \\
-2\pi i \kappa_{n-1} p_{n-1}(z) & -\kappa_{n-1} \int_{\mathbb{R}} \frac{p_{n-1}(x)w(x)}{x-z}dx
\end{pmatrix}.
\end{equation}
Note that $Y$ is expressed only in terms of $\kappa_{n}^{-1}p_{n}$, $\kappa_{n-1}p_{n-1}$ and their Cauchy transforms. It is known \cite{FokasItsKitaev} that $Y$ can be characterized as the solution to a boundary value problem for analytic functions, called the RH problem for $Y$. Furthermore, $Y$ defined by \eqref{Y definition} (as well as the solution to the RH problem for $Y$), exists if and only if $D_{n-1}^{(n)} \neq 0$, $D_{n}^{(n)} \neq 0$ and $D_{n+1}^{(n)} \neq 0$, and is always unique.
\subsubsection*{RH problem for $Y$}
\begin{itemize}
\item[(a)] $Y : \mathbb{C}\setminus \mathbb{R} \to \mathbb{C}^{2\times 2}$ is analytic.
\item[(b)] The limits of $Y(z)$ as $z$ approaches $\mathbb{R}\setminus \{t_{1},...,t_{m}\}$ from above and below exist, are continuous on $\mathbb{R}\setminus \{t_{1},...,t_{m}\}$ and are denoted by $Y_+$ and $Y_-$ respectively. Furthermore they are related by
\begin{equation}\label{jump relations of Y}
Y_{+}(x) = Y_{-}(x) \begin{pmatrix}
1 & w(x) \\ 0 & 1
\end{pmatrix}, \hspace{0.5cm} \mbox{ for } x \in \mathbb{R}\setminus \left\{t_{1},...,t_{m} \right\}.
\end{equation}
\item[(c)] As $z \to \infty$, we have $Y(z) = \left(I + \bigO(z^{-1})\right) z^{n\sigma_{3}}$, where $\sigma_{3} = \begin{pmatrix}
1 & 0 \\ 0 & -1
\end{pmatrix}$.
\item[(d)] As $z$ tends to $t_{k}$, $k \in \{1,...,m\}$, the behaviour of $Y$ is
\begin{equation}
\begin{array}{l l}
\displaystyle Y(z) = \begin{pmatrix}
\bigO(1) & \bigO(\log (z-t_{k})) \\
\bigO(1) & \bigO(\log (z-t_{k}))
\end{pmatrix}, & \displaystyle \mbox{ if } \Re \alpha_{k} = 0, \\[0.35cm]
\displaystyle Y(z) = \begin{pmatrix}
\bigO(1) & \bigO(1)+\bigO((z-t_{k})^{\alpha_{k}}) \\
\bigO(1) & \bigO(1)+\bigO((z-t_{k})^{\alpha_{k}})
\end{pmatrix}, & \displaystyle \mbox{ if } \Re \alpha_{k} \neq 0.
\end{array}
\end{equation}
\end{itemize}
The above condition (b) follows from the Sokhotski formula and relies on the assumption that $W$ is locally H\"{o}lder continuous on $\mathbb{R}$ (see e.g. \cite{Gakhov}).

\section{Differential identities}\label{Section:diff identities}
In this section, we derive several identities which will be useful to prove Theorem \ref{Main theorem}. These identities will be valid only when all orthogonal polynomials $p_{k}$ for $k = 0,...,n$ exist. We give here an overview of how we will handle these technicalities in this paper, following \cite{Krasovsky, ItsKrasovsky, DIK}. By applying Lebesgue's dominated convergence theorem on the moments $w_{j}$, we can prove that
\begin{align*}
& \mbox{(a) the determinants } D_{k}^{(n)}(\vec{\alpha},\vec{\beta},2x^{2},0) \mbox{ are analytic functions of } (\vec{\alpha},\vec{\beta}) \in \mathcal{P}_{\alpha} \times \mathbb{C}^{m}, \mbox{ where } \\
& \hspace{0.6cm}\mathcal{P}_{\alpha} = \{ \vec{\alpha}\in \mathbb{C}^{m} :  \Re \alpha_{j} > -1, \mbox{ for all } j = 1,...,m \}, \\
& \mbox{(b) the determinants } D_{k}^{(n)}(\vec{\alpha},\vec{\beta},V_{s},0) \mbox{ are analytic functions of } (\vec{\alpha},\vec{\beta},s) \in \mathcal{P}_{\alpha} \times \mathbb{C}^{m} \times (0,1), \\
& \mbox{(c) the determinants } D_{k}^{(n)}(\vec{\alpha},\vec{\beta},V,W_{t}) \mbox{ are analytic functions of } (\vec{\alpha},\vec{\beta},t) \in \mathcal{P}_{\alpha} \times \mathbb{C}^{m} \times [0,1].
\end{align*}
The property (b) deserves further explanation. If $V(x)/x^{2} \to + \infty$ as $|x| \to +\infty$, the moments are not well-defined if $\Re s <0$, since $V_{s}(x)\to -\infty$ as $|x| \to +\infty$. In this case, the determinants $D_{k}^{(n)}(\vec{\alpha},\vec{\beta},V_{s},0)$ are not analytic as functions of $s$ in a neighbourhood of $s=0$. Similarly, if $V(x)=o(x^{2})$ as $|x| \to +\infty$, the determinants $D_{k}^{(n)}(\vec{\alpha},\vec{\beta},V_{s},0)$ are not analytic as functions of $s$ in a neighbourhood of $s=1$. Note also from (a), (b), (c) above and from \eqref{Y definition}, \eqref{lol 33} and \eqref{lol 32}, that $Y^{(n)}(z;\vec{\alpha},\vec{\beta},2x^{2},0)$, $Y^{(n)}(z;\vec{\alpha},\vec{\beta},V_{s},0)$ and $Y^{(n)}(z;\vec{\alpha},\vec{\beta},V,W_{t})$ are meromorphic functions of $(\vec{\alpha},\vec{\beta}) \in \mathcal{P}_{\alpha} \times \mathbb{C}^{m}$, $(\vec{\alpha},\vec{\beta},s) \in \mathcal{P}_{\alpha} \times \mathbb{C}^{m} \times (0,1)$ and $(\vec{\alpha},\vec{\beta},t) \in \mathcal{P}_{\alpha} \times \mathbb{C}^{m} \times [0,1]$ respectively.

\vspace{0.2cm}\hspace{-0.55cm}We will find identities of the forms (see equations \eqref{differential identity}, \eqref{diff identity s} and \eqref{diff identity t} below)
\begin{align}
& \partial_{\beta_{k}}\log D_{n}(\vec{\alpha},\vec{\beta},2x^{2},0) = F_{1,n}(\vec{\alpha},\vec{\beta}), \label{F1} \\
& \partial_{s} \log D_{n}(\vec{\alpha},\vec{\beta},V_{s},0) = F_{2,n}(\vec{\alpha},\vec{\beta},s), \label{F2}\\
& \partial_{t} \log D_{n}(\vec{\alpha},\vec{\beta},V,W_{t}) = F_{3,n}(\vec{\alpha},\vec{\beta},t), \label{F3}
\end{align}
where the functions $F_{1,n}$, $F_{2,n}$ and $F_{3,n}$ (we omit the dependence of $F_{1,n}$ in $k$) are expressed only in terms of $Y$. The identities \eqref{F1}, \eqref{F2} and \eqref{F3} are valid only for $(\vec{\alpha},\vec{\beta})\in \widetilde{\mathcal{P}}_{1}^{(n)}$, for $(\vec{\alpha},\vec{\beta},s)\in \widetilde{\mathcal{P}}_{2}^{(n)}$ and for $(\vec{\alpha},\vec{\beta},t)\in \widetilde{\mathcal{P}}_{3}^{(n)}$ respectively, where
\begin{align*}
& \widetilde{\mathcal{P}}_{1}^{(n)} = \left\{ (\vec{\alpha},\vec{\beta})\in \mathcal{P}_{\alpha}\times \mathbb{C}^{m}: D_{j}^{(n)}(\vec{\alpha},\vec{\beta},2x^{2},0) \neq 0 \mbox{ for all } j = 1,...,n+1 \right\}, \\
& \widetilde{\mathcal{P}}_{2}^{(n)} = \left\{ (\vec{\alpha},\vec{\beta},s)\in \mathcal{P}_{\alpha}\times \mathbb{C}^{m} \times (0,1): D_{j}^{(n)}(\vec{\alpha},\vec{\beta},V_{s},0) \neq 0 \mbox{ for all } j = 1,...,n+1 \right\}, \\
& \widetilde{\mathcal{P}}_{3}^{(n)} = \left\{ (\vec{\alpha},\vec{\beta},t)\in \mathcal{P}_{\alpha}\times \mathbb{C}^{m} \times [0,1]: D_{j}^{(n)}(\vec{\alpha},\vec{\beta},V,W_{t}) \neq 0 \mbox{ for all } j = 1,...,n+1 \right\}.
\end{align*}
In Section \ref{Section: steepest descent}, we will perform an asymptotic analysis on the RH problem for $Y$ as $n \to +\infty$. Let $\Omega$ be a compact subset of $\mathcal{P}_{\alpha}\times \mathcal{P}_{\beta}^{(\frac{1}{4})}$, where
\begin{equation*}
\mathcal{P}_{\beta}^{(\frac{1}{4})} = \{ \vec{\beta}\in \mathbb{C}^{m} :  \Re \beta_{j} \in (-\tfrac{1}{4},\tfrac{1}{4}), \mbox{ for all } j = 1,...,m \}.
\end{equation*}
The analysis of Section \ref{Section: steepest descent} will imply the existence of $F_{1,n}(\vec{\alpha},\vec{\beta})$, $F_{2,n}(\vec{\alpha},\vec{\beta},s)$ and $F_{3,n}(\vec{\alpha},\vec{\beta},t)$ for \textit{all} $n \geq n_{\star}$, where $n_{\star}$ depends only on $\Omega$, for \textit{all} $(\vec{\alpha},\vec{\beta}) \in \Omega$ and for \textit{all} $s,t \in [0,1]$. Furthermore, their large $n$ asymptotics will be explicitly computed up to the constant term in Section \ref{Section: beta}, \ref{Section: integrating V} and \ref{Section: Tau}, and will be valid uniformly for \textit{all} $(\vec{\alpha},\vec{\beta}) \in \Omega$ and for \textit{all} $s,t \in [0,1]$. 

\vspace{0.2cm}\hspace{-0.54cm}Thus, we will prove that the r.h.s. of \eqref{F1} exists for \textit{all} $n \geq n_{\star}$ and for \textit{all} $(\vec{\alpha},\vec{\beta}) \in \Omega$, but the identity \eqref{F1} itself is valid only for 
\begin{equation}
(\vec{\alpha},\vec{\beta}) \in \Omega \cap \widetilde{\mathcal{P}}_{1}^{(n)} = \Omega \setminus \widetilde{\Omega}_{1}^{(n)},
\end{equation}
where $\widetilde{\Omega}_{1}^{(n)}$ consists of at most a finite number of points (see (a) above). The analyticity in $\beta_{k}$ of the determinant $D_{n}(\vec{\alpha},\vec{\beta},2x^{2},0)$ will allow us to extend this differential identity from $\Omega \setminus \widetilde{\Omega}_{1}^{(n)}$ to $\Omega$. We will show that explicitly when integrating the differential identity in Section \ref{Section: beta}. 

\vspace{0.2cm}\hspace{-0.54cm}Likewise, the subset of $\Omega\times (0,1)$ for which the differential identity \eqref{F2} is valid is 
\begin{equation}
\big(\Omega\times (0,1)\big) \cap \widetilde{\mathcal{P}}_{2}^{(n)} = \big(\Omega\times (0,1)\big) \setminus \widetilde{\Omega}_{2}^{(n)}.
\end{equation}
However, since the determinants $D_{k}^{(n)}(\vec{\alpha},\vec{\beta},V_{s},0)$ are analytic as functions of $s \in (0,1)$, but not analytic at $s = 0$ and $s = 1$ (see (b) above), $\widetilde{\Omega}_{2}^{(n)}$ is locally finite on its interior but may have accumulation points on its boundary. More precisely, if $\epsilon \in (0,\tfrac{1}{2})$, then $\widetilde{\Omega}_{2,\epsilon}^{(n)} = \{ (\vec{\alpha},\vec{\beta},s)\in \widetilde{\Omega}_{2}^{(n)}:s \in (\epsilon,1-\epsilon) \}$ consists of at most a finite number of points. This extra difficulty does not appear in \cite{Krasovsky, ItsKrasovsky, DIK}, but does not change drastically the idea. We provide the details in Section \ref{Section: integrating V}.

\vspace{0.2cm}\hspace{-0.54cm}Since $D_{k}^{(n)}(\vec{\alpha},\vec{\beta},V,W_{t})$ are analytic as functions of $t\in [0,1]$ (see (c) above), the subset of $\Omega\times[0,1]$ for which the identity \eqref{F3} is valid is $(\Omega \times [0,1]) \cap \widetilde{\mathcal{P}}_{3}^{(n)} = (\Omega \times [0,1]) \setminus \widetilde{\Omega}_{3}^{(n)}$, where $\widetilde{\Omega}_{3}^{(n)}$ consists of at most a finite number of points. We can extend this identity for \textit{all} $(\vec{\alpha},\vec{\beta},t) \in \Omega \times [0,1]$ in a similar way as for the identity \eqref{F1}.

\vspace{0.2cm}\hspace{-0.54cm}Once the differential identities \eqref{F1}, \eqref{F2} and \eqref{F3} are extended from $\Omega \setminus \widetilde{\Omega}_{1}^{(n)}$ to $\Omega$, from $\big(\Omega\times (0,1)\big) \setminus \widetilde{\Omega}_{2}^{(n)}$ to $\Omega\times(0,1)$ and from $(\Omega \times [0,1]) \setminus \widetilde{\Omega}_{3}^{(n)}$ to $\Omega\times[0,1]$ respectively, we prove Theorem \ref{Main theorem} by integrating them.

\subsection{Differential identity with respect to $\beta_{k}$, $k \in \{1,...,m\}$}\label{subsection: diff identity beta}

As mentioned in the outline, we will first compute the asymptotics of $D_{n}(\vec{\alpha},\vec{\beta},2x^{2},0)$, i.e. when the weight has the form $w(x) = e^{-2nx^{2}}\omega(x)$. For this we will need a differential identity with respect to $\beta_{k}$, for each $k \in \{1,...,m\}$. Suppose $f$ is a smooth and integrable function on $\mathbb{R}$ with sufficient decay at $\pm\infty$. The integral of $\frac{f(x) \omega(x)}{x-t_{k}}$ on $\mathbb{R}$ is not well-defined if $\Re \alpha_{k} \leq 0$, even in the sense of principal value, because of the jump caused by $\beta_{k}$. Therefore, we define a regularized integral by
\begin{equation}\label{def VP}
\mathrm{Reg}_{k}\left( f \right) = \lim_{\epsilon \to 0_{+}} \left[ \alpha_{k} \int_{\mathbb{R}\setminus[t_{k}-\epsilon,t_{k}+\epsilon]} \frac{f(x)\omega(x)}{x-t_{k}} dx - f(t_{k})\omega_{k}(t_{k})(e^{i\pi\beta_{k}}-e^{-i\pi\beta_{k}})\epsilon^{\alpha_{k}} \right],
\end{equation}
where
\begin{equation}\label{def omega_k}
\omega_{k}(z) = \prod_{j\neq k} \omega_{\alpha_{j}}(z)\omega_{\beta_{j}}(z).
\end{equation}
It is easy to observe that the limit in \eqref{def VP} converges. We have the following property.
\begin{proposition}\label{Prop: regularised integral} (adapted from \cite{Krasovsky}).
The regularised integral \eqref{def VP} satisfies
\begin{equation}\label{lim VP}
\mathrm{Reg}_{k}(f) = \lim_{z\to t_{k}}\left( \alpha_{k} \int_{\mathbb{R}} \frac{f(x) \omega(x)}{x-z}dx - \mathcal{J}_{k}(z)\right),
\end{equation}
where the limit is taken along a non-tangential to the real line path in $\{ z \in \mathbb{C}:\Im z > 0 \}$, and
\begin{equation}\label{substraction in VP}
\mathcal{J}_{k}(z) = \left\{ \begin{array}{l l}
\displaystyle \frac{\pi \alpha_{k}}{\sin(\pi\alpha_{k})} f(t_{k})\omega_{k}(t_{k}) (e^{i\pi\beta_{k}}-e^{-i\pi  \alpha_{k}}e^{-i\pi\beta_{k} })(z-t_{k})^{\alpha_{k}}, & \displaystyle \mbox{if } \Re \alpha_{k} \leq 0, \alpha_{k} \neq 0, \\[0.3cm]
\displaystyle f(t_{k})\omega_{k}(t_{k})(e^{i\pi\beta_{k}}-e^{-i\pi\beta_{k}}), & \displaystyle \mbox{if } \alpha_{k} = 0,  \\[0.3cm]
\displaystyle 0, & \displaystyle \mbox{if } \Re \alpha_{k} > 0.
\end{array} \right.
\end{equation}
\end{proposition}
\begin{proof}
In \cite[equation (29) and below]{Krasovsky}, the author proved that if $\alpha_{k} \neq 0$, we have as $\epsilon \to 0_{+}$ that
\begin{align}
& \int_{-\infty}^{t_{k}-\epsilon} \frac{f(x)\omega_{k}(x)|x-t_{k}|^{\alpha_{k}}}{x-t_{k}}dx - \int_{-\infty}^{t_{k}-\epsilon} \frac{f(x)\omega_{k}(x)|x-t_{k}|^{\alpha_{k}}}{x-z}dx = \label{Kras1} \\ 
& \hspace{2.5cm} \frac{-\pi f(z)\omega_{k}(z)}{\sin(\pi \alpha_{k})}(z-t_{k})^{\alpha_{k}} + \frac{\epsilon^{\alpha_{k}}}{\alpha_{k}}f(t_{k})\omega_{k}(t_{k})+ \frac{G_{1}(z)(z-t_{k}) + \bigO(\epsilon^{\Re \alpha_{k}+1})}{2i\sin(\pi\alpha_{k})},\nonumber \\[0.35cm] 
& \int_{t_{k}+\epsilon}^{\infty} \frac{f(x)\omega_{k}(x)|x-t_{k}|^{\alpha_{k}}}{x-t_{k}}dx - \int_{t_{k}+\epsilon}^{\infty} \frac{f(x)\omega_{k}(x)|x-t_{k}|^{\alpha_{k}}}{x-z}dx = \label{Kras2} \\ 
& \hspace{2.5cm} \frac{\pi e^{-i\pi  \alpha_{k}} f(z)\omega_{k}(z)}{\sin(\pi \alpha_{k})}(z-t_{k})^{\alpha_{k}} - \frac{\epsilon^{\alpha_{k}}}{\alpha_{k}}f(t_{k})\omega_{k}(t_{k})- \frac{G_{2}(z)(z-t_{k}) + \bigO(\epsilon^{\Re \alpha_{k}+1})}{2i\sin(\pi\alpha_{k})}, \nonumber
\end{align} 
where $G_{1}$ and $G_{2}$ are analytic in a neighbourhood of $t_{k}$. Multiplying \eqref{Kras1} by $\alpha_{k} e^{i\pi\beta_{k} }$ and \eqref{Kras2} by $\alpha_{k} e^{-i\pi\beta_{k} }$, and summing them we obtain as $\epsilon \to 0_{+}$
\begin{multline*}
\alpha_{k} \int_{\mathbb{R}\setminus[t_{k}-\epsilon,t_{k}+\epsilon]} \frac{f(x)\omega(x)}{x-t_{k}}dx - f(t_{k})\omega_{k}(t_{k}) (e^{i\pi\beta_{k}}-e^{-i\pi\beta_{k}})\epsilon^{\alpha_{k}} = \alpha_{k} \int_{\mathbb{R}} \frac{f(x)\omega(x)}{x-z}dx  \\ -\frac{\pi\alpha_{k} f(z)\omega_{k}(z)}{\sin(\pi\alpha_{k})}(z-t_{k})^{\alpha_{k}}(e^{i\pi\beta_{k}} -e^{- i \pi\alpha_{k}}e^{-i\pi\beta_{k} }) + \frac{(G_{1}(z)e^{i\pi\beta_{k}}-G_{2}(z)e^{-i\pi\beta_{k}})(z-t_{k}) + \bigO(\epsilon^{\Re\alpha_{k}+1})}{2i\sin \pi\alpha_{k}}\alpha_{k}.
\end{multline*}
Taking first the limit $\epsilon \to 0_{+}$ and then $z \to t_{k}$, we find the claim for $\alpha_{k} \neq 0$. The claim for $\alpha_{k} = 0$ is straightforward.
\end{proof}
The following developments are similar to those done in \cite{ItsKrasovsky, Krasovsky}, but the presence of both root-type and jump-type singularities requires some changes, and we provide the details below. In this subsection we assume that $(\vec{\alpha},\vec{\beta}) \in \widetilde{\mathcal{P}}_{1}^{(n)}$, such that the formula \eqref{lol 34} holds. We start with a general differential identity obtained in \cite{Krasovsky}, which we apply to the parameter $\beta_{k}$, for a $k \in \{1,...,m\}$:
\begin{equation}\label{lol 20}
\partial_{\beta_{k}} \log D_{n} = - n \partial_{\beta_{k}}\log \kappa_{n-1} + \frac{\kappa_{n-1}}{\kappa_{n}}(I_{1}-I_{2}),
\end{equation}
where
\begin{equation}
I_{1} = \int_{-\infty}^{\infty} p_{n-1}^{\prime}(x) \partial_{\beta_{k}}p_{n}(x) w(x)dx, \qquad I_{2} = \int_{-\infty}^{\infty} p_{n}^{\prime}(x) \partial_{\beta_{k}}p_{n-1}(x) w(x)dx.
\end{equation}
We will simplify $I_{1}$ and $I_{2}$ for our particular weight $w(x) = e^{-2nx^{2}}\omega(x)$. Let $\epsilon >0$ be sufficiently small such that $\min_{i\neq j} |t_{i}-t_{j}| > 3\epsilon$, $I_{1}$ can be split as
\begin{equation}
I_{1} = \sum_{l=1}^{m+1} \int_{t_{l-1}+\epsilon}^{t_{l}-\epsilon} p_{n-1}^{\prime}(x) \partial_{\beta_{k}}p_{n}(x) w(x)dx + \sum_{l=1}^{m} \int_{t_{l}-\epsilon}^{t_{l}+\epsilon} p_{n-1}^{\prime}(x) \partial_{\beta_{k}}p_{n}(x) w(x)dx,
\end{equation}
where $t_{0} = -\infty$ and $t_{m+1} = \infty$. Noting that
\begin{equation}
w^{\prime}(x) = \bigg( -4nx + \sum_{j=1}^{m} \frac{\alpha_{j}}{x-t_{j}} \bigg) w(x), \qquad x \in  \mathbb{R}\setminus\{ t_{1},...,t_{m} \},
\end{equation}
we have, by integration by parts, for $l \in \{1,...,m+1\}$
\begin{multline}\label{one integral inside J_1}
\int_{t_{l-1}+\epsilon}^{t_{l}-\epsilon} p_{n-1}^{\prime}(x) \partial_{\beta_{k}}p_{n}(x) w(x)dx = \\ p_{n-1}(t_{l}-\epsilon)\partial_{\beta_{k}}p_{n}(t_{l}-\epsilon)w(t_{l}-\epsilon)-p_{n-1}(t_{l-1}+\epsilon)\partial_{\beta_{k}}p_{n}(t_{l-1}+\epsilon)w(t_{l-1}+\epsilon) \\ - \int_{t_{l-1}+\epsilon}^{t_{l}-\epsilon} p_{n-1}(x) \bigg[ \partial_{\beta_{k}}p_{n}^{\prime}(x) + \partial_{\beta_{k}}p_{n}(x) \bigg( -4nx + \sum_{j=1}^{m} \frac{\alpha_{j}}{x-t_{j}} \bigg) \bigg] w(x)dx.
\end{multline}
In the above expression, the quantities $p_{n-1}(t_{m+1}-\epsilon)\partial_{\beta_{k}}p_{n}(t_{m+1}-\epsilon)w(t_{m+1}-\epsilon)$ and $p_{n-1}(t_{0}+\epsilon)\partial_{\beta_{k}}p_{n}(t_{0}+\epsilon)w(t_{0}+\epsilon)$ have to be understood as equal to $0$. On the other hand, for $l \in \{1,...,m\}$,
\begin{equation}
\int_{t_{l}-\epsilon}^{t_{l}+\epsilon} p_{n-1}^{\prime}(x) \partial_{\beta_{k}}p_{n}(x) w(x)dx = \bigO\big(\epsilon^{1+\Re \alpha_{l}}\big), \qquad \mbox{ as } \epsilon \to 0.
\end{equation}
Therefore, summing \eqref{one integral inside J_1} for $l = 1,...,m+1$, and taking the limit $ \epsilon \to 0$, we obtain
\begin{align}
& I_{1} = - \int_{-\infty}^{\infty} p_{n-1}(x) \partial_{\beta_{k}} p_{n}^{\prime}(x)w(x)dx + 4n \int_{-\infty}^{\infty} p_{n-1}(x)\partial_{\beta_{k}}p_{n}(x)xw(x)dx \\ & - \sum_{j=1}^{m} \alpha_{j} \int_{-\infty}^{\infty} p_{n-1}(x) \frac{\partial_{\beta_{k}}p_{n}(x)-\partial_{\beta_{k}}p_{n}(t_{j})}{x-t_{j}}w(x)dx - \sum_{j=1}^{m} \partial_{\beta_{k}}p_{n}(t_{j})\mathrm{Reg}_{j} \big(p_{n-1}(x)e^{-2nx^{2}}\big). \nonumber
\end{align}
By orthogonality,
\begin{equation*}
\int_{-\infty}^{\infty} p_{n-1}(x)\partial_{\beta_{k}}p_{n}^{\prime}(x)w(x)dx = n \frac{\partial_{\beta_{k}} \kappa_{n}}{\kappa_{n-1}}, \quad \int_{-\infty}^{\infty} p_{n-1}(x) \frac{\partial_{\beta_{k}}p_{n}(x)-\partial_{\beta_{k}}p_{n}(t_{j})}{x-t_{j}}w(x)dx = \frac{\partial_{\beta_{k}}\kappa_{n}}{\kappa_{n-1}}.
\end{equation*}
The next integral is a bit more involved and was done in \cite[equation (21)]{Krasovsky}, 
\begin{equation}
\int_{-\infty}^{\infty} p_{n-1}(x) \partial_{\beta_{k}}p_{n}(x) x w(x)dx = \frac{\kappa_{n}}{\kappa_{n-1}} \left[ \frac{\partial_{\beta_{k}}\kappa_{n}}{\kappa_{n}}\left( \frac{\kappa_{n-1}}{\kappa_{n}} \right)^{2} + \partial_{\beta_{k}}\gamma_{n} - \eta_{n} \partial_{\beta_{k}}\eta_{n} \right]. 
\end{equation}
Therefore, one can rewrite $I_{1}$ as
\begin{multline}\label{J_1}
I_{1} = - (n + \mathcal{A}) \frac{\partial_{\beta_{k}}\kappa_{n}}{\kappa_{n-1}} + 4n \frac{\kappa_{n}}{\kappa_{n-1}} \left[ \frac{\partial_{\beta_{k}}\kappa_{n}}{\kappa_{n}}\left( \frac{\kappa_{n-1}}{\kappa_{n}} \right)^{2} + \partial_{\beta_{k}}\gamma_{n} - \eta_{n} \partial_{\beta_{k}}\eta_{n} \right] \\ - \sum_{j=1}^{m} \partial_{\beta_{k}}p_{n}(t_{j})\mathrm{Reg}_{j}\big( p_{n-1}(x)e^{-2nx^{2}} \big),
\end{multline}
where $\mathcal{A} = \sum_{j=1}^{m} \alpha_{j}$. A similar and simpler calculation for $I_{2}$ leads to 
\begin{equation}\label{J_2}
I_{2} = 4n \frac{\partial_{\beta_{k}}\kappa_{n-1}}{\kappa_{n}} - \sum_{j=1}^{m} \partial_{\beta_{k}}p_{n-1}(t_{j})\mathrm{Reg}_{j}\big( p_{n}(x)e^{-2nx^{2}} \big).
\end{equation}
For $l \in \{1,...,m\}$, we define $\widetilde{Y}(t_{l})$ as
\begin{equation}\label{def second column of tilde Y}
\widetilde{Y}(t_{l}) = \begin{pmatrix}
Y_{11}(t_{l}) & \mathrm{Reg}_{l} \left( \frac{1}{2\pi i}Y_{11}(x)e^{-2nx^{2}} \right) \\
Y_{21}(t_{l}) & \mathrm{Reg}_{l} \left( \frac{1}{2\pi i}Y_{21}(x)e^{-2nx^{2}} \right) \\
\end{pmatrix}.
\end{equation}
Note that from Proposition \ref{Prop: regularised integral}, for $l \in \{1,...,m\}$,
\begin{equation}
\widetilde{Y}_{j2}(t_{l}) = \lim_{z\to t_{l}} \alpha_{l}Y_{j2}(z) - c_{l} Y_{j1}(t_{l})(z-t_{l})^{\alpha_{l}}, \qquad j = 1,2,
\end{equation}
where the limit is taken along a non-tangential to the real line path and
\begin{equation}
c_{l} = \frac{\pi\alpha_{l}}{\sin(\pi\alpha_{l})}\frac{e^{-2nt_{l}^{2}}}{2\pi i} \omega_{l}(t_{l})(e^{i\pi\beta_{l}}-e^{-i\pi\alpha_{l}}e^{-i\pi\beta_{l}}).
\end{equation}
Thus the matrix $\widetilde{Y}(t_{l})$ has not determinant $1$, but
\begin{equation}
\alpha_{l} = \alpha_{l} \det Y(t_{l}) = \det \widetilde{Y}(t_{l}).
\end{equation}
Putting together \eqref{lol 20}, \eqref{J_1} and \eqref{J_2}, we obtain
\begin{multline}\label{differential identity}
\partial_{\beta_{k}}\log D_{n}(\vec{\alpha},\vec{\beta},2x^{2},0) = -(n+\mathcal{A}) \partial_{\beta_{k}}\log(\kappa_{n}\kappa_{n-1}) - 2n\partial_{\beta_{k}}\left( \frac{\kappa_{n-1}^{2}}{\kappa_{n}^{2}} \right) + 4n \partial_{\beta_{k}} \left(\gamma_{n}-\frac{\eta_{n}^{2}}{2}\right) \\[0.1cm] 
\hspace*{0.6cm} + \sum_{j=1}^{m} \left[\widetilde{Y}_{22}(t_{j})\partial_{\beta_{k}}Y_{11}(t_{j}) - \widetilde{Y}_{12}(t_{j})\partial_{\beta_{k}}Y_{21}(t_{j}) + Y_{11}(t_{j})\widetilde{Y}_{22}(t_{j})\partial_{\beta_{k}} \log (\kappa_{n}\kappa_{n-1}) \right].
\end{multline}

\subsection{A general differential identity}\label{subsection: diff identity s and t}

Let $\gamma$ be a parameter of the weight $w$ whose dependence is smooth. From the well-known \cite{Szego OP, Deift} relation $D_{n}(\vec{\alpha},\vec{\beta},V,W) = \prod_{j=0}^{n-1} \kappa_{j}^{-2}$ (which we assume is valid), taking the $\log$ and then differentiating with respect to $\gamma$, we have
\begin{equation}\label{general diff identity 1}
\partial_{\gamma} \log D_{n}(\vec{\alpha},\vec{\beta},V,W) = - \int_{\mathbb{R}} \partial_{\gamma} \bigg( \sum_{j=0}^{n-1} p_{j}^{2}(x) \bigg) w(x)dx = \int_{\mathbb{R}} \bigg( \sum_{j=0}^{n-1} p_{j}^{2}(x) \bigg) \partial_{\gamma} w(x)dx,
\end{equation}
where the last equality comes from the orthogonality relations, since $\int_{\mathbb{R}} \sum_{j=0}^{n-1} p_{j}^{2}(x)w(x)dx = n$.
By the Christoffel-Darboux formula (see e.g. \cite{Szego OP, Deift}) and \eqref{Y definition}, the equation \eqref{general diff identity 1} can be rewritten as 
\begin{equation}\label{general diff identity}
\displaystyle \partial_{\gamma} \log D_{n}(\vec{\alpha},\vec{\beta},V,W) =  \frac{1}{2\pi i}\int_{\mathbb{R}}[Y^{-1}(x)Y^{\prime}(x)]_{21}\partial_{\gamma}w(x)dx.
\end{equation}
We will use this differential identity for the steps 2 and 3, which correspond respectively to $w_{s}$ and $w_{t}$, as described in the outline. In these cases, \eqref{general diff identity} becomes
\begin{align}
& \partial_{s} \log D_{n}(\vec{\alpha},\vec{\beta},V_{s},0) =  \frac{1}{2\pi i}\int_{\mathbb{R}}[Y^{-1}(x)Y^{\prime}(x)]_{21}\partial_{s}w_{s}(x)dx, \label{diff identity s} \\ 
& \partial_{t} \log D_{n}(\vec{\alpha},\vec{\beta},V,W_{t}) =  \frac{1}{2\pi i}\int_{\mathbb{R}}[Y^{-1}(x)Y^{\prime}(x)]_{21}\partial_{t}w_{t}(x)dx, \label{diff identity t}
\end{align}
where \eqref{diff identity s} is valid for $(\vec{\alpha},\vec{\beta},s)\in \widetilde{\mathcal{P}}_{2}^{(n)}$, and where \eqref{diff identity t} is valid for $(\vec{\alpha},\vec{\beta},t)\in \widetilde{\mathcal{P}}_{3}^{(n)}$ (see the discussion at the beginning of Section \ref{Section:diff identities}).

\section{Steepest descent analysis}\label{Section: steepest descent}\label{Section:steepest}
In this section, we will perform an asymptotic analysis on the RH problem for $Y$ as $n \to \infty$. Our analysis is based on the Deift/Zhou steepest descent method \cite{DeiftZhou1992,DeiftZhou}.

\subsection{Equilibrium measure and $g$-function}
We denote by $U_{V}$ the maximal open neighbourhood of $\mathbb{R}$ in which $V$ is analytic, and by $U_W$ an open neighbourhood of $\mathcal{S} = [-1,1]$ in which $W$ is analytic, sufficiently small such that $U_W \subset U_V$.
From \cite[Theorem 1.38 (i), equations (1.39), (1.43), (1.44) and (1.45)]{DeiKriMcL}, we have
\begin{equation}\label{lol 23}
\psi(x) \sqrt{1-x^{2}} = \Re \left( i \frac{V^{\prime}(x)}{2\pi} + \frac{h(x)}{2\pi} \sqrt{1-x^{2}} \right),
\end{equation}
where
\begin{equation}
h(z) = -\frac{1}{2\pi i} \int_{\mathcal{C}} \frac{V^{\prime}(y)}{\sqrt{y^{2}-1}}\frac{dy}{y-z},
\end{equation}
and $\mathcal{C}$ is a closed loop oriented in the clockwise direction encircling $z$ and $[-1,1]$. Clearly, $h$ is analytic in $U_V$. Also, since $V$ is real, \eqref{lol 23} reduces to
\begin{equation}
\psi(x) = \frac{h(x)}{2\pi} = \frac{1}{2\pi^{2}} \dashint_{-1}^{1} \frac{V^{\prime}(y)}{\sqrt{1-y^{2}}} \frac{dy}{y-x}.
\end{equation}
Thus, we have that $\psi$ is analytic in $U_V$ and $\psi(x) \in \mathbb{R}$ if $x \in \mathbb{R}$.

\hspace{-0.51cm}We define the $g$-function, which is useful in order to normalize the RH problem at $\infty$, by 
\begin{equation}
g(z) = \int_{\mathcal{S}} \log (z-s) \rho(s) ds, \qquad \rho(x) = \psi(x) \sqrt{1-x^{2}},
\end{equation}
where the principal branch cut is chosen for the logarithm. The $g$-function is analytic in $\mathbb{C}\setminus (-\infty,1]$ and possesses the following properties
\begin{align}
& g_{+}(x) + g_{-}(x) = 2 \int_{S} \log |x-s| \rho(s)ds, & & x \in \mathbb{R}, \label{g+ + g-} \\
& g_{+}(x) - g_{-}(x) = 2 \pi i, & & x < -1, \label{g+ - g- 1} \\
& g_{+}(x) - g_{-}(x) = 2 \pi i \int_{x}^{1} \rho(s)ds, & & x \in \mathcal{S}, \label{g+ - g- 2}\\
& g_{+}(x) - g_{-}(x) = 0, & & 1 < x. \label{g+ - g- 3}
\end{align}
Consider the function
\begin{equation}\label{integral form of xi}
\xi(z) = - \pi \int_{1}^{z}\tilde{\rho}(s)ds,
\end{equation}
where the path of integration lies in $U_V \setminus (-\infty,1)$ and $\tilde{\rho}(s) = \psi(s)\sqrt{s^{2}-1}$ is analytic in $U_V\setminus \mathcal{S}$. Since $\tilde{\rho}_{\pm}(s) = \pm i \rho(s)$ for $s \in \mathcal{S}$, we have
\begin{equation} \label{xi +}
2\xi_{\pm}(x) = g_{\pm}(x)-g_{\mp}(x) = 2g_{\pm}(x) + \ell - V(x), \qquad x \in \mathcal{S},
\end{equation}
where we have used \eqref{var equality} together with \eqref{g+ + g-}. Analytically continuing $\xi(z)-g(z)$ in \eqref{xi +}, we have
\begin{equation}\label{relation between g and xi}
\xi(z) = g(z) + \frac{\ell}{2} - \frac{V(z)}{2}, \qquad \mbox{ for all } z \in U_V\setminus(-\infty,1).
\end{equation}
The function $\xi(z)$ possesses certain properties which will be useful later to show that the jump matrices of a RH problem converge to the identity matrix as $n \to \infty$. First, note from \eqref{relation between g and xi} that the variational inequality \eqref{var inequality} (which we recall is strict) can be rewritten as
\begin{equation}\label{lol 15}
\xi_{+}(x)+\xi_{-}(x)<0, \qquad \mbox{for } |x|>1, \mbox{ } x \in \mathbb{R}.
\end{equation}
Also, since $V(x)/\log|x| \to +\infty$ and $g(x) = \bigO(\log |x|)$ as $|x| \to +\infty$, $x \in \mathbb{R}$, we have
\begin{equation}\label{lol 16}
\big(\xi_{+}(x)+\xi_{-}(x)\big)/V(x) \to -1, \qquad \mbox{as } |x|\to+\infty, \mbox{ } x \in \mathbb{R}.
\end{equation}
Finally, from the assumption that $\psi$ is positive on $\mathcal{S}$, a direct analysis (see e.g. \cite[Lemma 4.4]{BerWebbWong}) of the integral \eqref{integral form of xi} shows that for $z$ in a small enough neighbourhood of $(-1,1)$, 
\begin{equation}\label{lol 17}
\Re \xi(z) >0, \qquad \mbox{ if } \Im z \neq 0.
\end{equation}
We will also need later large $z$ asymptotics of $e^{ng(z)}$ for the case $V(x) = 2x^{2}$:
\begin{equation}\label{asymptotics for e^{ng}}
e^{ng(z)} = z^{n} \left( 1-\frac{n}{8z^{2}}+\bigO(z^{-4}) \right), \qquad \mbox{ as } z \to \infty.
\end{equation}
\subsection{First transformation: $Y \mapsto T$}
The first step consists of normalizing the RH problem at $\infty$, which can be done with the following transformation
\begin{equation}
T(z) = e^{\frac{n\ell}{2}\sigma_{3}}Y(z)e^{-ng(z)\sigma_{3}}e^{-\frac{n\ell}{2}\sigma_{3}}.
\end{equation}
$T$ satisfies the following RH problem.
\subsubsection*{RH problem for $T$}
\begin{itemize}
\item[(a)] $T : \mathbb{C}\setminus \mathbb{R} \to \mathbb{C}^{2\times 2}$ is analytic.
\item[(b)] Using \eqref{var equality}, \eqref{g+ + g-}, \eqref{g+ - g- 1}, \eqref{g+ - g- 2}, \eqref{g+ - g- 3} and \eqref{xi +}, a direct computation shows that $T$ has the following jumps:
\begin{equation}
T_{+}(x) = T_{-}(x) J_{T}(x), \hspace{0.5cm} \mbox{ for } x \in \mathbb{R}\setminus \left\{t_{1},...,t_{m} \right\},
\end{equation}
where
\begin{equation}
J_{T}(x) = \left\{ \begin{array}{l l}
\displaystyle \begin{pmatrix}
1 & e^{W(x)} \omega(x) e^{n(\xi_{+}(x)+\xi_{-}(x))} \\
0 & 1
\end{pmatrix}, & \displaystyle \mbox{if } x < -1, \\[0.23cm]
\begin{pmatrix}
e^{-2n\xi_{+}(x)} & e^{W(x)} \omega(x) \\
0 & e^{2n\xi_{+}(x)}
\end{pmatrix}, & \displaystyle \mbox{if } x\in (-1,1)\setminus\{t_{1},...,t_{m}\}, \\[0.23cm]
\displaystyle \begin{pmatrix}
1 & e^{W(x)}\omega(x) e^{2n\xi(x)} \\
0 & 1
\end{pmatrix}, & \displaystyle \mbox{if } 1<x. \\
\end{array} \right.
\end{equation}
\item[(c)] As $z \to \infty$, we have $T(z) = I + \bigO(z^{-1})$.
\item[(d)] As $z$ tends to $t_{k}$ for a certain $k \in \{1,...,m\}$, the behaviour of $T$ is
\begin{equation}
\begin{array}{l l}
\displaystyle T(z) = \begin{pmatrix}
\bigO(1) & \bigO(\log (z-t_{k})) \\
\bigO(1) & \bigO(\log (z-t_{k}))
\end{pmatrix}, & \displaystyle \mbox{ if } \Re \alpha_{k} = 0, \\[0.35cm]
\displaystyle T(z) = \begin{pmatrix}
\bigO(1) & \bigO(1)+\bigO((z-t_{k})^{\Re\alpha_{k}}) \\
\bigO(1) & \bigO(1)+\bigO((z-t_{k})^{\Re\alpha_{k}})
\end{pmatrix}, & \displaystyle \mbox{ if } \Re \alpha_{k} \neq 0.
\end{array}
\end{equation}

As $z$ tends to $-1$ or $1$, we have $T(z) = \bigO(1)$.
\end{itemize}
\subsection{Second transformation: $T \mapsto S$}
We will use the following factorization of $J_{T}(x)$ for $x \in \mathcal{S}\setminus\{t_{1},...,t_{m}\}$:
\begin{multline}
\begin{pmatrix}
e^{-2n\xi_{+}(x)} & e^{W(x)} \omega(x) \\
0 & e^{-2n\xi_{-}(x)}
\end{pmatrix} = \begin{pmatrix}
1 & 0 \\ e^{-W(x)}\omega(x)^{-1}e^{-2n \xi_{-}(x)} & 1
\end{pmatrix} \\ \times \begin{pmatrix}
0 & e^{W(x)}\omega(x) \\ -e^{-W(x)}\omega(x)^{-1} & 0
\end{pmatrix} \begin{pmatrix}
1 & 0 \\ e^{-W(x)}\omega(x)^{-1}e^{-2n \xi_{+}(x)} & 1
\end{pmatrix}.
\end{multline}
\begin{figure}[t]
    \begin{center}
    \setlength{\unitlength}{1truemm}
    \begin{picture}(100,30)(28,30)
        \put(12,40){\line(1,0){120}}
        \put(30,39.9){\thicklines\vector(1,0){.0001}}
        \put(50,39.9){\thicklines\vector(1,0){.0001}}
        \put(70,39.9){\thicklines\vector(1,0){.0001}}
        \put(90,39.9){\thicklines\vector(1,0){.0001}}
        \put(110,39.9){\thicklines\vector(1,0){.0001}}
        \put(49.6,47.4){\thicklines\vector(1,0){.0001}}
        \put(49.6,32.4){\thicklines\vector(1,0){.0001}}
        \put(69.6,47.4){\thicklines\vector(1,0){.0001}}
        \put(69.6,32.4){\thicklines\vector(1,0){.0001}}
        \put(89.6,47.4){\thicklines\vector(1,0){.0001}}
        \put(89.6,32.4){\thicklines\vector(1,0){.0001}}
        \put(38,39.9){\thicklines\circle*{1.2}} \put(36,36.5){$-1$}
        \put(58,39.9){\thicklines\circle*{1.2}} \put(57.35,36.3){$t_{1}$}
        \put(78,39.9){\thicklines\circle*{1.2}} \put(77,36.3){$t_{m}$}
        \put(98,39.9){\thicklines\circle*{1.2}} \put(98,36.8){$1$}
        \qbezier(38,39.9)(48,55)(58,39.9)
        \qbezier(58,39.9)(68,55)(78,39.9)
        \qbezier(78,39.9)(88,55)(98,39.9)
        \qbezier(38,39.9)(48,24.8)(58,39.9)
        \qbezier(58,39.9)(68,24.8)(78,39.9)
        \qbezier(78,39.9)(88,24.8)(98,39.9)
        \put(65.5,50){$\gamma_+$}
        \put(65.5,29){$\gamma_-$}
    \end{picture}
    \caption{Jump contours for the RH problem for $S$ with $m=2$. The lens contours are labelled $\gamma_+ \subset U_W$ and $\gamma_- \subset U_W$, and they lie in the upper and lower half plane respectively. \label{Fig:S}}
\end{center}
\end{figure}
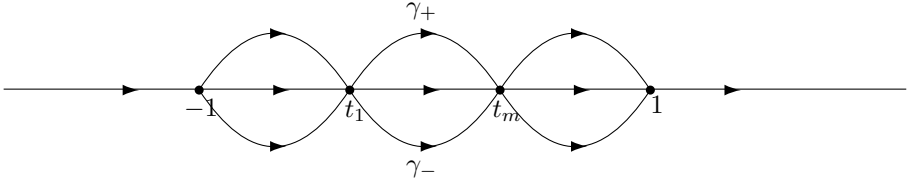
The functions $\omega_{\alpha_{k}}$ and $\omega_{\beta_{k}}$ (see \eqref{first def of omega}) can be analytically continued as follows:
\begin{equation}
\omega_{\alpha_{k}}(z) = \left\{ \begin{array}{l l}
(t_{k}-z)^{\alpha_{k}}, & \mbox{ if } \Re z < t_{k}, \\
(z-t_{k})^{\alpha_{k}}, & \mbox{ if } \Re z > t_{k},
\end{array} \right. \qquad \omega_{\beta_{k}}(z) = \left\{ \begin{array}{l l}
e^{i\pi\beta_{k}}, & \mbox{ if } \Re z < t_{k}, \\
e^{-i \pi \beta_{k}}, & \mbox{ if } \Re z > t_{k}.
\end{array}  \right.
\end{equation}
We open the lenses $\gamma_{+}$ and $\gamma_{-}$ around $S$ as illustrated in Figure \ref{Fig:S}, such that they are inside $U_W$ and we define
\begin{equation}
S(z) = T(z)  \left\{ \begin{array}{l l}
\begin{pmatrix}
1 & 0 \\ \displaystyle -e^{-W(z)}\omega(z)^{-1}e^{-2n \xi(z)} & 1
\end{pmatrix}, & \mbox{ if } z \mbox{ is inside the lenses}, \Im z > 0, \\
\begin{pmatrix}
1 & 0 \\ \displaystyle e^{-W(z)}\omega(z)^{-1}e^{-2n \xi(z)} & 1
\end{pmatrix}, & \mbox{ if } z \mbox{ is inside the lenses}, \Im z < 0, \\
I, & \mbox{ if } z \mbox{ is outside the lenses}. \\
\end{array} \right.
\end{equation}
$S$ satisfies the following RH problem
\subsubsection*{RH problem for $S$}
\begin{itemize}
\item[(a)] $S : \mathbb{C}\setminus (\mathbb{R} \cup \gamma_{+} \cup \gamma_{-}) \to \mathbb{C}^{2\times 2}$ is analytic.
\item[(b)] $S$ has the following jumps:
\begin{align}
& S_{+}(z) = S_{-}(z)\begin{pmatrix}
1 & e^{W(z)}\omega(z) e^{n(\xi_{+}(z)+\xi_{-}(z))} \\
0 & 1
\end{pmatrix}, & & \mbox{ if } z < -1, \\
& S_{+}(z) = S_{-}(z)\begin{pmatrix}
1 & e^{W(z)}\omega(z)e^{2n\xi(z)} \\
0 & 1
\end{pmatrix}, & & \mbox{ if } 1 < z, \\
& S_{+}(z) = S_{-}(z)\begin{pmatrix}
0 & e^{W(z)}\omega(z) \\ -e^{-W(z)}\omega(z)^{-1} & 0
\end{pmatrix}, & & \mbox{ if } z\in(-1,1)\setminus\{t_{1},...,t_{m}\}, \\
& S_{+}(z) = S_{-}(z)\begin{pmatrix}
1 & 0 \\ e^{-W(z)}\omega(z)^{-1}e^{-2n \xi(z)} & 1
\end{pmatrix}, & & \mbox{ if } z \in \gamma_{+} \cup \gamma_{-}.
\end{align}
\item[(c)] As $z \to \infty$, we have $S(z) = I + \bigO(z^{-1})$.
\item[(d)] As $z$ tends to $t_{k}$ for a certain $k \in \{1,...,m\}$, we have
\begin{equation*}
\begin{array}{l l}
\displaystyle S(z) = \left\{ \begin{array}{l l}
\begin{pmatrix}
\bigO(1) & \bigO(\log (z-t_{k})) \\
\bigO(1) & \bigO(\log (z-t_{k}))
\end{pmatrix}, & \mbox{if } z \mbox{ is outside the lenses}, \\
\begin{pmatrix}
\bigO(\log (z-t_{k})) & \bigO(\log (z-t_{k})) \\
\bigO(\log (z-t_{k})) & \bigO(\log (z-t_{k}))
\end{pmatrix}, & \mbox{if } z \mbox{ is inside the lenses},
\end{array} \right., & \displaystyle \mbox{ if } \Re \alpha_{k} = 0, \\[0.9cm]

\displaystyle S(z) = \left\{ \begin{array}{l l}
\begin{pmatrix}
\bigO(1) & \bigO(1) \\
\bigO(1) & \bigO(1)
\end{pmatrix}, & \mbox{if } z \mbox{ is outside the lenses}, \\
\begin{pmatrix}
\bigO((z-t_{k})^{-\Re\alpha_{k}}) & \bigO(1) \\
\bigO((z-t_{k})^{-\Re\alpha_{k}}) & \bigO(1)
\end{pmatrix}, & \mbox{if } z \mbox{ is inside the lenses},
\end{array} \right. , & \displaystyle \mbox{ if } \Re\alpha_{k} > 0, \\[0.9cm]

\displaystyle S(z) = \begin{pmatrix}
\bigO(1) & \bigO((z-t_{k})^{\Re\alpha_{k}}) \\
\bigO(1) & \bigO((z-t_{k})^{\Re\alpha_{k}}) 
\end{pmatrix}, & \displaystyle \mbox{ if } \Re\alpha_{k} < 0.
\end{array}
\end{equation*}
As $z$ tends to $-1$ or $1$, we have $S(z) = \bigO(1)$.
\end{itemize}
From \eqref{lol 15} and \eqref{lol 17}, the jumps for $S$ on $(\gamma_{+}\cup\gamma_{-}\cup\mathbb{R})\setminus \mathcal{S}$ are exponentially close to the identity matrix as $n \to \infty$. Nevertheless, this is only a pointwise convergence and it breaks down as $z$ approaches $t_{1}$,...,$t_{m}$, $-1$ and $1$. This simply follows from 
\begin{align*}
& \xi_{+}(x) + \xi_{-}(x) \to \xi_{+}(\pm 1) + \xi_{-}(\pm 1) = 0, & & \mbox{ as } x\to \pm 1, \mbox{ } x \in \mathbb{R}, \\
& \Re \xi_{\pm}(x) = 0, & & \mbox{ for } x \in \mathcal{S}.
\end{align*}
Therefore, for $z$ outside of a neighbourhood of $\mathcal{S}$, from \eqref{lol 16} and the assumption $W(x) = \bigO(V(x))$ as $|x| \to \infty$, the jumps for $S$ are uniformly exponentially close to $I$ as $n \to \infty$. If we ignore these jumps, we are left with a simpler RH problem (called the global parametrix) whose solution is a good approximation of $S$ away from these points. For $z$ in small neighbourhoods of $t_{1}$,...,$t_{m}$, $-1$ and $1$, the global parametrix is no longer a good approximation for $S$, and we will construct local parametrices around these points.
\subsection{Global parametrix}
Ignoring the exponentially small terms as $n \to \infty$ in the jumps of $S$ and small neighbourhoods of $-1$, $t_{1},...,t_{m}$ and $1$, we are left with the following RH problem, whose solution $P^{(\infty)}$ is a good approximation of $S$ away from neighbourhoods of $-1$, $t_{1},...,t_{m}$ and $1$:
\subsubsection*{RH problem for $P^{(\infty)}$}
\begin{itemize}
\item[(a)] $P^{(\infty)} : \mathbb{C}\setminus [-1,1] \to \mathbb{C}^{2\times 2}$ is analytic.
\item[(b)] $P^{(\infty)}$ has the following jumps:
\begin{align}\label{jumps for P inf}
& \hspace{-0.3cm} P^{(\infty)}_{+}(z) = P^{(\infty)}_{-}(z)\begin{pmatrix}
0 & e^{W(z)}\omega(z) \\ -e^{-W(z)}\omega(z)^{-1} & 0
\end{pmatrix}, & & \mbox{ if } z \in (-1,1)\setminus\{t_{1},...,t_{m}\}.
\end{align}
\item[(c)] As $z \to \infty$, we have $P^{(\infty)}(z) = I + P_{1}^{(\infty)} z^{-1} + P_{2}^{(\infty)} z^{-2} + \bigO(z^{-3})$.

As $z$ tends to $-1$, we have $P^{(\infty)}(z) = \bigO((z+1)^{-1/4})$.

As $z$ tends to $1$, we have $P^{(\infty)}(z) = \bigO((z-1)^{-1/4})$.

As $z$ tends to $t_{k}$, $k\in\{1,...,m\}$, we have $P^{(\infty)}(z) = \bigO(1)(z-t_{k})^{-(\frac{\alpha_{k}}{2}+\beta_{k})\sigma_{3}}$.
\end{itemize}
The unique solution $P^{(\infty)}$ of the above RH problem can be constructed similarly as in \cite{Kuijlaars2}. Define $a(z) = \sqrt[4]{\frac{z+1}{z-1}}$, analytic on $\mathbb{C}\setminus [-1,1]$ and such that $a(z) \sim 1$ as $z \to \infty$. It can be checked that $P^{(\infty)}$ is given by
\begin{equation}\label{Pinf Region 1}
P^{(\infty)}(z) = D_{\infty}^{\sigma_{3}} \begin{pmatrix}
\frac{1}{2}(a(z)+a^{-1}(z)) & \frac{1}{-2i}(a(z)-a^{-1}(z)) \\
\frac{1}{2i}(a(z)-a^{-1}(z)) & \frac{1}{2}(a(z)+a^{-1}(z))
\end{pmatrix} D(z)^{-\sigma_{3}},
\end{equation}
with $D(z) = D_{W}(z)D_{\alpha}(z)D_{\beta}(z)$, and $D_{W}(z)$, $D_{\alpha}(z)$ and $D_{\beta}(z)$ are the three Szeg\H{o} functions defined by
\begin{align}
& D_{W}(z) = \exp \left( \frac{\sqrt{z^{2}-1}}{2\pi}\int_{-1}^{1} \frac{W(x)}{\sqrt{1-x^{2}}}\frac{dx}{z-x} \right), \label{def D_T} \\
& D_{\alpha}(z) = \prod_{j=1}^{m}\exp\left( \frac{\sqrt{z^{2}-1}}{2\pi}\int_{-1}^{1} \frac{\log \omega_{\alpha_{j}}(x)}{\sqrt{1-x^{2}}} \frac{dx}{z-x} \right) = (z+\sqrt{z^{2}-1})^{-\frac{\mathcal{A}}{2}}\prod_{j=1}^{m} (z-t_{j})^{\frac{\alpha_{j}}{2}}, \label{def D_alpha} \\
& D_{\beta}(z) = \prod_{j=1}^{m}\exp\left( \frac{\sqrt{z^{2}-1}}{2\pi}\int_{-1}^{1} \frac{\log \omega_{\beta_{j}}(x)}{\sqrt{1-x^{2}}} \frac{dx}{z-x} \right) = e^{\frac{i\pi\mathcal{B}}{2}} \prod_{j=1}^{m} \left( \frac{zt_{j}-1-i\sqrt{(z^{2}-1)(1-t_{j}^{2})}}{z-t_{j}} \right)^{\beta_{j}}, \label{def D_beta}
\end{align}
where $\mathcal{B} = \sum_{j=1}^{m} \beta_{j}$, and we recall $\mathcal{A} = \sum_{j=1}^{m} \alpha_{j}$.
A proof of the simplified form of \eqref{def D_alpha} and \eqref{def D_beta} can be found in \cite[equation (51)]{Krasovsky} and \cite[equation (4.14)]{ItsKrasovsky} respectively.
Also, the function $D$ has a limit at $\infty$ given by
\begin{align}\label{D inf}
& D_{\infty} = \lim_{z\to\infty}D(z)= D_{W,\infty}D_{\alpha,\infty}D_{\beta,\infty}, \qquad D_{W,\infty} = \exp \left( \frac{1}{2\pi} \int_{-1}^{1} \frac{W(x)}{\sqrt{1-x^{2}}}dx \right), \\ & D_{\alpha,\infty} = 2^{-\frac{\mathcal{A}}{2}}, \qquad D_{\beta,\infty} = \exp \bigg( i\sum_{j=1}^{m}\beta_{j} \arcsin t_{j} \bigg). \nonumber
\end{align}
We will use later the following expansions of $D(z)$: as $z \to t_{k}$, $\Im z>0$, we have
\begin{align}
& D_{\alpha}(z) = e^{-i \frac{\mathcal{A}}{2}\arccos t_{k}}\bigg( \prod_{j \neq k} |t_{k}-t_{j}|^{\frac{\alpha_{j}}{2}} \prod_{j=k+1}^{m}e^{\frac{i\pi\alpha_{j}}{2}} \bigg) (z-t_{k})^{\frac{\alpha_{k}}{2}}(1+\bigO(z-t_{k})), \label{asymptotics for D_alpha in D_t} \\
& D_{\beta}(z) = e^{-\frac{i\pi}{2}(\mathcal{B}_{k}+\beta_{k})}  \bigg( \prod_{j\neq k} T_{kj}^{\beta_{j}}\bigg) (1-t_{k}^{2})^{-\beta_{k}}2^{-\beta_{k}}(z-t_{k})^{\beta_{k}}(1+\bigO(z-t_{k})), \label{asymptotics for D_beta in D_t} 
\end{align}
where
\begin{equation*}
\mathcal{B}_{k} = \sum_{j=1}^{k-1}\beta_{j} - \sum_{j=k+1}^{m} \beta_{j}, \qquad T_{kj} = \frac{1-t_{k}t_{j}-\sqrt{(1-t_{k}^{2})(1-t_{j}^{2})}}{|t_{k}-t_{j}|}.
\end{equation*}
The expansions of $D(z)$ near $1$ and $-1$ are nicely expressed in terms of
\begin{equation*}
\widetilde{\mathcal{B}}_{1} = 2 i \sum_{j=1}^{m} \sqrt{\frac{1+t_{j}}{1-t_{j}}}\beta_{j}, \qquad \widetilde{\mathcal{B}}_{-1} = 2 i \sum_{j=1}^{m} \sqrt{\frac{1-t_{j}}{1+t_{j}}}\beta_{j}.
\end{equation*}
As $z \to 1$, $\Im z>0$, we have
\begin{align}
& D_{\alpha}^{2}(z)\prod_{j=1}^{m}(z-t_{j})^{-\alpha_{j}} =  1-\sqrt{2} \mathcal{A} \sqrt{z-1} + \mathcal{A}^{2}(z-1)+ \bigO((z-1)^{3/2}), \label{asymptotics for D_alpha in D_1}\\
& D_{\beta}^{2}(z)e^{i\pi\mathcal{B}} = 1+   \sqrt{2} \widetilde{\mathcal{B}}_{1} \sqrt{z-1} + \widetilde{\mathcal{B}}_{1}^{2}(z-1) + \bigO((z-1)^{3/2}). \label{asymptotics for D_beta in D_1}
\end{align}
As $z \to -1$, $\pm\Im z > 0$, we have
\begin{align}
& D_{\alpha}^{2}(z)\prod_{j=1}^{m}(t_{j}-z)^{-\alpha_{j}} =  1\pm i\sqrt{2} \mathcal{A} \sqrt{z+1} - \mathcal{A}^{2}(z+1)+ \bigO((z+1)^{3/2}), \label{asymptotics for D_alpha in D_-1}\\
& D_{\beta}^{2}(z)e^{-i\pi\mathcal{B}} = 1 \pm  i\sqrt{2}  \widetilde{\mathcal{B}}_{-1} \sqrt{z+1} - \widetilde{\mathcal{B}}_{-1}^{2} (z+1) + \bigO((z+1)^{3/2}). \label{asymptotics for D_beta in D_-1}
\end{align}
The first term of the expansion of \eqref{Pinf Region 1} as $z \to \infty$ with $W \equiv 0$ is given by
\begin{equation}\label{P_1^inf}
P_{1}^{(\infty)} = \begin{pmatrix}
\displaystyle \sum_{j=1}^{m} \left( \frac{\alpha_{j}t_{j}}{2}+i\sqrt{1-t_{j}^{2}}\beta_{j}\right) & \displaystyle \frac{i}{2} D_{\infty}^{2} \\
\displaystyle -\frac{i}{2}D_{\infty}^{-2} & \displaystyle -\sum_{j=1}^{m} \left( \frac{\alpha_{j}t_{j}}{2}+i\sqrt{1-t_{j}^{2}}\beta_{j}\right)
\end{pmatrix}.
\end{equation}
The next term requires more calculations, but we will only use the $1,1$ entry:
\begin{equation}\label{P_2^inf}
P_{2,11}^{(\infty)} = \frac{1}{8}+\frac{1}{2}\left( \sum_{j=1}^{m} \left( \frac{\alpha_{j}t_{j}}{2}+i\sqrt{1-t_{j}^{2}}\beta_{j}\right) \right)^{2} - \frac{1}{8}\sum_{j=1}^{m}\alpha_{j}(1-2t_{j}^{2}) + \frac{i}{2}\sum_{j=1}^{m}t_{j}\sqrt{1-t_{j}^{2}}\beta_{j}.
\end{equation}
\subsection{Local parametrix near $t_{k}$}
We assume that there exists $\delta > 0$ independent of $n$ such that
\begin{equation}
\min_{j\neq k}\{ |t_{j}-t_{k}|,|t_{j}-1|,|t_{j}+1|\} \geq \delta.
\end{equation}
Consider a fixed disk $\mathcal{D}_{t_{k}}$, centered at $t_{k}$, of radius smaller than $\delta/3$ and such that $\mathcal{D}_{t_{k}} \subset U_W$. Inside $\mathcal{D}_{t_{k}}$, the local parametrix $P^{(t_{k})}$ is a good approximation of $S$ and must satisfy the following RH problem.
\subsubsection*{RH problem for $P^{(t_{k})}$}
\begin{itemize}
\item[(a)] $P^{(t_{k})} : \mathcal{D}_{t_{k}}\setminus (\mathbb{R} \cup \gamma_{+} \cup \gamma_{-}) \to \mathbb{C}^{2\times 2}$ is analytic.
\item[(b)] $P^{(t_{k})}$ has the following jumps:
\begin{align}
& P_{+}^{(t_{k})}(z) = P_{-}^{(t_{k})}(z)\begin{pmatrix}
0 & e^{W(z)}\omega(z) \\ -e^{-W(z)}\omega(z)^{-1} & 0
\end{pmatrix}, & & \mbox{ if } z \in (\mathbb{R}\setminus\{t_{k}\})\cap \mathcal{D}_{t_{k}}, \\
& P_{+}^{(t_{k})}(z) = P_{-}^{(t_{k})}(z)\begin{pmatrix}
1 & 0 \\ e^{-W(z)}\omega(z)^{-1}e^{-2n \xi(z)} & 1
\end{pmatrix}, & & \mbox{ if } z \in (\gamma_{+} \cup \gamma_{-}) \cap \mathcal{D}_{t_{k}}.
\end{align}
\item[(c)] As $n \to \infty$, we have $P^{(t_{k})}(z) = (I + \bigO(n^{-1+2|\Re \beta_{k}|}))P^{(\infty)}(z)$ uniformly for $z \in \partial \mathcal{D}_{t_{k}}$.
\item[(d)] As $z$ tends to $t_{k}$, we have $S(z)P^{(t_{k})}(z)^{-1} = \bigO(1)$.
\end{itemize}
We follow \cite{ItsKrasovsky,FouMarSou,DIK} to construct the solution of the above RH problem, and we define the function $f_{t_{k}}$ by
\begin{equation}
f_{t_{k}}(z) = -2 \left\{ \begin{array}{l l}
\xi(z)-\xi_{+}(t_{k}), & \Im z > 0, \\
-\xi(z)-\xi_{+}(t_{k}), & \Im z < 0,
\end{array} \right. = 2\pi i \int_{t_{k}}^{z}\psi(s)\sqrt{1-s^{2}}ds.
\end{equation}
This a conformal map from $\mathcal{D}_{t_{k}}$ to a neighbourhood of $0$, as shown from the expansion of $f_{t_{k}}(z)$ near $t_{k}$
\begin{equation}\label{asymptotics for f in D_t}
f_{t_{k}}(z) = 2\pi i \psi(t_{k}) \sqrt{1-t_{k}^{2}}(z-t_{k}) \left( 1+\bigO((z-t_{k})) \right), \qquad \mbox{ as } z \to t_{k}.
\end{equation}
Now, we use the freedom we had in the choice of the lenses by requiring that 
\begin{equation}
f_{t_{k}}(\gamma_{+}\cap \mathcal{D}_{t_{k}}) \subset (\Gamma_{4}\cup \Gamma_{2}), \qquad f_{t_{k}}(\gamma_{-}\cap \mathcal{D}_{t_{k}}) \subset (\Gamma_{6}\cup \Gamma_{8}),
\end{equation}
where the contour $\Gamma_{4}, \Gamma_{2}, \Gamma_{6}$ and $\Gamma_{8}$ are shown in Figure \ref{Fig:HG}. We define \begin{equation}\label{def of Wk}
\widetilde{W}_{k}(z) = \left\{ \begin{array}{l l}
(z-t_{k})^{\frac{\alpha_{k}}{2}}e^{-\frac{ i \pi\alpha_{k}}{2}}, & z \in Q_{+,k}^{R}, \\
(t_{k}-z)^{\frac{\alpha_{k}}{2}}e^{\frac{ i \pi\alpha_{k}}{2}}, & z \in Q_{+,k}^{L}, \\
(t_{k}-z)^{\frac{\alpha_{k}}{2}}e^{-\frac{ i \pi\alpha_{k}}{2}}, & z \in Q_{-,k}^{L}, \\
(z-t_{k})^{\frac{\alpha_{k}}{2}}e^{\frac{ i \pi\alpha_{k}}{2}}, & z \in Q_{-,k}^{R}, \\
\end{array} \right.
\end{equation}
and where $Q_{+,k}^{R}$, $Q_{+,k}^{L}$, $Q_{-,k}^{L}$ and $Q_{-,k}^{R}$ are the following four quadrant centred at $t_{k}$:
\begin{align*}
& Q_{\pm,k}^{R} = \{ z \in \mathcal{D}_{t_{k}}: \mp \Re f_{t_{k}}(z) > 0 \mbox{, } \Im f_{t_{k}}(z) >0 \}, \qquad Q_{\pm,k}^{L} = \{ z \in \mathcal{D}_{t_{k}}: \mp \Re f_{t_{k}}(z) > 0 \mbox{, } \Im f_{t_{k}}(z) <0 \}.
\end{align*}
The solution of the above RH problem is given in terms of the confluent hypergeometric model RH problem, denoted $\Phi_{\mathrm{HG}}(z;\alpha_{k},\beta_{k})$ (see Section \ref{Section:Appendix} for the definition and properties of $\Phi_{\mathrm{HG}}(z;\alpha_{k},\beta_{k})$). We obtain
\begin{equation}\label{lol 36}
P^{(t_{k})}(z) = E_{t_{k}}(z)\Phi_{\mathrm{HG}}(nf_{t_{k}}(z);\alpha_{k},\beta_{k})\widetilde{W}_{k}(z)^{-\sigma_{3}}e^{-n\xi(z)\sigma_{3}}e^{-\frac{W(z)}{2}\sigma_{3}}\omega_{k}(z)^{-\frac{\sigma_{3}}{2}},
\end{equation}
where $E_{t_{k}}$ is given by
\begin{equation}\label{E in D_t}
E_{t_{k}}(z) = P^{(\infty)}(z) \omega_{k}(z)^{\frac{\sigma_{3}}{2}}e^{\frac{W(z)}{2}\sigma_{3}} \widetilde{W}_{k}(z)^{\sigma_{3}}\hspace{-0.08cm} \left\{ \hspace{-0.18cm} \begin{array}{l l}
e^{ \frac{i\pi\alpha_{k}}{4}\sigma_{3}}e^{-i\pi\beta_{k} \sigma_{3}}, \hspace{-0.2cm} & z \in Q_{+,k}^{R} \\
e^{-\frac{i\pi\alpha_{k}}{4}\sigma_{3}}e^{-i\pi\beta_{k}\sigma_{3}}, \hspace{-0.2cm} & z \in Q_{+,k}^{L} \\
e^{\frac{i\pi\alpha_{k}}{4}\sigma_{3}}\begin{pmatrix}
0 & 1 \\ -1 & 0
\end{pmatrix} , \hspace{-0.2cm} & z \in Q_{-,k}^{L} \\
e^{-\frac{i\pi\alpha_{k}}{4}\sigma_{3}}\begin{pmatrix}
0 & 1 \\ -1 & 0
\end{pmatrix} , \hspace{-0.2cm} & z \in Q_{-,k}^{R} \\
\end{array} \hspace{-0.2cm} \right\} \hspace{-0.08cm} e^{n\xi_{+}(t_{k})\sigma_{3}} (nf_{t_{k}}(z))^{\beta_{k}\sigma_{3}}.
\end{equation} 
From the jumps of $P^{(\infty)}$ \eqref{jumps for P inf} and the definition of $\widetilde{W}_{k}$ \eqref{def of Wk}, we can check that $E_{t_{k}}$ has no jumps at all in $\mathcal{D}_{t_{k}}$. Furthermore, since $P^{(\infty)}(z) = \bigO(1)(z-t_{k})^{-(\frac{\alpha_{k}}{2}+\beta_{k})\sigma_{3}}$ as $z\to t_{k}$, we conclude from \eqref{E in D_t} that $E_{t_{k}}(z) = \bigO(1)$ as $z \to t_{k}$. As a result, $E_{t_{k}}$ is analytic in the disk $\mathcal{D}_{t_{k}}$. Note also from \eqref{E in D_t} that $E_{t_{k}}(z) = \bigO(n^{|\Re \beta_{k}|})$ as  $n \to \infty$, uniformly for $z \in \mathcal{D}_{t_{k}}$. Since $P^{(t_{k})}$ and $S$ have exactly the same jumps on $(\mathbb{R}\cup \gamma_{+}\cup \gamma_{-})\cap \mathcal{D}_{t_{k}}$, $S(z)P^{(t_{k})}(z)^{-1}$ is analytic in $\mathcal{D}_{t_{k}} \setminus \{t_{k}\}$. As $z \to t_{k}$, $z$ outside the lenses, by condition (d) in the RH problem for $S$, by \eqref{lol 35} and by \eqref{lol 36}, $S(z)P^{(t_{k})}(z)^{-1}$ behaves as $\bigO(\log(z-t_{k}))$ if $\Re \alpha_{k} = 0$, as $\bigO(1)$ if $\Re \alpha_{k} > 0$, and as $\bigO((z-t_{k})^{\Re \alpha_{k}})$ if $\Re \alpha_{k} < 0$. Thus, in all the cases, the singularity of $S(z)P^{(t_{k})}(z)^{-1}$ at $z = t_{k}$ is removable and condition (d) of the RH problem for $P^{(t_{k})}$ is verified.

The value of $E_{t_{k}}(t_{k})$ can be obtained by taking the limit $z \to t_{k}$ in \eqref{E in D_t}. Using \eqref{asymptotics for f in D_t}, \eqref{asymptotics for D_alpha in D_t} and \eqref{asymptotics for D_beta in D_t}, it gives
\begin{multline*}
E_{t_{k}}(t_{k}) = \frac{D_{\infty}^{\sigma_{3}}}{2 \sqrt[4]{1-t_{k}^{2}}}\begin{pmatrix}
e^{-\frac{\pi i}{4}}\sqrt{1+t_{k}} + e^{\frac{\pi i}{4}}\sqrt{1-t_{k}} & i \left( e^{-\frac{\pi i}{4}}\sqrt{1+t_{k}} - e^{\frac{\pi i}{4}}\sqrt{1-t_{k}} \right) \\
-i \left( e^{-\frac{\pi i}{4}}\sqrt{1+t_{k}} - e^{\frac{\pi i}{4}}\sqrt{1-t_{k}} \right) & e^{-\frac{\pi i}{4}}\sqrt{1+t_{k}} + e^{\frac{\pi i}{4}}\sqrt{1-t_{k}}
\end{pmatrix} \Lambda_{k}^{\sigma_{3}},
\end{multline*}
where
\begin{equation}\label{def Lambda}
\Lambda_{k} = e^{\frac{W(t_{k})}{2}}D_{W,+}(t_{k})^{-1}e^{i \frac{\lambda_{k}}{2} } (4\pi \psi(t_{k})n(1-t_{k}^{2})^{3/2})^{\beta_{k}} \prod_{j\neq k} T_{kj}^{-\beta_{j}}, 
\end{equation}
and
\begin{equation}
\lambda_{k} = \mathcal{A} \arccos t_{k} - \frac{\pi}{2}\alpha_{k} - \sum_{j=k+1}^{m} \pi \alpha_{j} + 2\pi n \int_{t_{k}}^{1} \rho(s)ds.
\end{equation}

A direct calculation, using equation \eqref{Asymptotics HG} in the appendix, shows that as $n \to \infty$, uniformly for $z \in \partial \mathcal{D}_{t_{k}}$, $P^{(t_{k})}(z)P^{(\infty)}(z)^{-1}$ admits an expansion in the inverse power of $n$ multiplied by $n^{2|\Re \beta_{k}|}$. We will need explicitly the $n^{-1}$ term:
\begin{equation}\label{asymptotics on the disk D_t}
P^{(t_{k})}(z)P^{(\infty)}(z)^{-1} = I + \frac{v_{k}}{n f_{t_{k}}(z)} E_{t_{k}}(z) \begin{pmatrix}
-1 & \tau(\alpha_{k},\beta_{k}) \\ - \tau(\alpha_{k},-\beta_{k}) & 1
\end{pmatrix}E_{t_{k}}(z)^{-1} + \bigO (n^{-2+2|\Re\beta_{k}|}),
\end{equation}
where $v_{k} = \beta_{k}^{2}-\frac{\alpha_{k}^{2}}{4}$, and $\tau(\alpha_{k},\beta_{k})$ is defined in equation \eqref{def of tau}.
\subsection{Local parametrix near $1$}\label{subsection: airy parametrix near 1}
Consider a fixed disk $\mathcal{D}_{1}$, centered at $1$, of radius smaller than $\delta/3$ and such that $\mathcal{D}_{1} \subset U_W$. Inside $\mathcal{D}_{1}$, the local parametrix $P^{(1)}$ is a good approximation of $S$ and must satisfy the following RH problem.
\subsubsection*{RH problem for $P^{(1)}$}
\begin{itemize}
\item[(a)] $P^{(1)} : \mathcal{D}_{1}\setminus (\mathbb{R} \cup \gamma_{+} \cup \gamma_{-}) \to \mathbb{C}^{2\times 2}$ is analytic.
\item[(b)] $P^{(1)}$ has the following jumps:
\begin{align}
& P_{+}^{(1)}(z) = P_{-}^{(1)}(z)\begin{pmatrix}
0 & e^{W(z)}\omega(z) \\ -e^{-W(z)}\omega(z)^{-1} & 0
\end{pmatrix}, & & \mbox{ if } z \in (-\infty,1)\cap \mathcal{D}_{1}, \\
& P_{+}^{(1)}(z) = P_{-}^{(1)}(z)\begin{pmatrix}
1 & e^{W(z)}\omega(z)e^{2n\xi(z)} \\ 0 & 1
\end{pmatrix}, & & \mbox{ if } z \in (1,\infty)\cap \mathcal{D}_{1}, \\
& P_{+}^{(1)}(z) = P_{-}^{(1)}(z)\begin{pmatrix}
1 & 0 \\ e^{-W(z)}\omega(z)^{-1}e^{-2n \xi(z)} & 1
\end{pmatrix}, & & \mbox{ if } z \in (\gamma_{+} \cup \gamma_{-}) \cap \mathcal{D}_{1}.
\end{align}
\item[(c)] As $n \to \infty$, we have $P^{(1)}(z) = (I + \bigO(n^{-1}))P^{(\infty)}(z)$ uniformly for $z \in \partial \mathcal{D}_{1}$.
\item[(d)] As $z$ tends to $1$, we have $P^{(1)}(z) = \bigO(1)$.
\end{itemize}
The construction of this local parametrix is now standard (see e.g. \cite{DKMVZ1}). We denote
\begin{equation}
f_{1}(z) = \left( -\frac{3}{2}\xi(z) \right)^{2/3} = \left( \frac{3\pi}{2} \int_{1}^{z} \psi(s)\sqrt{s^{2}-1}ds \right)^{2/3}.
\end{equation}
This is a conformal map from $\mathcal{D}_{1}$ to a neighbourhood of $0$ and
\begin{equation}\label{asymptotics f near 1}
f_{1}(z) = (\sqrt{2}\pi \psi(1))^{2/3}(z-1)\left( 1+\frac{1}{10}\left( 1+4\frac{\psi^{\prime}(1)}{\psi(1)} \right)(z-1) + \bigO((z-1)^{2}) \right), \qquad \mbox{as } z \to 1.
\end{equation}
We choose the lenses such that $f_{1}(\gamma_{+}\cap \mathcal{D}_{1}) \subset e^{\frac{2\pi i}{3}}\mathbb{R}^{+}$ and $f_{1}(\gamma_{-}\cap \mathcal{D}_{1}) \subset e^{-\frac{2\pi i}{3}}\mathbb{R}^{+}$.
The solution of the above RH problem is given by
\begin{equation}
P^{(1)}(z) = E_{1}(z) \Phi_{\mathrm{Ai}}(n^{2/3}f_{1}(z))\omega(z)^{-\frac{\sigma_{3}}{2}}e^{-n\xi(z)\sigma_{3}}e^{-\frac{W(z)}{2}\sigma_{3}},
\end{equation}
where 
\begin{equation}
E_{1}(z) = P^{(\infty)}(z)e^{\frac{W(z)}{2}\sigma_{3}}\omega(z)^{\frac{\sigma_{3}}{2}}N^{-1}f_{1}(z)^{\frac{\sigma_{3}}{4}}n^{\frac{\sigma_{3}}{6}}, \qquad N = \frac{1}{\sqrt{2}}\begin{pmatrix}
1 & i \\ i & 1
\end{pmatrix},
\end{equation}
and $\Phi_{\mathrm{Ai}}(z)$ is the solution to the Airy model RH problem presented in the Appendix, see Section \ref{Section:Appendix}. One can check that $E_{1}$ has no jumps in $\mathcal{D}_{1}$ and is bounded as $z\to 1$, and thus $E_{1}$ is analytic in $\mathcal{D}_{1}$. A direct calculation, using \eqref{Asymptotics Airy}, shows that as $n \to \infty$, uniformly for $z \in \partial \mathcal{D}_{1}$, $P^{(1)}(z)P^{(\infty)}(z)^{-1}$ admits an expansion in the inverse power of $n$. The explicit forms of the first terms are given in the following expression
\begin{equation}\label{asymptotics on the disk D_1}
P^{(1)}(z)P^{(\infty)}(z)^{-1} = I + \frac{P^{(\infty)}(z)e^{\frac{W(z)}{2}\sigma_{3}}\omega(z)^{\frac{\sigma_{3}}{2}}}{8n f_{1}(z)^{3/2}}\begin{pmatrix}
\frac{1}{6} & i \\ i & -\frac{1}{6}
\end{pmatrix}\omega(z)^{-\frac{\sigma_{3}}{2}}e^{-\frac{W(z)}{2}\sigma_{3}}P^{(\infty)}(z)^{-1} + \bigO(n^{-2}).
\end{equation}

\subsection{Local parametrix near $-1$}
Consider a fixed disk $\mathcal{D}_{-1}$, centered at $-1$, of radius smaller than $\delta/3$ and such that $\mathcal{D}_{-1} \subset U_W$. Inside $\mathcal{D}_{-1}$, the local parametrix $P^{(-1)}$ is a good approximation of $S$ and must satisfy the following RH problem.
\subsubsection*{RH problem for $P^{(-1)}$}
\begin{itemize}
\item[(a)] $P^{(-1)} : \mathcal{D}_{-1}\setminus (\mathbb{R} \cup \gamma_{+} \cup \gamma_{-}) \to \mathbb{C}^{2\times 2}$ is analytic.
\item[(b)] $P^{(-1)}$ has the following jumps:
\begin{align}
& P_{+}^{(-1)}(z) = P_{-}^{(-1)}(z)\begin{pmatrix}
1 & e^{W(z)}\omega(z) e^{n(\xi_{+}(z)+\xi_{-}(z))} \\ 0 & 1
\end{pmatrix}, & & \mbox{ if } z \in (-\infty,-1)\cap \mathcal{D}_{-1}, \\
& P_{+}^{(-1)}(z) = P_{-}^{(-1)}(z)\begin{pmatrix}
0 & e^{W(z)}\omega(z) \\ -e^{-W(z)}\omega(z)^{-1} & 0
\end{pmatrix}, & & \mbox{ if } z \in (-1,\infty)\cap \mathcal{D}_{-1}, \\
& P_{+}^{(-1)}(z) = P_{-}^{(-1)}(z)\begin{pmatrix}
1 & 0 \\ e^{-W(z)}\omega(z)^{-1}e^{-2n \xi(z)} & 1
\end{pmatrix}, & & \mbox{ if } z \in (\gamma_{+} \cup \gamma_{-}) \cap \mathcal{D}_{-1}.
\end{align}
\item[(c)] As $n \to \infty$, we have $P^{(-1)}(z) = (I + \bigO(n^{-1}))P^{(\infty)}(z)$ uniformly for $z \in \partial \mathcal{D}_{-1}$.
\item[(d)] As $z$ tends to $-1$, we have $P^{(-1)}(z) = \bigO(1)$.
\end{itemize}
The construction of this local parametrix around $-1$ is very similar to the construction done in subsection \ref{subsection: airy parametrix near 1} for the local parametrix around 1. We denote
\begin{equation}
f_{-1}(z) = -\left( -\frac{3}{2} \left\{ \begin{array}{l l}
\xi(z)-\xi_{+}(-1), & \mbox{Im}z>0, \\
\xi(z)-\xi_{-}(-1), & \mbox{Im}z<0, \\
\end{array} \right\} \right)^{2/3} = \left( \frac{3\pi}{2} \int_{-1}^{z} \psi(s)\sqrt{1-s^{2}}ds \right)^{2/3}.
\end{equation}
This is a conformal map from $\mathcal{D}_{-1}$ to a neighbourhood of $0$ and
\begin{equation}\label{asymptotics of f near -1}
f_{-1}(z) = (\sqrt{2}\pi \psi(-1))^{2/3}(z+1)\left( 1-\frac{1}{10}\left( 1-4\frac{\psi^{\prime}(-1)}{\psi(-1)} \right)(z+1) + \bigO((z+1)^{2}) \right), \quad \mbox{as } z \to -1.
\end{equation}
We choose the lenses such that $-f_{-1}(\gamma_{+}\cap \mathcal{D}_{-1}) \subset e^{-\frac{2\pi i}{3}}\mathbb{R}^{+}$ and $-f_{-1}(\gamma_{-}\cap \mathcal{D}_{-1}) \subset e^{\frac{2\pi i}{3}}\mathbb{R}^{+}$. The solution of the above RH problem is given by
\begin{equation}
P^{(-1)}(z) = E_{-1}(z) \sigma_{3} \Phi_{\mathrm{Ai}}(-n^{2/3}f_{-1}(z))\sigma_{3}\omega(z)^{-\frac{\sigma_{3}}{2}}e^{-n\xi(z)\sigma_{3}}e^{-\frac{W(z)}{2}\sigma_{3}},
\end{equation}
where $E_{-1}$ is analytic in $\mathcal{D}_{-1}$ and is given by
\begin{equation}
E_{-1}(z) = (-1)^{n} P^{(\infty)}(z)e^{\frac{W(z)}{2}\sigma_{3}}\omega(z)^{\frac{\sigma_{3}}{2}}N (-f_{-1}(z))^{\frac{\sigma_{3}}{4}}n^{\frac{\sigma_{3}}{6}}.
\end{equation}
Similarly to \eqref{asymptotics on the disk D_1}, one shows with a direct calculation together with \eqref{Asymptotics Airy}, that as $n \to \infty$, uniformly for $z \in \partial \mathcal{D}_{-1}$, $P^{(-1)}(z)P^{(\infty)}(z)^{-1}$ admits an expression in the inverse power of $n$. The first terms are given by
\begin{equation}\label{asymptotics on the disk D_-1}
P^{(-1)}(z)P^{(\infty)}(z)^{-1} = I + \frac{P^{(\infty)}(z)e^{\frac{W(z)}{2}\sigma_{3}}\omega(z)^{\frac{\sigma_{3}}{2}}}{8n (-f_{-1}(z))^{3/2}}\begin{pmatrix}
\frac{1}{6} & -i \\ -i & -\frac{1}{6}
\end{pmatrix}\omega(z)^{-\frac{\sigma_{3}}{2}}e^{-\frac{W(z)}{2}\sigma_{3}}P^{(\infty)}(z)^{-1} + \bigO ( n^{-2} ).
\end{equation}

\subsection{Small norm RH problem} \label{Subsection: small-norm RH problem}
Let $\mathcal{D} = \mathcal{D}_{-1} \cup \mathcal{D}_{1} \cup \bigcup_{j=1}^{m} \mathcal{D}_{t_{j}}$ be the union of the disks, whose boundaries are oriented in the clockwise direction as shown in Figure \ref{Fig:R}, and let $P$ be defined on $\mathcal{D}$ such that $\left. P \right|_{\mathcal{D}_{a}} = P^{(a)}$, where $a = -1,t_{1},...,t_{m},1$.
We define 
\begin{equation}\label{def of R}
R(z) = R(z;\vec{\alpha},\vec{\beta},V,W) = \left\{ \begin{array}{l l}
S(z) P^{(\infty)}(z)^{-1}, & z \in \mathbb{C}\setminus \mathcal{D}, \\
S(z) P(z)^{-1}, & z \in \mathcal{D}.
\end{array} \right.
\end{equation}
\begin{figure}[t]
    \begin{center}
    \setlength{\unitlength}{1truemm}
    \begin{picture}(100,30)(28,30)
        \put(12,40){\line(1,0){21}}
        \put(25,40){\thicklines\vector(1,0){.0001}}
        \put(103,40){\line(1,0){20}}
        \put(115,40){\thicklines\vector(1,0){.0001}}
        \put(49.6,47.5){\thicklines\vector(1,0){.0001}}
        \put(49.6,32.3){\thicklines\vector(1,0){.0001}}
        \put(69.6,47.5){\thicklines\vector(1,0){.0001}}
        \put(69.6,32.3){\thicklines\vector(1,0){.0001}}
        \put(89.6,47.5){\thicklines\vector(1,0){.0001}}
        \put(89.6,32.3){\thicklines\vector(1,0){.0001}}
        \put(38,39.9){\thicklines\circle{10}} \put(35.7,39){$-1$} 
        \put(39.3,44.9){\thicklines\vector(1,0){.0001}}
        \put(58,39.9){\thicklines\circle{10}} \put(57.35,39){$t_{1}$}
        \put(59.3,44.9){\thicklines\vector(1,0){.0001}}
        \put(78,39.9){\thicklines\circle{10}} \put(77,39){$t_{m}$}
        \put(79.3,44.9){\thicklines\vector(1,0){.0001}}
        \put(98,39.9){\thicklines\circle{10}} \put(98,39){$1$}
        \put(99.3,44.9){\thicklines\vector(1,0){.0001}}
        \qbezier(41.3,43.3)(48,52)(54.7,43.3)
        \qbezier(41.3,36.5)(48,28)(54.7,36.5)
        
        \qbezier(61.3,43.3)(68,52)(74.7,43.3)
        \qbezier(61.3,36.5)(68,28)(74.7,36.5)

        \qbezier(81.3,43.3)(88,52)(94.7,43.3)
        \qbezier(81.3,36.5)(88,28)(94.7,36.5)

        \put(65.5,53){$\Sigma_R$}
    \end{picture}
    \caption{Jump contours for the RH problem for $R$ with $m=2$. The circles are oriented in the clockwise direction. \label{Fig:R}}
\end{center}
\end{figure}
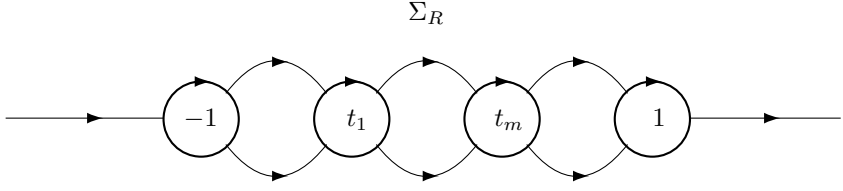
\hspace{-0.08cm}Since $P$ and $S$ have exactly the same jumps inside $\mathcal{D}$, $R$ is analytic in $\mathcal{D}\setminus\{-1,t_{1},...,t_{m},1 \}$. Also, from the behaviours of $P$ and $S$ near $-1$, $t_{1}$, ..., $t_{m}$ and $1$, these points are removable singularities of $R$. Thus, $R$ satisfies the following RH problem:
\subsubsection*{RH problem for $R$}
\begin{itemize}
\item[(a)] $R : \mathbb{C}\setminus \Sigma_{R} \to \mathbb{C}^{2\times 2}$ is analytic, where $\Sigma_{R}$ is shown in Figure \ref{Fig:R}.
\item[(b)] $R$ has the jumps $R_{+}(z) = R_{-}(z)J_{R}(z)$ for $z \in \Sigma_{R}$, where $J_{R}(z)$ is given in \eqref{asymptotics on the disk D_t}, \eqref{asymptotics on the disk D_1} and \eqref{asymptotics on the disk D_-1} for $z$ on respectively $\partial \mathcal{D}_{t_{k}}$, $\partial\mathcal{D}_{1}$, and $\partial\mathcal{D}_{-1}$. As $n \to \infty$ we have
\begin{align}
& J_{R}(z) = I+\bigO(e^{-cn}), & & \mbox{ uniformly for } z \in (\gamma_{+} \cup \gamma_{-}\cup \mathbb{R})\setminus (\mathcal{S}\cup \mathcal{D}), \label{J_R 1} \\
& J_{R}(z) = I + \bigO (n^{-1} ), & & \mbox{ uniformly for } z \in \partial \mathcal{D}_{-1} \cup \partial \mathcal{D}_{1}, \\
& J_{R}(z) = I + \bigO (n^{-1+2|\Re\beta_{k}|}), & & \mbox{ uniformly for } z \in \partial \mathcal{D}_{t_{k}}, k = 1,...,m,
\end{align}
where $c>0$ is a constant.
\item[(c)] As $z \to \infty$, we have $R(z) = I+\bigO(z^{-1})$.
\end{itemize}
Let us define
\begin{equation}
\mathcal{P}_{\beta}^{(\frac{1}{2})} = \{ \vec{\beta}\in \mathbb{C}^{m} :  \Re \beta_{j} \in (-\tfrac{1}{2},\tfrac{1}{2}), \mbox{ for all } j = 1,...,m \}.
\end{equation}
From the standard theory for small-norm RH problems \cite{DKMVZ2,DKMVZ1}, if $\Omega$ is a compact subset of $\mathcal{P}_{\alpha}\times \mathcal{P}_{\beta}^{(\frac{1}{2})}$, there exists $n_{\star} = n_{\star}(\Omega)$ such that $R$ exists for \textit{all} $n \geq n_{\star}$, for \textit{all} $(\vec{\alpha},\vec{\beta}) \in \Omega$ and for \textit{all} $z \in \mathbb{C}\setminus\Sigma_{R}$. Furthermore, for any $r \in \mathbb{N}$, as $n \to \infty$ we have
\begin{align}\label{asymptotics for R}
& R(z) = \sum_{j=0}^{r} R^{(j)}(z)n^{-j} + R_{R}^{(r+1)}(z)n^{-r-1}, \qquad R^{(0)}(z) \equiv I, \\ 
& R^{(j)}(z) = \bigO ( n^{2\beta_{\max}} ) \quad R^{(j)}(z)^{\prime} = \bigO ( n^{2\beta_{\max}} ), \quad R_{R}^{(r+1)}(z) = \bigO ( n^{2\beta_{\max}} ) \quad R_{R}^{(r+1)}(z)^{\prime} = \bigO ( n^{2\beta_{\max}} ),   \nonumber 
\end{align}
uniformly for $z \in \mathbb{C}\setminus \Sigma_{R}$, uniformly for $(\vec{\alpha},\vec{\beta}) \in \Omega$ and uniformly in $t_{1},...,t_{m}$, as long as there exists $\delta > 0$ independent of $n$ such that
\begin{equation}\label{lol 44}
\min_{j\neq k}\{ |t_{j}-t_{k}|,|t_{j}-1|,|t_{j}+1|\} \geq \delta.
\end{equation}
We can find explicit expressions for $R^{(1)}$, $R^{(2)}$,..., $R^{(r)}$ by induction. The factor $n^{\beta_{k}\sigma_{3}}$ in the expression \eqref{E in D_t} for $E_{t_{k}}$ induces factors of the forms $n^{\pm 2\beta_{k}}$ in the entries of $J_{R}^{(j)}(z)$, for $j \geq 1$ and $z \in \partial \mathcal{D}_{t_{k}}$, as can be seen in \eqref{asymptotics on the disk D_t}. Thus, from \eqref{J_R 1}, \eqref{asymptotics on the disk D_t}, \eqref{asymptotics on the disk D_1} and \eqref{asymptotics on the disk D_-1}, the jumps of $R$ admit an expansion as $n \to \infty$ in powers of $n^{-1}$ multiplied by terms of order $n^{2\beta_{\max}}$. As $n \to \infty$, we have
\begin{align}\label{asymptotics for the jumps of R}
& J_{R}(z) = \sum_{j=0}^{r} J_{R}^{(j)}(z)n^{-j} + \bigO (n^{-r-1+2 \beta_{\max}} ), & & \qquad \mbox{uniformly for } z \in \partial \mathcal{D},
\end{align}
where $J_{R}^{(0)}(z) \equiv I$ and for $j \geq 1$, $J_{R}^{(j)}(z) = \bigO(n^{2 \beta_{\max}})$ as $n \to \infty$. From a perturbative analysis of the small norm RH problem for $R$, for every $j \geq 1$
, $R^{(j)}(z)$ is analytic on $\mathbb{C}\setminus \partial \mathcal{D}$ and satisfies 
\begin{align}
& R^{(j)}_{+}(z) = R^{(j)}_{-}(z) + \sum_{\ell=1}^{j} R_{-}^{(j-\ell)}(z)J_{R}^{(\ell)}(z), & & z \in \partial \mathcal{D}, \\
& R^{(j)}(z) = \bigO(z^{-1}), & & \mbox{ as } z \to \infty.
\end{align}
Therefore, by the Sokhotsky-Plemelj formula, $R^{(j)}$ is simply given by
\begin{equation}\label{lol 43}
R^{(j)}(z) = \frac{1}{2\pi i} \int_{\partial \mathcal{D}} \frac{\sum_{\ell=1}^{j} R_{-}^{(j-\ell)}(s)J_{R}^{(\ell)}(s)}{s-z}ds, \qquad j \geq 1,
\end{equation}
where we recall that the orientation on the disks is clockwise, as shown in Figure \ref{Fig:R}. 

\vspace{0.2cm}\hspace{-0.65cm}
Since $R(z)$ exists for \textit{all} $n \geq n_\star$ and for \textit{all} $(\vec{\alpha},\vec{\beta}) \in \Omega$, by inverting the transformations $Y \mapsto T \mapsto S \mapsto R$, we have constructed a solution for the RH problem for $Y$ that exists for \textit{all} $n \geq n_\star$ and for \textit{all} $(\vec{\alpha},\vec{\beta}) \in \Omega$. From \eqref{Y definition} and the determinantal representation for orthogonal polynomials (see \eqref{lol 33}), this implies that the polynomials $\kappa_{n}^{-1}p_{n}$ and $\kappa_{n-1}p_{n-1}$ are analytic functions of $(\vec{\alpha},\vec{\beta}) \in \Omega$ (since they exist, they can not have poles). Therefore, $R(z)$ is also an analytic function of $(\vec{\alpha},\vec{\beta}) \in \Omega$ (some extra calculations are required to verify this at $z = t_{1},...,t_{m}$, which are similar to those in Subsection \ref{Subsection:Y tilde regularized}, and we omit them here).
Furthermore, by \eqref{lol 43}, and since the asymptotic series of the confluent hypergeometric function that appears in the jumps for $R$ \eqref{asymptotics on the disk D_t} (see also \eqref{Asymptotics HG}) is differentiable in $\alpha_{k}$'s and $\beta_{k}$'s, the matrices $R^{(j)}$, $j = 1,...,r,$ are also analytic functions of $(\vec{\alpha},\vec{\beta}) \in \Omega$. Hence, the terms in \eqref{asymptotics for R} are both uniform and analytic in $(\vec{\alpha},\vec{\beta}) \in \Omega$. Furthermore, since factors of the forms $n^{\pm 2 \beta_{k}}$, $k = 1,...,m$, appear in the entries of $J_{R}^{(j)}$ for $j \geq 1$ (see \eqref{asymptotics for the jumps of R} and the comment above), as $n \to \infty$ we have
\begin{equation}\label{estimates for R with beta}
\partial_{\beta_{k}}R^{(j)}(z) = \bigO (  n^{2\beta_{\max}} \log n), \qquad \partial_{\beta_{k}}R_{R}^{(r+1)}(z) = \bigO (  n^{2\beta_{\max}} \log n),
\end{equation}
and these asymptotics are also uniform in $z$, $\vec{\alpha}$, $\vec{\beta}$ and $t_{1}$,...,$t_{m}$ as in \eqref{asymptotics for R}. 
\begin{remark}\label{Remark: R}
We will need the following further properties of $R$ for particular values of the parameters $V$ and $W$.
\begin{itemize}
\item[(a)] If $V$ is replaced by $V_{s}$ (see \eqref{deformation parameter potential} for the definition of $V_{s}$)  and $W$ by $0$, the asymptotics \eqref{asymptotics for R} also hold uniformly for $s \in [0,1]$. This follows from standard arguments for small-norm Riemann-Hilbert problems and a detailed analysis of the Cauchy operator associated to $R$, see e.g. \cite[Lemma 4.35]{BerWebbWong} for a similar situation. Furthermore, the determinantal representation for orthogonal polynomials implies that $R$ is an analytic function of $s \in (0,1)$. 
\item[(b)] If $W$ is replaced by $W_{t}$ (see \eqref{deformation parameter Tau} for the definition of $W_{t}$), the asymptotics \eqref{asymptotics for R} also hold uniformly for $t \in [0,1]$, see e.g. \cite[Lemma 4.35]{BerWebbWong} for a similar situation. Furthermore, the determinantal representation for orthogonal polynomials implies that $R$ is analytic for $t \in [0,1]$.
\end{itemize}
\end{remark}
The goal for the rest of this section is to compute $R^{(1)}(z)$ from \eqref{lol 43}. Note that the expression for $J_{R}^{(1)}(z)$ for $z\in \partial \mathcal{D}$ can be analytically continued on $\overline{\mathcal{D}}$, except at $-1,t_{1},...,t_{m}$ and $1$, where the expressions in \eqref{asymptotics on the disk D_t}, \eqref{asymptotics on the disk D_1} and \eqref{asymptotics on the disk D_-1} admit poles. These poles are of order $2$ at $-1$ and $1$, and of order $1$ at $t_{k}$, $k=1,...,m$. Therefore, for $z$ outside the disks, $R^{(1)}(z)$ is given by
\begin{equation}\label{R^{(1)}}
\begin{array}{r c l}
\displaystyle R^{(1)}(z) & = & \displaystyle \sum_{j=1}^{m} \frac{1}{z-t_{j}} \mbox{Res}\big(J_{R}^{(1)}(s),s=t_{j}\big) \\
 &  & \displaystyle +\frac{1}{z-1} \mbox{Res}\big(J_{R}^{(1)}(s),s=1\big)+ \frac{1}{(z-1)^{2}} \mbox{Res}\big((s-1)J_{R}^{(1)}(s),s=1\big) \\
  &  & \displaystyle +\frac{1}{z+1} \mbox{Res}\big(J_{R}^{(1)}(s),s=-1\big)+ \frac{1}{(z+1)^{2}} \mbox{Res}\big((s+1)J_{R}^{(1)}(s),s=-1\big).  \\
\end{array}
\end{equation}
We can compute these residues explicitly in the case $W\equiv 0$. For the residues at $t_{k}$, $k \in \{1,...,m\}$, using \eqref{asymptotics on the disk D_t}, \eqref{asymptotics for f in D_t}, \eqref{asymptotics for D_alpha in D_t} and \eqref{asymptotics for D_beta in D_t}, we have
\begin{equation}
\mbox{Res}\big(J_{R}^{(1)}(z),z=t_{k}\big) = \frac{v_{k} D_{\infty}^{\sigma_{3}}}{2\pi \psi(t_{k})(1-t_{k}^{2})}\begin{pmatrix}
t_{k} + \widetilde{\Lambda}_{I,k} & -i - i\widetilde{\Lambda}_{R,2,k} \\
-i + i \widetilde{\Lambda}_{R,1,k} & - t_{k} - \widetilde{\Lambda}_{I,k}
\end{pmatrix}D_{\infty}^{-\sigma_{3}},
\end{equation}
where
\begin{align}
& \widetilde{\Lambda}_{I,k} = \frac{\tau(\alpha_{k},\beta_{k})\Lambda_{k}^{2}-\tau(\alpha_{k},-\beta_{k})\Lambda_{k}^{-2}}{2i}, \\[0.3cm]
& \widetilde{\Lambda}_{R,1,k} = \frac{\tau(\alpha_{k},\beta_{k})\Lambda_{k}^{2}e^{i\arcsin t_{k}}+\tau(\alpha_{k},-\beta_{k})\Lambda_{k}^{-2}e^{-i\arcsin t_{k}}}{2}, \\[0.3cm] & \widetilde{\Lambda}_{R,2,k} = \frac{\tau(\alpha_{k},\beta_{k})\Lambda_{k}^{2}e^{-i\arcsin t_{k}}+\tau(\alpha_{k},-\beta_{k})\Lambda_{k}^{-2}e^{i\arcsin t_{k}}}{2}.
\end{align}
Note the relation
\begin{equation}\label{connection Lambdas}
\widetilde{\Lambda}_{R,1,k} - \widetilde{\Lambda}_{R,2,k} = -2t_{k} \widetilde{\Lambda}_{I,k}.
\end{equation}
For the two residues at $1$, using \eqref{asymptotics on the disk D_1}, \eqref{asymptotics f near 1}, \eqref{asymptotics for D_alpha in D_1} and \eqref{asymptotics for D_beta in D_1}, we obtain
\begin{equation}
\mbox{Res}\big((z-1)J_{R}^{(1)}(z),z=1\big) = \frac{5}{2^{5} 3 \pi \psi(1)} D_{\infty}^{\sigma_{3}} \begin{pmatrix}
-1 & i \\ i & 1
\end{pmatrix} D_{\infty}^{-\sigma_{3}},
\end{equation}
and
\begin{multline*}
\mbox{Res}\big(J_{R}^{(1)}(z),z=1\big) = \frac{D_{\infty}^{\sigma_{3}}}{2^{5}\pi \psi(1)} \times \\[0.3cm] \hspace{0.5cm} \begin{pmatrix}
-2(\mathcal{A}-\widetilde{\mathcal{B}_{1}})^{2} + 1 + \frac{\psi^{\prime}(1)}{\psi(1)}  & 2i \left( (\mathcal{A}-\widetilde{\mathcal{B}_{1}})^{2} + 2 (\mathcal{A}-\widetilde{\mathcal{B}_{1}}) + \frac{2}{3} - \frac{1}{2}\frac{\psi^{\prime}(1)}{\psi(1)} \right) \\
2i \left( (\mathcal{A}-\widetilde{\mathcal{B}_{1}})^{2} - 2 (\mathcal{A}-\widetilde{\mathcal{B}_{1}}) + \frac{2}{3} - \frac{1}{2}\frac{\psi^{\prime}(1)}{\psi(1)} \right) & 2(\mathcal{A}-\widetilde{\mathcal{B}_{1}})^{2} - 1 -\frac{\psi^{\prime}(1)}{\psi(1)}
\end{pmatrix}D_{\infty}^{-\sigma_{3}}.
\end{multline*}
Finally, for the two residues at $-1$, using \eqref{asymptotics on the disk D_-1}, \eqref{asymptotics of f near -1}, \eqref{asymptotics for D_alpha in D_-1} and \eqref{asymptotics for D_beta in D_-1}, we obtain
\begin{equation}
\mbox{Res}\big((z+1)J_{R}^{(1)}(z),z=-1\big) = \frac{5}{2^{5} 3 \pi \psi(-1)} D_{\infty }^{\sigma_{3}} \begin{pmatrix}
-1 & -i \\ -i & 1
\end{pmatrix} D_{\infty}^{-\sigma_{3}},
\end{equation}
and
\begin{align*}
& \mbox{Res}\big(J_{R}^{(1)}(z),z=-1\big) = \frac{D_{\infty}^{\sigma_{3}}}{2^{5}\pi \psi(-1)} \times \\[0.3cm] 
& \hspace{0.0cm} \begin{pmatrix}
2(\mathcal{A}+\widetilde{\mathcal{B}}_{-1})^{2}-1 + \frac{\psi^{\prime}(-1)}{\psi(-1)} & 2i \left( (\mathcal{A}+\widetilde{\mathcal{B}}_{-1})^{2} + 2(\mathcal{A}+\widetilde{\mathcal{B}}_{-1})+ \frac{2}{3} + \frac{1}{2}\frac{\psi^{\prime}(-1)}{\psi(-1)} \right) \\
2i \left( (\mathcal{A}+\widetilde{\mathcal{B}}_{-1})^{2} - 2(\mathcal{A}+\widetilde{\mathcal{B}}_{-1})+ \frac{2}{3} + \frac{1}{2}\frac{\psi^{\prime}(-1)}{\psi(-1)}  \right) & -2(\mathcal{A}+\widetilde{\mathcal{B}}_{-1})^{2} + 1 - \frac{\psi^{\prime}(-1)}{\psi(-1)}
\end{pmatrix}D_{\infty}^{-\sigma_{3}}.
\end{align*}

\section{Integration in $\beta_{1},...,\beta_{m}$}\label{Section: beta}
In this section we use the differential identity \eqref{differential identity} and the RH analysis of Section \ref{Section: steepest descent} only for $W\equiv 0$ and $V(x) = 2x^{2}$.
\subsection{Computation of $\widetilde{Y}(t_{k})$}\label{Subsection:Y tilde regularized}
For $z \in \mathcal{D}_{t_{k}}$, $k \in \{1,...,m\}$, $z$ outside the lenses and $z \in Q_{+,k}^{R}$, we invert the transformations $Y \mapsto T \mapsto S \mapsto R$ in order to express $Y$ in terms of $R$ and $P^{(t_{k})}$,
\begin{equation}\label{Y in Dt}
Y(z) = e^{- \frac{n\ell}{2}\sigma_{3}} R(z)P^{(t_{k})}(z) e^{ng(z)\sigma_{3}}e^{\frac{n\ell}{2}\sigma_{3}},
\end{equation}
with
\begin{equation}
P^{(t_{k})}(z) = E_{t_{k}}(z) \Phi_{\mathrm{HG}}(nf_{t_{k}}(z);\alpha_{k},\beta_{k})e^{\frac{\pi i \alpha_{k}}{2}\sigma_{3}}(z-t_{k})^{-\frac{\alpha_{k}}{2}\sigma_{3}}e^{-n\xi(z)\sigma_{3}}\omega_{k}(z)^{-\frac{\sigma_{3}}{2}}.
\end{equation}
Note that for $z \in Q_{+,k}^{R}$, $z$ outside the lenses, $nf_{t_{k}}(z) \in II$ (see Figure \ref{Fig:HG}) and we can use \eqref{phi_HG}
\begin{multline*}
\Phi_{\mathrm{HG}}(nf_{t_{k}}(z))e^{\frac{i\pi  \alpha_{k}}{4}\sigma_{3}} = \\  \begin{pmatrix}
 \frac{\Gamma(1 + \frac{\alpha_{k}}{2}-\beta_{k})}{\Gamma(1+\alpha_{k})}G(\frac{\alpha_{k}}{2}+\beta_{k}, \alpha_{k}; nf_{t_{k}}(z))e^{-\frac{i\pi\alpha_{k}}{2}} & -\frac{\Gamma(1 + \frac{\alpha_{k}}{2}-\beta_{k})}{\Gamma(\frac{\alpha_{k}}{2}+\beta_{k})}H(1+\frac{\alpha_{k}}{2}-\beta_{k},\alpha_{k};nf_{t_{k}}(z)e^{-\pi i}) \\ 
 \frac{\Gamma(1 + \frac{\alpha_{k}}{2}+\beta_{k})}{\Gamma(1+\alpha_{k})}G(1+\frac{\alpha_{k}}{2}+\beta_{k},\alpha_{k};nf_{t_{k}}(z))e^{-\frac{i\pi\alpha_{k}}{2}} & H(\frac{\alpha_{k}}{2}-\beta_{k},\alpha_{k};nf_{t_{k}}(z)e^{-\pi i})
\end{pmatrix}.
\end{multline*} \normalsize
In \cite[Section 13.14(iii)]{NIST}, using the relations \eqref{relation between G and H and Whittaker}, we can find asymptotics of these functions near the origin:
\begin{equation}
G(a,\alpha_{k};z) = z^{\frac{\alpha_{k}}{2}}(1+\bigO(z)), \qquad \mbox{ as } z \to 0,
\end{equation}
and, if $\alpha_{k} \neq 0$ and $a - \frac{\alpha_{k}}{2} \pm \frac{\alpha_{k}}{2} \neq 0,-1,-2,...$,
\begin{equation}
H(a,\alpha_{k};z) = \left\{ \begin{array}{l l}
\displaystyle \frac{\Gamma(\alpha_{k})}{\Gamma(a)}z^{-\frac{\alpha_{k}}{2}} + \bigO(z^{1-\frac{\Re \alpha_{k}}{2}}) + \bigO(z^{\frac{\Re\alpha_{k}}{2}}), & \displaystyle \mbox{ if } \Re \alpha_{k} > 0, \\
\displaystyle \frac{\Gamma(-\alpha_{k})}{\Gamma(a-\alpha_{k})}z^{\frac{\alpha_{k}}{2}} + \frac{\Gamma(\alpha_{k})}{\Gamma(a)}z^{-\frac{\alpha_{k}}{2}} + \bigO(z^{1+\frac{\Re\alpha_{k}}{2}}), & \displaystyle \mbox{ if } -1 < \Re \alpha_{k} \leq 0.
\end{array} \right.
\end{equation}
Note that conditions $a - \frac{\alpha_{k}}{2} \pm \frac{\alpha_{k}}{2} \neq 0,-1,-2,...$ are always satisfied for $a = \frac{\alpha_{k}}{2}-\beta_{k}$ and $a = 1+\frac{\alpha_{k}}{2}-\beta_{k}$, since $\Re \beta_{k} \in \big( \frac{-1}{4},\frac{1}{4}\big)$ and $\Re \alpha_{k} > -1$.
Therefore, using \eqref{asymptotics for f in D_t}, the leading terms in $\Phi_{\mathrm{HG}}(nf_{t_{k}}(z))e^{\frac{\pi i \alpha_{k}}{2}\sigma_{3}}(z-t_{k})^{-\frac{\alpha_{k}}{2}\sigma_{3}}\omega_{k}(z)^{-\frac{\sigma_{3}}{2}}$ as $z \to t_{k}$ for $-1 < \Re \alpha_{k} \leq 0$, $\alpha_{k} \neq 0$ are given by
\begin{equation}\label{lol 6}
\begin{pmatrix}
\displaystyle \Phi_{k,11} & \displaystyle \alpha_{k}^{-1}\big(\Phi_{k,12} + \widetilde{c}_{k} \Phi_{k,11}(z-t_{k})^{\alpha_{k}}\big) \\
\displaystyle \Phi_{k,21} & \displaystyle \alpha_{k}^{-1}\big(\Phi_{k,22} + \widetilde{c}_{k} \Phi_{k,21}(z-t_{k})^{\alpha_{k}}\big)
\end{pmatrix}, \qquad \widetilde{c}_{k} = \frac{\alpha_{k}\Gamma(1+\alpha_{k})\Gamma(-\alpha_{k})e^{-\frac{\pi i \alpha_{k}}{2}}\omega_{k}(t_{k})}{\Gamma(-\frac{\alpha_{k}}{2}-\beta_{k})\Gamma(1+\frac{\alpha_{k}}{2}+\beta_{k})},
\end{equation}
and
\begin{equation}\label{def Psi matrix}
\hspace{-0.2cm} \begin{array}{l l}
\displaystyle \Phi_{k,11} = \frac{\Gamma(1+\frac{\alpha_{k}}{2}-\beta_{k})}{\Gamma(1+\alpha_{k})}\Big(4n\sqrt{1-t_{k}^{2}}\Big)^{\frac{\alpha_{k}}{2}}\omega_{k}(t_{k})^{-\frac{1}{2}}, & \displaystyle \Phi_{k,12} = \frac{-\alpha_{k}\Gamma(\alpha_{k})}{\Gamma(\frac{\alpha_{k}}{2}+\beta_{k})}\Big(4n\sqrt{1-t_{k}^{2}}\Big)^{-\frac{\alpha_{k}}{2}}\omega_{k}(t_{k})^{\frac{1}{2}}, \\[0.4cm]
\displaystyle \Phi_{k,21} = \frac{\Gamma(1+\frac{\alpha_{k}}{2}+\beta_{k})}{\Gamma(1+\alpha_{k})}\Big(4n\sqrt{1-t_{k}^{2}}\Big)^{\frac{\alpha_{k}}{2}}\omega_{k}(t_{k})^{-\frac{1}{2}}, & \displaystyle \Phi_{k,22} = \frac{\alpha_{k}\Gamma(\alpha_{k})}{\Gamma(\frac{\alpha_{k}}{2}-\beta_{k})}\Big(4n\sqrt{1-t_{k}^{2}}\Big)^{-\frac{\alpha_{k}}{2}}\omega_{k}(t_{k})^{\frac{1}{2}}.
\end{array}
\end{equation}
By the connection formula $\Gamma(z)\Gamma(1-z) = \frac{\pi}{\sin(\pi z)}$, we can rewrite $\widetilde{c}_{k}$ as
\begin{equation}
\widetilde{c}_{k} = \frac{\pi\alpha_{k}}{\sin(\pi\alpha_{k})} \frac{e^{i\pi\beta_{k}} - e^{- i\pi\alpha_{k}}e^{-i\pi\beta_{k}}}{2\pi i} \omega_{k}(t_{k}).
\end{equation}
If $\Re \alpha_{k} > 0$, the same expression \eqref{lol 6} is valid, except that $\widetilde{c}_{k}$ has to be replaced by $0$. Therefore, by \eqref{lim VP}, \eqref{substraction in VP}, \eqref{def second column of tilde Y}, \eqref{relation between g and xi}, \eqref{asymptotics for R} and \eqref{Y in Dt}, we obtain as $n \to \infty$, for $\alpha_{k} \neq 0$,
\begin{equation}\label{Y tilde alpha_k neq 0}
\widetilde{Y}(t_{k}) = e^{-\frac{n\ell}{2}\sigma_{3}} \left(I + \bigO (n^{-1+2\beta_{\max}})\right)E_{t_{k}}(t_{k})\begin{pmatrix}
\Phi_{k,11} & \Phi_{k,12} \\ \Phi_{k,21} & \Phi_{k,22}
\end{pmatrix}e^{nt_{k}^{2}\sigma_{3}}.
\end{equation}
The case of $\alpha_{k} = 0$ is simpler. We can obtain $\widetilde{Y}(t_{k})$ directly from \eqref{def second column of tilde Y} and \eqref{def VP}. It suffices to take the limit $\alpha_{k}\to 0$ in \eqref{Y tilde alpha_k neq 0} and to note that $\lim_{x\to 0}x\Gamma(x) = 1$ (see e.g. \cite[Chapter 5]{NIST}). 
The goal for the rest of this section is to prove the following.
\begin{proposition}\label{prop:beta}
As $n \to \infty$,
\begin{equation*}
\begin{array}{r c l}
\displaystyle \log \frac{D_{n}(\vec{\alpha},\vec{\beta},2x^{2},0)}{D_{n}(\vec{\alpha},\vec{0},2x^{2},0)} & = & \displaystyle 2in \sum_{j=1}^{m} \Big(\arcsin t_{j} + t_{j} \sqrt{1-t_{j}^{2}}\Big)\beta_{j} + i \mathcal{A} \sum_{j=1}^{m} \beta_{j} \arcsin t_{j} \\
& & \displaystyle -\frac{i\pi}{2}\sum_{j=1}^{m} \beta_{j}\mathcal{A}_{j} - \sum_{j=1}^{m} \beta_{j}^{2} \log\big(8n(1-t_{j}^{2})^{3/2}\big) +  \sum_{1 \leq j < k \leq m} \log  T_{jk}^{2\beta_{j}\beta_{k}} \\
& & \displaystyle + \sum_{j=1}^{m} \log \left(   \frac{G(1+\frac{\alpha_{j}}{2}+\beta_{j})G(1+\frac{\alpha_{j}}{2}-\beta_{j})}{G(1+\frac{\alpha_{j}}{2})^{2}} \right) + \bigO \left( \frac{\log n}{n^{1-4 \beta_{\max}}}\right),
\end{array}
\end{equation*}
where $\mathcal{A}_{k} = \sum_{j=1}^{k-1} \alpha_{j} - \sum_{j=k+1}^{m} \alpha_{j}$.
\end{proposition}
By \eqref{def coeff OPs} and \eqref{Y definition}, we have
\begin{equation}\label{How to find the coeff}
\kappa_{n-1}^{2} = \lim_{z\to\infty} \frac{i Y_{21}(z)}{2\pi z^{n-1}}, \quad \eta_{n} = \lim_{z\to\infty}  \frac{Y_{11}(z)-z^{n}}{z^{n-1}}, \quad \gamma_{n} = \lim_{n\to\infty} \frac{Y_{11}(z)-z^{n}-\eta_{n}z^{n-1}}{z^{n-2}}.
\end{equation}
For $z$ outside the lenses and outside the disks, by inverting the transformations, we can express $Y$ in terms of $R$ and $P^{(\infty)}$, one has
\begin{equation}\label{Y near inf}
Y(z) = e^{-\frac{n\ell}{2}\sigma_{3}}R(z)P^{(\infty)}(z)e^{ng(z)\sigma_{3}}e^{\frac{n\ell}{2}\sigma_{3}}.
\end{equation}
From \eqref{asymptotics for e^{ng}}, \eqref{P_1^inf}, \eqref{asymptotics for R} and \eqref{R^{(1)}}, we obtain 
\begin{equation}\label{kappa_n-1}
\kappa_{n-1}^{2} = e^{n}2^{2(n-1)+\mathcal{A}}\pi^{-1} \exp\bigg(\hspace{-0.14cm}-2i\sum_{j=1}^{m}\beta_{j} \arcsin t_{j}\bigg) \bigg( 1+\frac{R_{1,21}^{(1)}}{n P^{(\infty)}_{1,21}} + \bigO (n^{-2+2\beta_{\max}})\bigg),
\end{equation}
where the notation $R_{k,ij}^{(1)}$ refers to the $z^{-k}$ coefficient in the large $z$ asymptotics of $R_{ij}^{(1)}$. More explicitly, combining \eqref{P_1^inf} and \eqref{R^{(1)}}, we find
\begin{equation}
 \frac{R_{1,21}^{(1)}}{P^{(\infty)}_{1,21}} = \sum_{j=1}^{m} \frac{v_{j}(1-\widetilde{\Lambda}_{R,1,j})}{2(1-t_{j}^{2})} - \frac{1}{8} \left( \frac{2}{3}-2\mathcal{A} +\mathcal{A}^{2}+(1-\mathcal{A})(\widetilde{\mathcal{B}}_{1}-\widetilde{\mathcal{B}}_{-1})+\frac{\widetilde{\mathcal{B}}_{1}^{2} + \widetilde{\mathcal{B}}_{-1}^{2}}{2} \right).
\end{equation}
The asymptotics for $\kappa_{n}^{2}$ are a bit more delicate to obtain than simply replace $n$ by $n+1$ in \eqref{kappa_n-1}, because the weight $w$ depends also on $n$. By writing explicitly the dependence of the weight $w$ in the parameters $n$, $t_{1}$,...,$t_{m}$, we have the relation
\begin{equation}\label{lol 7}
w\left(\sqrt{\frac{n+1}{n}}x;n,t_{1},...,t_{m}\right) = \left( \frac{n+1}{n} \right)^{\frac{\mathcal{A}}{2}} w\left(x;n+1,\sqrt{\frac{n}{n+1}}t_{1},...,\sqrt{\frac{n}{n+1}}t_{m} \right).
\end{equation}
From \eqref{lol 7} and with a change of variable in the orthogonality conditions \eqref{ortho conditions}, we see that to get large $n$ asymptotics for $\kappa_{n}^{2}$, first we have to replace $n$ by $n+1$ in \eqref{kappa_n-1}, then we multiply the result by $\left( \frac{n}{n+1} \right)^{n + \frac{1+\mathcal{A}}{2}}$ and finally we replace each $t_{k}$ by $\sqrt{\frac{n}{n+1}}t_{k}$. We find
\begin{multline*}
\kappa_{n}^{2} = e^{n}2^{2n+\mathcal{A}}\pi^{-1} \exp\Bigg(-2i\sum_{j=1}^{m}\beta_{j} \arcsin t_{j}\Bigg) \Bigg( 1+\frac{R_{1,21}^{(1)}}{n P^{(\infty)}_{1,21}} - \frac{\mathcal{A}}{2n} + \sum_{j=1}^{m} \frac{it_{j}\beta_{j}}{\sqrt{1-t_{j}^{2}}n} + \bigO (n^{-2+2\beta_{\max}}) \Bigg).
\end{multline*}
By \eqref{How to find the coeff}, \eqref{Y near inf}, \eqref{asymptotics for e^{ng}}, \eqref{P_1^inf} and \eqref{R^{(1)}}, we have as $n \to \infty$
\begin{equation}
\eta_{n} = P_{1,11}^{(\infty)} + \frac{R_{1,11}^{(1)}}{n} + \bigO(n^{-2+2\beta_{\max}}).
\end{equation}
From the same equations, and with increasing effort, we obtain large $n$ asymptotics for $\gamma_{n}$,
\begin{equation}
\gamma_{n} = -\frac{n}{8} + P_{2,11}^{(\infty)} + \frac{1}{n} \left[ P_{1,11}^{(\infty)}R_{1,11}^{(1)} + R_{2,11}^{(1)} + P_{1,21}^{(\infty)}R_{1,12}^{(1)} \right] + \bigO(n^{-2+2\beta_{\max}}).
\end{equation}
Taking into account that $R_{1,11}^{(1)} = \bigO(n^{2\beta_{\max}})$ as $n \to \infty$, we see that
\begin{equation}
\gamma_{n} - \frac{\eta_{n}^{2}}{2} = -\frac{n}{8} + P_{2,11}^{(\infty)} - \frac{(P_{1,11}^{(\infty)})^{2}}{2} + \frac{1}{n} \left[ R_{2,11}^{(1)} + P_{1,21}^{(\infty)}R_{1,12}^{(1)} \right] + \bigO(n^{-2+4\beta_{\max}}), \qquad \mbox{ as } n \to \infty.
\end{equation}
These terms are made explicit using \eqref{P_1^inf}, \eqref{P_2^inf} and \eqref{R^{(1)}}:
\begin{align*}
& P_{2,11}^{(\infty)}-\frac{(P_{1,11}^{(\infty)})^{2}}{2} = \frac{1}{8} - \frac{1}{8}\sum_{j=1}^{m}\alpha_{j}(1-2t_{j}^{2}) + \frac{i}{2}\sum_{j=1}^{m}t_{j}\sqrt{1-t_{j}^{2}}\beta_{j}, \\
& R_{2,11}^{(1)} = \sum_{j=1}^{m} \frac{t_{j}v_{j}(t_{j}+\widetilde{\Lambda}_{I,j})}{4(1-t_{j}^{2})} - \frac{1}{2^{5}} \left( \frac{2}{3} + (\mathcal{A}-\widetilde{\mathcal{B}}_{1})^{2} + (\mathcal{A}+\widetilde{\mathcal{B}}_{-1})^{2} \right),\\
& P_{1,21}^{(\infty)}R_{1,12}^{(1)} = -\sum_{j=1}^{m} \frac{v_{j}(1+\widetilde{\Lambda}_{R,2,j})}{8(1-t_{j}^{2})} + \frac{1}{2^{5}} \left( \frac{2}{3}+2\mathcal{A} +\mathcal{A}^{2}-(1+\mathcal{A})(\widetilde{\mathcal{B}}_{1}-\widetilde{\mathcal{B}}_{-1}) +\frac{\widetilde{\mathcal{B}}_{1}^{2} + \widetilde{\mathcal{B}}_{-1}^{2}}{2} \right).
\end{align*}
We are now able to compute the differential identity \eqref{differential identity}. After some calculations (there is a lot of cancellations), using \eqref{connection Lambdas} (and \eqref{estimates for R with beta} for the $o(1)$ term), we obtain
\begin{multline}\label{explicit first part of the diff identity}
-(n+\mathcal{A}) \partial_{\beta_{k}}\log(\kappa_{n}\kappa_{n-1}) - 2n\partial_{\beta_{k}}\left( \frac{\kappa_{n-1}^{2}}{\kappa_{n}^{2}} \right) + 4n \partial_{\beta_{k}} \left(\gamma_{n}-\frac{\eta_{n}^{2}}{2}\right) = \\
2in \left( \arcsin t_{k} +  t_{k} \sqrt{1-t_{k}^{2}} \right) + 2i\mathcal{A} \arcsin t_{k}  - 2\beta_{k} + \bigO\left(\frac{\log n}{n^{1-4\beta_{\max}}}\right).
\end{multline}
The second part of the differential identity is obtained using \eqref{Y tilde alpha_k neq 0}:
\begin{align}\label{explicit second part of the diff identity}
& \sum_{j=1}^{m} \left[\widetilde{Y}_{22}(t_{j})\partial_{\beta_{k}}Y_{11}(t_{j}) - \widetilde{Y}_{12}(t_{j})\partial_{\beta_{k}}Y_{21}(t_{j}) + Y_{11}(t_{j})\widetilde{Y}_{22}(t_{j})\partial_{\beta_{k}} \log (\kappa_{n}\kappa_{n-1}) \right] = \\
& = -\mathcal{A} \partial_{\beta_{k}}\log D_{\infty} + \sum_{j=1}^{m} \big( \partial_{\beta_{k}}\Phi_{j,11} \Phi_{j,22} - \Phi_{j,12}\partial_{\beta_{k}} \Phi_{j,21} -2\beta_{j} \partial_{\beta_{k}} \log \Lambda_{j} \big) + \bigO\left( \frac{\log n}{n^{1-4\beta_{\max}}} \right). \nonumber
\end{align}
These terms have more explicit forms using \eqref{D inf}, \eqref{def Lambda} and \eqref{def Psi matrix}:
\begin{align}
& -2\sum_{j=1}^{m}\beta_{j}\partial_{\beta_{k}}\log \Lambda_{j} = - 2 \beta_{k} \log(8n(1-t_{k}^{2})^{3/2}) + \sum_{j \neq k} \log  T_{jk}^{2\beta_{j}}, \\
& \sum_{j=1}^{m} \big( \partial_{\beta_{k}}\Phi_{j,11} \Phi_{j,22} - \Phi_{j,12}\partial_{\beta_{k}} \Phi_{j,21}  \big) = -\frac{i\pi}{2} \mathcal{A}_{j}  \\
& \hspace{2cm} +\beta_{k} \partial_{\beta_{k}} \log \frac{\Gamma(1+\frac{\alpha_{k}}{2}+\beta_{k})}{\Gamma(1+\frac{\alpha_{k}}{2}-\beta_{k})} + \frac{\alpha_{k}}{2} \partial_{\beta_{k}} \log \left(\Gamma(1+\frac{\alpha_{k}}{2}+\beta_{k})\Gamma(1+\frac{\alpha_{k}}{2} - \beta_{k}) \right), \nonumber
\end{align}
and $\mathcal{A}_{j}$ is defined in Proposition \ref{prop:beta}. Finally, summing \eqref{explicit first part of the diff identity} and \eqref{explicit second part of the diff identity}, we obtain
\begin{align}\label{Explicit diff identity}
& \partial_{\beta_{k}} \log D_{n}(\vec{\alpha},\vec{\beta},2x^{2},0) = 2in \Big(\arcsin t_{k} + t_{k} \sqrt{1-t_{k}^{2}}\Big) + i \mathcal{A} \arcsin t_{k} - \frac{i\pi}{2}\mathcal{A}_{k}- 2 \beta_{k} +  \sum_{j\neq k } \log  T_{jk}^{2\beta_{j}} \nonumber \\
& \hspace{2cm} -2 \beta_{k} \log \big(8n(1-t_{k}^{2})^{3/2}\big)  + \frac{\alpha_{k}}{2}\partial_{\beta_{k}} \log \left( \Gamma(1+\frac{\alpha_{k}}{2}+\beta_{k})\Gamma(1+\frac{\alpha_{k}}{2}-\beta_{k}) \right) \nonumber \\ 
& \hspace{2cm}  +\beta_{k} \partial_{\beta_{k}} \log \frac{\Gamma(1+\frac{\alpha_{k}}{2}+\beta_{k})}{\Gamma(1+\frac{\alpha_{k}}{2}-\beta_{k})} + \bigO \left( \frac{\log n}{n^{1-4\beta_{\max}}} \right).
\end{align}
For convenience, we denote $F_{1,n}(\vec{\alpha},\vec{\beta})$ for the r.h.s. of \eqref{differential identity}, which can be expressed in terms of $R$ and the parametrices. From Subsection \ref{Subsection: small-norm RH problem}, if $\Omega$ is a compact subset of $\mathcal{P}_{\alpha}\times \mathcal{P}_{\beta}^{(\frac{1}{4})}$, there exists $n_{\star}$ depending only on $\Omega$, such that $F_{1,n}(\vec{\alpha},\vec{\beta})$ exists for \textit{all} $n \geq n_{\star}$ and for \textit{all} $(\vec{\alpha},\vec{\beta}) \in \Omega$, and is analytic for $(\vec{\alpha},\vec{\beta}) \in \Omega$. Furthermore, the large $n$ asymptotics for $F_{1,n}(\vec{\alpha},\vec{\beta})$ given by the r.h.s. of \eqref{Explicit diff identity} are uniform for \textit{all} $(\vec{\alpha},\vec{\beta}) \in \Omega$. Nevertheless, the identity \eqref{Explicit diff identity} itself is valid only for $(\vec{\alpha},\vec{\beta}) \in \Omega \setminus \widetilde{\Omega}_{1}^{(n)}$, where $\widetilde{\Omega}_{1}^{(n)}$ consists of at most a finite number of points (see the beginning of Section \ref{Section:diff identities}). 

\vspace{0.2cm}\hspace{-0.53cm}Now, we show how to extend this differential identity for \textit{all} $(\vec{\alpha},\vec{\beta}) \in \Omega$, following \cite{Krasovsky, ItsKrasovsky, DIK}. First, we use the identity \eqref{Explicit diff identity} with $k = 1$, $\beta_{2} = ... = \beta_{m} = 0$ with $\vec{\alpha}$ fixed. We assume that $(\vec{\alpha},\vec{0}) \in \Omega$, we write $\vec{\beta}_{1} = (\beta_{1},0,...,0)$ and we define
\begin{equation}
H(\beta_{1}) = D_{n}(\vec{\alpha},\vec{\beta}_{1},2x^{2},0)\exp\bigg(-\int_{0}^{\beta_{1}}F_{1,n}(\vec{\alpha},(s_{1},0,...,0))ds_{1}\bigg).
\end{equation}
Equation \eqref{Explicit diff identity} is the statement that $H^{\prime}(\beta_{1}) = 0$ for all $\beta_{1}$ such that $(\vec{\alpha},\vec{\beta}_{1}) \in \Omega\setminus \widetilde{\Omega}_{1}^{(n)}$. Since the large $n$ asymptotics for $F_{1,n}(\vec{\alpha},\vec{\beta})$ given by \eqref{Explicit diff identity} hold uniformly and continuously for all $\beta_{1}$ such that $(\vec{\alpha},\vec{\beta}_{1}) \in \Omega$, $H$ is continuously differentiable for all $\beta_{1}$ such that $(\vec{\alpha},\vec{\beta}_{1}) \in \Omega$ (for $n \geq n_{\star}$). Thus by continuity, since $\widetilde{\Omega}_{1}^{(n)}$ consists of at most a finite number of points, one has in fact that $H^{\prime}(\beta_{1}) = 0$ for all $\beta_{1}$ such that $(\vec{\alpha},\vec{\beta}_{1}) \in \Omega$ (for $n \geq n_{\star}$). Moreover, as $H(0) = D_{n}(\vec{\alpha},\vec{0},2x^{2},0) \neq 0$ (if $n$ is chosen sufficiently large, see \eqref{partition GUE} and \eqref{asymptotics alpha Krasovsky}), the determinant $D_{n}(\vec{\alpha},\vec{\beta}_{1},2x^{2},0)$ is never zero. Hence the identity \eqref{Explicit diff identity} for $k = 1$ holds for all $\beta_{1}$ such that $(\vec{\alpha},\vec{\beta}_{1}) \in \Omega$ (for sufficiently large $n$). Integrating this identity in $\beta_{1}$ gives, as $n \to \infty$
\begin{align}\label{lol 9}
& \log\frac{D_{n}(\vec{\alpha},(\beta_{1},0,...,0),2x^{2},0)}{D_{n}(\vec{\alpha},(0,0,...,0),2x^{2},0)} = 2in \Big(\arcsin t_{1} + t_{1} \sqrt{1-t_{1}^{2}}\Big)\beta_{1} + i \mathcal{A} \beta_{1} \arcsin t_{1} - \frac{i\pi}{2}\mathcal{A}_{1}\beta_{1} \nonumber \\
& \hspace{3cm} - \beta_{1}^{2} \log \big(8n(1-t_{1}^{2})^{3/2}\big) - \beta_{1}^{2} +  \frac{\alpha_{1}}{2}\log \frac{\Gamma(1+\frac{\alpha_{1}}{2}+\beta_{1})\Gamma(1+\frac{\alpha_{1}}{2}-\beta_{1})}{\Gamma(1+\frac{\alpha_{1}}{2})^{2}} \nonumber \\
& \hspace{3cm} + \int_{0}^{\beta_{1}} x \partial_{x} \log \frac{\Gamma(1+\frac{\alpha_{1}}{2}+x)}{\Gamma(1+\frac{\alpha_{1}}{2}-x)}dx + \bigO \left( \frac{\log n}{n^{1-4 \beta_{\max}}} \right).
\end{align}
Since $\Omega$ is arbitrary, \eqref{lol 9} is valid for any $\beta_{1}$ such that $\Re \beta_{1} \in (\tfrac{-1}{4},\tfrac{1}{4})$. 
This equation can be simplified using the known formula (see \cite[formula 5.17.4]{NIST})
\begin{equation}
\int_{0}^{z}\log \Gamma (1+x) dx = \frac{z}{2} \log 2\pi - \frac{z(z+1)}{2} + z \log \Gamma(z+1) - \log G(z+1),
\end{equation}
where $G$ is Barnes' $G$-function. After integrations by parts, one has
\begin{multline}
\int_{0}^{\beta_{1}} x \partial_{x} \log \frac{\Gamma(1+\frac{\alpha_{1}}{2}+x)}{\Gamma(1+\frac{\alpha_{1}}{2}-x)}dx = \beta_{1}^{2} - \frac{\alpha_{1}}{2}\log \frac{\Gamma(1+\frac{\alpha_{1}}{2}+\beta_{1})\Gamma(1+\frac{\alpha_{1}}{2}-\beta_{1})}{\Gamma(1+\frac{\alpha_{1}}{2})^{2}} \\[0.3cm] + \log \frac{G(1+\frac{\alpha_{1}}{2}+\beta_{1})G(1+\frac{\alpha_{1}}{2}-\beta_{1})}{G(1+\frac{\alpha_{1}}{2})^{2}}.
\end{multline}
Therefore \eqref{lol 9} can be rewritten as
\begin{multline}\label{in beta_1}
\log\frac{D_{n}(\alpha,(\beta_{1},0,...,0),2x^{2},0)}{D_{n}(\alpha,(0,0,...,0),2x^{2},0)} = 2in \Big(\arcsin t_{1} + t_{1} \sqrt{1-t_{1}^{2}}\Big)\beta_{1} + i \mathcal{A} \beta_{1} \arcsin t_{1} - \frac{i\pi}{2}\mathcal{A}_{1}\beta_{1} \\
- \beta_{1}^{2} \log \big(8n(1-t_{1}^{2})^{3/2}\big) + \log \frac{G(1+\frac{\alpha_{1}}{2}+\beta_{1})G(1+\frac{\alpha_{1}}{2}-\beta_{1})}{G(1+\frac{\alpha_{1}}{2})^{2}}+ \bigO \left( \frac{\log n}{n^{1-4 \beta_{\max}}}\right).
\end{multline}
Now, we use \eqref{Explicit diff identity} for $k = 2$, where we fix $\beta_{1}$ and $\vec{\alpha}$, and we set $\beta_{3}=...=\beta_{m} = 0$. For convenience, we write $\vec{\beta}_{2} = (\beta_{1},\beta_{2},0,...,0)$. Equation \eqref{Explicit diff identity} is valid for all $\beta_{2}$ such that $(\vec{\alpha},\vec{\beta}_{2}) \in \Omega\setminus \Omega_{1}^{(n)}$. We can extend it to all $\beta_{2}$ such that $(\vec{\alpha},\vec{\beta}_{2}) \in \Omega$ with a similar argument to the one done below equation \eqref{Explicit diff identity}. After an integration in $\beta_{2}$, we obtain a similar formula to \eqref{in beta_1} but with an extra term
\begin{multline*}
\log\frac{D_{n}(\alpha,(\beta_{1},\beta_{2},0,...,0),2x^{2},0)}{D_{n}(\alpha,(\beta_{1},0,...,0),2x^{2},0)} = 2in \Big(\arcsin t_{2} + t_{2} \sqrt{1-t_{2}^{2}}\Big)\beta_{2} + i \mathcal{A} \beta_{2} \arcsin t_{2}  +  \log T_{12}^{2\beta_{1}\beta_{2}} \\
- \frac{i\pi}{2}\mathcal{A}_{2}\beta_{2}- \beta_{2}^{2} \log \big(8n(1-t_{2}^{2})^{3/2}\big) + \log \frac{G(1+\frac{\alpha_{2}}{2}+\beta_{2})G(1+\frac{\alpha_{2}}{2}-\beta_{2})}{G(1+\frac{\alpha_{2}}{2})^{2}}+ \bigO \left( \frac{\log n}{n^{1-4 \beta_{\max}}}\right).
\end{multline*}
Again, by freedom in the choice of $\Omega$, the above expansion is valid for any $\beta_{2}$ such that $\Re \beta_{2} \in (\tfrac{-1}{4},\tfrac{1}{4})$. We can proceed in the same way recursively for each $\beta_{k}$, $k=1,...,m$. After $m$ integrations, it suffices to sum these identities to obtain Proposition \ref{prop:beta}.

\section{Integration in $V$}\label{Section: integrating V}
In this section, we will use the RH analysis done in Section \ref{Section: steepest descent} with $W \equiv 0$ and where $V$ is replaced by $V_{s}$, defined in \eqref{deformation parameter potential} and we will make explicit the dependence of the weight in $s$:
\begin{equation}
w_{s}(x) = e^{-nV_{s}(x)} \omega_{\alpha}(x)\omega_{\beta}(x).
\end{equation} 
In the RH analysis, we also have to replace the Euler-Lagrange constant $\ell$ by $\ell_{s}$ and the function $\psi$ by $\psi_{s}$, see \eqref{d_s and ell_s}. We will also use the differential identity 
\begin{equation}\label{lol 31}
\partial_{s} \log D_{n}(\vec{\alpha},\vec{\beta},V_{s},0) =  \frac{1}{2\pi i}\int_{\mathbb{R}}[Y^{-1}(x)Y^{\prime}(x)]_{21}\partial_{s}w_{s}(x)dx,
\end{equation}
which was obtained in \eqref{diff identity s}, and is valid only when $D_{k}^{(n)}(\vec{\alpha},\vec{\beta},V_{s},W) \neq 0$ for all $k = 1,2,...,n+1$, i.e. for $(\vec{\alpha},\vec{\beta},s) \in \widetilde{\mathcal{P}}_{2}^{(n)}$ (see also the discussion at the beginning of Section \ref{Section:diff identities}). Some calculations in this section and the next one are similar to those done in \cite{BerWebbWong}, in which the authors put great effort in showing all the details, so we will sometimes refer to equations and lemmas in their paper. Our goal is to prove the following.
\begin{proposition}\label{prop:V}
As $n \to \infty$, we have
\begin{align*}
& \log \frac{D_{n}(\vec{\alpha},\vec{\beta},V,0)}{D_{n}(\vec{\alpha},\vec{\beta},2x^{2},0)} = - \frac{n^{2}}{2} \int_{-1}^{1} \sqrt{1-x^{2}}\big(V(x)-2x^{2}\big)\Big( \frac{2}{\pi} + \psi(x)\Big)dx + n \sum_{j=1}^{m} \frac{\alpha_{j}}{2}\big(V(t_{j})-2t_{j}^{2}\big) \nonumber \\
& \hspace{2cm} - \frac{n\mathcal{A}}{2\pi} \int_{-1}^{1} \frac{V(x)-2x^{2}}{\sqrt{1-x^{2}}}dx -2\pi n \sum_{j=1}^{m} i \beta_{j} \int_{t_{j}}^{1} \Big(\psi(x)-\frac{2}{\pi}\Big)\sqrt{1-x^{2}}dx \nonumber \\
& \hspace{2cm} - \sum_{j=1}^{m} \bigg( \beta_{j}^{2} - \frac{\alpha_{j}^{2}}{4} \bigg) \log \left( \frac{\pi}{2}\psi(t_{j}) \right) - \frac{1}{24}\log\left( \frac{\pi^{2}}{4}\psi(1)\psi(-1) \right) + \bigO\big(n^{-1+4\beta_{\max}}\big).
\end{align*}
\end{proposition}
Using the jump relations \eqref{jump relations of Y} of $Y$, the differential identity \eqref{lol 31} becomes
\begin{equation}\label{lol 1}
\partial_{s} \log D_{n}(\vec{\alpha},\vec{\beta},V_{s},0) = \int_{\mathbb{R}\setminus [-1-\epsilon,1+\epsilon]} \hspace{-1.1cm} [Y^{-1}(x)Y^{\prime}(x)]_{21} \partial_{s} w_{s}(x) \frac{dx}{2\pi i} - \frac{1}{2\pi i} \int_{\mathcal{C}} [Y^{-1}(z)Y^{\prime}(z)]_{11} \partial_{s} \log w_{s}(z) dz,
\end{equation}
where $\epsilon >0$ is fixed and $\mathcal{C}$ is a closed curve surrounding $[-1,1]$ and the lenses $\gamma_{+}\cup\gamma_{-}$, is oriented clockwise, and passes through $-1-\epsilon$ and $1+\epsilon$. 

\hspace{-0.5cm}For $z$ outside the lenses, similarly to \eqref{Y near inf}, we have
\begin{equation}\label{Y near inf 2}
Y(z) = e^{-\frac{n\ell_{s}}{2}\sigma_{3}}R(z)P^{(\infty)}(z)e^{ng(z)\sigma_{3}}e^{\frac{n\ell_{s}}{2}\sigma_{3}},
\end{equation}
and thus, by \eqref{var inequality} and \eqref{g+ + g-}, one has
\begin{equation}
[Y^{-1}(x)Y^{\prime}(x)]_{21} \partial_{s} w_{s}(x) = \bigO(e^{-cn}), \qquad \mbox{ as } n \to \infty,
\end{equation}
uniformly for $x \in \mathbb{R}\setminus [-1-\epsilon,1+\epsilon]$ and $s \in [0,1]$ (see Remark \ref{Remark: R} (a)), where $c >0$ is a fixed constant. We can explicitly compute $[Y^{-1}(z)Y^{\prime}(z)]_{11}$ using \eqref{Y near inf 2}. Thus, equation \eqref{lol 1} becomes as $n \to \infty$
\begin{equation}\label{lol 4}
\begin{array}{r c l}
\displaystyle \partial_{s} \log D_{n}(\vec{\alpha},\vec{\beta},V_{s},0)  & = & \displaystyle I_{1,s} + I_{2,s} + I_{3,s} + \bigO(e^{-cn}), \\[0.3cm]
\displaystyle I_{1,s} & = & \displaystyle \frac{-n}{2\pi i} \int_{\mathcal{C}} g^{\prime}(z) \partial_{s} \log w_{s}(z) dz, \\[0.3cm]
\displaystyle I_{2,s} & = & \displaystyle \frac{-1}{2\pi i} \int_{\mathcal{C}} [P^{(\infty)}(z)^{-1}P^{(\infty)}(z)^{\prime}]_{11} \partial_{s} \log w_{s}(z) dz, \\[0.3cm]
\displaystyle I_{3,s} & = & \displaystyle \frac{-1}{2\pi i} \int_{\mathcal{C}} [P^{(\infty)}(z)^{-1}R^{-1}(z)R^{\prime}(z)P^{(\infty)}(z)]_{11} \partial_{s} \log w_{s}(z) dz.
\end{array}
\end{equation}
Since $\partial_{s} \log w(z) = -n \partial_{s} V_{s}(z)$, by \eqref{asymptotics for R}, we have  that 
\begin{equation}
I_{1,s} = \bigO(n^{2}), \qquad I_{2,s} = \bigO(n), \qquad I_{3,s} = \bigO(n^{2\beta_{\max}}), \qquad \mbox{ as } n \to \infty.
\end{equation}
More detailed calculations will later show that we actually have $\int_{0}^{1}I_{3,s}ds = \bigO(1)$ as $n \to \infty$, due to the fact that $I_{3,s}$ is highly oscillatory in $s$ as $n\to\infty$. By \eqref{g+ - g- 2}, a direct calculation shows that
\begin{equation}\label{lol 37}
I_{1,s} = - n^{2} \int_{-1}^{1} \big(V(x)-2x^{2}\big)\psi_{s}(x)\sqrt{1-x^{2}}dx,
\end{equation}
and therefore
\begin{equation}
\int_{0}^{1} I_{1,s} ds = - \frac{n^{2}}{2} \int_{-1}^{1} \sqrt{1-x^{2}}\big(V(x)-2x^{2}\big)\Big( \frac{2}{\pi} + \psi(x)\Big)dx.
\end{equation}
From \eqref{Pinf Region 1}, we have $[P^{(\infty)}(z)^{-1}P^{(\infty)}(z)^{\prime}]_{11} = - \partial_{z} \log D(z)$, and thus $I_{2,s} = I_{2,s,\alpha} + I_{2,s,\beta}$, with
\begin{equation}
I_{2,s,\alpha} = \displaystyle \frac{-n}{2\pi i} \int_{\mathcal{C}} \partial_{z} \log D_{\alpha}(z) \partial_{s}V_{s}(z)dz, \qquad I_{2,s,\beta} = \displaystyle \frac{-n}{2\pi i} \int_{\mathcal{C}} \partial_{z} \log D_{\beta}(z) \partial_{s}V_{s}(z)dz. 
\end{equation}
Note that since $\partial_{s}V_{s}(x) = V(x)-2x^{2}$, $I_{2,s}$ is in fact independent of $s$. We can compute $\partial_{z}\log D(z)$ using the explicit form of \eqref{def D_alpha} and \eqref{def D_beta}. This gives 
\begin{equation}\label{lol 5}
\partial_{z} \log D_{\alpha}(z) = - \frac{\mathcal{A}}{2} \frac{1}{\sqrt{z^{2}-1}} + \sum_{j=1}^{m} \frac{\alpha_{j}}{2} \frac{1}{z-t_{j}}, \qquad \partial_{z} \log D_{\beta}(z) = \sum_{j=1}^{m} \frac{i \beta_{j} \sqrt{1-t_{j}^{2}}}{\sqrt{z^{2}-1}(z-t_{j})}.
\end{equation}
By a contour deformation of $\mathcal{C}$, we obtain
\begin{equation}\label{final expression for I_2,s,alpha}
I_{2,s,\alpha} = n \sum_{j=1}^{m} \frac{\alpha_{j}}{2}\big(V(t_{j})-2t_{j}^{2}\big) - \frac{n\mathcal{A}}{2\pi} \int_{-1}^{1} \frac{V(x)-2x^{2}}{\sqrt{1-x^{2}}}dx,
\end{equation}
and
\begin{equation}
I_{2,s,\beta} = n \sum_{j=1}^{m} \frac{i\beta_{j}}{\pi}\sqrt{1-t_{j}^{2}}  \dashint_{-1}^{1} \frac{V(x)-2x^{2}}{\sqrt{1-x^{2}}(x-t_{j})}dx.
\end{equation}
By \cite[equation (5.17)]{BerWebbWong}, this principal value integral can be a bit simplified:
\begin{equation}\label{final expression for I_2,s,beta}
I_{2,s,\beta} = -2\pi n \sum_{j=1}^{m} i \beta_{j} \int_{t_{j}}^{1} \Big(\psi(x)-\frac{2}{\pi}\Big)\sqrt{1-x^{2}}dx.
\end{equation}
Integrating \eqref{final expression for I_2,s,alpha} and \eqref{final expression for I_2,s,beta} in $s$, we get
\begin{equation}
\begin{array}{r c l}
\displaystyle \int_{0}^{1} I_{2,s} ds & = & \displaystyle n \sum_{j=1}^{m} \frac{\alpha_{j}}{2}\big(V(t_{j})-2t_{j}^{2}\big) - \frac{n\mathcal{A}}{2\pi} \int_{-1}^{1} \frac{V(x)-2x^{2}}{\sqrt{1-x^{2}}}dx \\
& & \displaystyle -2\pi n \sum_{j=1}^{m} i \beta_{j} \int_{t_{j}}^{1} \Big(\psi(x)-\frac{2}{\pi}\Big)\sqrt{1-x^{2}}dx.
\end{array}
\end{equation}
An explicit form of the asymptotics for $I_{3,s}$ requires more work, as it involves $R^{(1)}$. By \eqref{asymptotics for R} and Remark \ref{Remark: R} (a), $I_{3,s}$ can be rewritten as 
\begin{equation}
I_{3,s} = \frac{1}{2\pi i} \int_{\mathcal{C}} [P^{(\infty)}(z)^{-1}R^{(1)}(z)^{\prime}P^{(\infty)}(z)]_{11} \partial_{s}V_{s}(z)dz + \bigO\big(n^{-1+4\beta_{\max}}\big), \qquad \mbox{ as } n \to \infty,
\end{equation}
uniformly for $(\vec{\alpha},\vec{\beta},s) \in \Omega\times [0,1]$, where $\Omega$ is an arbitrary compact subset of $\mathcal{P}_{\alpha}\times \mathcal{P}_{\beta}^{(\frac{1}{4})}$. Using \eqref{Pinf Region 1}, it becomes, as $n \to \infty$
\begin{multline}\label{lol 12}
I_{3,s} = \frac{1}{2\pi i} \int_{\mathcal{C}} \left( \frac{a(z)^{2}+a(z)^{-2}}{4}[R_{11}^{(1)}(z)^{\prime} - R_{22}^{(1)}(z)^{\prime}] + \frac{1}{2}[R_{11}^{(1)}(z)^{\prime}+R_{22}^{(1)}(z)^{\prime}] \right. \\ \left. + i \frac{a(z)^{-2}-a(z)^{2}}{4} [R_{12}^{(1)}(z)^{\prime}D_{\infty}^{-2} + R_{21}^{(1)}(z)^{\prime}D_{\infty}^{2}] \right) \big(V(z)-2z^{2}\big)dz + \bigO\big(n^{-1+4\beta_{\max}}\big).
\end{multline}
From \eqref{R^{(1)}} and below, we have
\begin{align*}
& R_{11}^{(1)\prime}(z) - R_{22}^{(1)\prime}(z) = \sum_{j=1}^{m} \frac{1}{(z-t_{j})^{2}} \frac{v_{j}(-2t_{j}-2\widetilde{\Lambda}_{I,j})}{2\pi \psi_{s}(t_{j}) (1-t_{j}^{2})}+ \frac{1}{(z-1)^{3}} \frac{5}{2^{3}3\pi \psi_{s}(1)}  \\
& \hspace{0cm} +\frac{1}{(z-1)^{2}} \frac{(\mathcal{A}-\widetilde{\mathcal{B}}_{1})^{2} - \frac{1}{2} - \frac{1}{2}\frac{\psi_{s}^{\prime}(1)}{\psi_{s}(1)}}{2^{3}\pi \psi_{s}(1)}+\frac{1}{(z+1)^{2}} \frac{-(\mathcal{A}+\widetilde{\mathcal{B}}_{-1})^{2} + \frac{1}{2} - \frac{1}{2}\frac{\psi_{s}^{\prime}(-1)}{\psi_{s}(-1)}}{2^{3}\pi \psi_{s}(-1)}+ \frac{1}{(z+1)^{3}} \frac{5}{2^{3}3\pi \psi_{s}(-1)}, \nonumber \\[0.3cm]
& R_{11}^{(1)\prime}(z) + R_{22}^{(1)\prime}(z) = 0, \\
& i[R_{12}^{(1)\prime}(z)D_{\infty}^{-2}+R_{21}^{(1)\prime}(z)D_{\infty}^{2}] = \sum_{j=1}^{m} \frac{1}{(z-t_{j})^{2}} \frac{v_{j}(-2+\widetilde{\Lambda}_{R,1,j}-\widetilde{\Lambda}_{R,2,j})}{2\pi \psi_{s}(t_{j}) (1-t_{j}^{2})}+ \frac{1}{(z-1)^{3}} \frac{5}{2^{3}3\pi \psi_{s}(1)} \\
& \hspace{0cm} +\frac{1}{(z-1)^{2}} \frac{(\mathcal{A}-\widetilde{\mathcal{B}}_{1})^{2} + \frac{2}{3} - \frac{1}{2}\frac{\psi_{s}^{\prime}(1)}{\psi_{s}(1)}}{2^{3}\pi \psi_{s}(1)} +\frac{1}{(z+1)^{2}} \frac{(\mathcal{A}+\widetilde{\mathcal{B}}_{-1})^{2} + \frac{2}{3} + \frac{1}{2}\frac{\psi_{s}^{\prime}(-1)}{\psi_{s}(-1)}}{2^{3}\pi \psi_{s}(-1)}+ \frac{1}{(z+1)^{3}} \frac{-5}{2^{3}3\pi \psi_{s}(-1)}.
\end{align*}
Plugging these equations in \eqref{lol 12}, we can split $I_{3,s}$ into $m+2$ integrals and an error term:
\begin{equation}\label{lol 3}
I_{3,s} = \sum_{j=1}^{m} I_{3,s,t_{j}} + I_{3,s,1} + I_{3,s,-1} + \bigO\big(n^{-1+4\beta_{\max}}\big), \qquad \mbox{ as } n \to \infty,
\end{equation}
where
\begin{align}
& I_{3,s,t_{k}} = \frac{-v_{k}}{8\pi^{2} \psi_{s}(t_{k})\sqrt{1-t_{k}^{2}}} \int_{\mathcal{C}} \left[ \frac{a_{+}^{2}(t_{k})}{a^{2}(z)} + \frac{a^{2}(z)}{a_{+}^{2}(t_{k})} + \widetilde{\Lambda}_{I,k} \left( \frac{a_{+}^{2}(t_{k})}{a^{2}(z)} - \frac{a^{2}(z)}{a_{+}^{2}(t_{k})} \right) \right] \frac{\partial_{s}V_{s}(z)}{(z-t_{k})^{2}}dz, \nonumber \\
& I_{3,s,1} = \int_{\mathcal{C}} \Bigg[ \frac{a^{-2}(z)}{4\pi \psi_{s}(1)} \Bigg( \frac{2(\mathcal{A}-\widetilde{\mathcal{B}}_{1})^{2}+\frac{1}{6}-\frac{\psi_{s}^{\prime}(1)}{\psi_{s}(1)}}{2^{3}(z-1)^{2}}+\frac{5}{2^{2}3(z-1)^{3}} \Bigg) + \frac{a^{2}(z)}{4(z-1)^{2}}\frac{-\frac{7}{6}}{2^{3}\pi \psi_{s}(1)}  \Bigg] \partial_{s}V_{s}(z)\frac{dz}{2\pi i}, \nonumber \\
& I_{3,s,-1} \hspace{-0.1cm} = \hspace{-0.1cm} \int_{\mathcal{C}} \hspace{-0.1cm} \Bigg[ \frac{a^{2}(z)}{4\pi \psi_{s}(-1)} \Bigg( \frac{-2(\mathcal{A}+\widetilde{\mathcal{B}}_{-1})^{2}-\frac{1}{6}-\frac{\psi_{s}^{\prime}(-1)}{\psi_{s}(-1)}}{2^{3}(z+1)^{2}}+\frac{5}{2^{2}3(z+1)^{3}} \Bigg) \hspace{-0.1cm} + \hspace{-0.1cm} \frac{a^{-2}(z)}{4(z+1)^{2}}\frac{\frac{7}{6}}{2^{3}\pi \psi_{s}(-1)}  \Bigg] \hspace{-0.07cm} \partial_{s}V_{s}(z)\frac{dz}{2\pi i}. \nonumber
\end{align}
From \cite[equations (5.16), (5.17), (5.22), (5.23)]{BerWebbWong}, we have
\begin{align}
& \int_{\mathcal{C}} \left( \frac{a_{+}^{2}(t_{k})}{a^{2}(z)} + \frac{a^{2}(z)}{a_{+}^{2}(t_{k})} \right)\frac{\partial_{s}V_{s}(z)}{(z-t_{k})^{2}}dz = 8\pi^{2} \sqrt{1-t_{k}^{2}} \Big(\psi(t_{k})- \frac{2}{\pi}\Big), \\
& \int_{\mathcal{C}} \left( \frac{a_{+}^{2}(t_{k})}{a^{2}(z)} - \frac{a^{2}(z)}{a_{+}^{2}(t_{k})} \right)\frac{\partial_{s}V_{s}(z)}{(z-t_{k})^{2}}dz = \frac{8\pi^{2}}{1-t_{k}^{2}} \int_{t_{k}}^{1} \Big(\psi(x)- \frac{2}{\pi}\Big)\sqrt{1-x^{2}}dx.
\end{align}
This allows us to evaluate $I_{3,s,t_{k}}$ more explicitly:
\begin{equation}\label{final expression for I_3,s,t_k}
I_{3,s,t_{k}} = - \frac{v_{k}}{\psi_{s}(t_{k})}\Big( \psi(t_{k})-\frac{2}{\pi}\Big) - \frac{v_{k} \widetilde{\Lambda}_{I,k}}{(1-t_{k}^{2})^{3/2}\psi_{s}(t_{k})}\int_{t_{k}}^{1} \Big( \psi(x)-\frac{2}{\pi}\Big) \sqrt{1-x^{2}}dx.
\end{equation}
Note that $\Lambda_{k}$ (and therefore $\widetilde{\Lambda}_{I,k}$) depends on $s$. Therefore, integrating it in $s$, we see that
\begin{equation}
\int_{0}^{1} I_{3,s,t_{k}}ds = -v_{k}\log \left( \frac{\pi}{2}\psi(t_{k}) \right) - \frac{v_{k}}{(1-t_{k}^{2})^{3/2}}\int_{t_{k}}^{1} \Big( \psi(x)-\frac{2}{\pi}\Big) \sqrt{1-x^{2}}dx \int_{0}^{1} \frac{\widetilde{\Lambda}_{I,k}}{\psi_{s}(t_{k})}ds.
\end{equation}
From \eqref{def Lambda}, one has $\Lambda_{k}^{\pm 2} = \bigO(n^{\pm 2\beta_{k}})$ as $n \to \infty$, but it is also highly oscillatory in $s$. We can write
\begin{equation}
\Lambda_{k}^{2} = g_{k,n}n^{2\beta_{k}}\psi_{s}(t_{k})^{2\beta_{k}} e^{2\pi i n s \int_{t_{k}}^{1} \left( \psi(x)-\frac{2}{\pi}\right) \sqrt{1-x^{2}}dx},
\end{equation}
where $g_{k,n}$ is independent of $s$ and $g_{k,n} = \bigO(1)$ as $n \to \infty$. Thus,
\begin{equation}\label{lol 2}
\int_{t_{k}}^{1} \Big( \psi(x)-\frac{2}{\pi}\Big) \sqrt{1-x^{2}}dx \int_{0}^{1} \frac{\Lambda_{k}^{\pm 2}}{\psi_{s}(t_{k})}ds = \frac{\pm g_{k,n}^{\pm 1}}{n^{1\mp 2\beta_{k}}} \int_{0}^{1} \frac{\partial_{s} \left( e^{\pm 2\pi i n s \int_{t_{k}}^{1} \left( \psi(x)-\frac{2}{\pi}\right) \sqrt{1-x^{2}}dx} \right)}{\psi_{s}(t_{k})^{1\mp 2\beta_{k}}}\frac{ds}{2\pi  i}.
\end{equation}
From integration by parts of the right-hand side of \eqref{lol 2}, it follows that this integral is $\bigO(n^{-1\pm 2\beta_{k}})$ as $n \to \infty$, and thus
\begin{equation}
\int_{0}^{1} I_{3,s,t_{k}}ds = -v_{k}\log \left( \frac{\pi}{2}\psi(t_{k}) \right) + \bigO(n^{-1+2|\Re \beta_{k}|}), \qquad \mbox{ as } n \to \infty.
\end{equation}
We now turn to the computation of $I_{3,s,1}$. From \cite[equations (5.27), (5.28) and (5.29)]{BerWebbWong}, we have $\int_{\mathcal{C}} \frac{a^{-2}(z)}{(z-1)^{2}}\partial_{s} V_{s}(z) dz = 0$ and 
\begin{align*}
& \int_{\mathcal{C}} \frac{a^{-2}(z)}{(z-1)^{3}}\partial_{s} V_{s}(z) dz = -\frac{8\pi^{2} i}{3}\Big( \psi(1)-\frac{2}{\pi}\Big), \quad \int_{\mathcal{C}} \frac{a^{2}(z)}{(z-1)^{2}}\partial_{s} V_{s}(z) dz = -\frac{16\pi^{2} i}{3}\Big( \psi(1)-\frac{2}{\pi}\Big).
\end{align*}
Therefore, we obtain $I_{3,s,1} = \frac{-\left( \psi(1)-\frac{2}{\pi} \right)}{24\psi_{s}(1)}$. The computation of $I_{3,s,-1}$ is similar, and gives $I_{3,s,-1} = \frac{-\left( \psi(-1)-\frac{2}{\pi} \right)}{24\psi_{s}(-1)}$. By integrating it in $s$ from $0$ to $1$, it gives
\begin{equation}\label{integration of I_3,s,1}
\int_{0}^{1} I_{3,s,1} ds = -\frac{1}{24}\log \left( \frac{\pi}{2}\psi(1) \right).
\end{equation}
The computation of $I_{3,s,-1}$ is similar, and gives $\int_{0}^{1} I_{3,s,-1} ds = -\frac{1}{24}\log \left( \frac{\pi}{2}\psi(-1) \right)$.
Since the $\bigO$ term in \eqref{lol 3} is uniform in $s\in [0,1]$ (see Remark \ref{Remark: R} (a)), we obtain, as $n \to \infty$
\begin{equation}
\int_{0}^{1} I_{3,s}ds = - \sum_{j=1}^{m} v_{j} \log \left( \frac{\pi}{2}\psi(t_{j}) \right) - \frac{1}{24}\log\left( \frac{\pi^{2}}{4}\psi(1)\psi(-1) \right) + \bigO(n^{-1+4\beta_{\max}}).
\end{equation}
From the above calculations, as $n \to \infty$, we have
\begin{equation}\label{lol 38}
\partial_{s}\log D_{n}(\vec{\alpha},\vec{\beta},V_{s},0) = I_{1,s} + I_{2,s,\alpha} + I_{2,s,\beta} + \sum_{j=1}^{m} I_{3,s,t_{j}} + I_{3,s,1} + I_{3,s,-1} + \bigO\big(n^{-1+4\beta_{\max}}\big),
\end{equation}
where the quantities $I_{1,s}$, $I_{2,s,\alpha}$,..., as well as the integrals $\int_{0}^{1}I_{1,s}ds$, $\int_{0}^{1}I_{2,s,\alpha}ds$,... have already been computed explicitly. For convenience, we denote $F_{2,n}(\vec{\alpha},\vec{\beta},s)$ for the r.h.s. of \eqref{lol 1}, which can be expressed in terms of $R$ and the parametrices. We recall that (see the beginning of Section \ref{Section:diff identities} and Subsection \ref{Subsection: small-norm RH problem}), if $\Omega$ is a compact subset of $\mathcal{P}_{\alpha}\times \mathcal{P}_{\beta}^{(\frac{1}{4})}$, there exists $n_{\star}$ depending only on $\Omega$, such that $F_{2,n}(\vec{\alpha},\vec{\beta},s)$ exists for \textit{all} $n \geq n_{\star}$ and for \textit{all} $(\vec{\alpha},\vec{\beta},s) \in \Omega\times[0,1]$, and is analytic for $s \in (0,1)$. Furthermore, the large $n$ asymptotics for $F_{2,n}(\vec{\alpha},\vec{\beta},s)$ given by the r.h.s. of \eqref{lol 38} is uniform for \textit{all} $(\vec{\alpha},\vec{\beta},s) \in \Omega\times[0,1]$. Nevertheless, the identity \eqref{lol 38} itself is valid only for $(\vec{\alpha},\vec{\beta},s) \in (\Omega \times (0,1)) \setminus \widetilde{\Omega}_{2}^{(n)}$, where $\widetilde{\Omega}_{2}^{(n)}$ is such that, for any $\epsilon \in (0,\tfrac{1}{2})$, $\widetilde{\Omega}_{2,\epsilon}^{(n)} = \{ (\vec{\alpha},\vec{\beta},s)\in \Omega_{2}^{(n)}:s \in (\epsilon,1-\epsilon) \}$ consists of at most a finite number of points. 

\vspace{0.2cm}\hspace{-0.55cm}Now, we show how to extend this differential identity for \textit{all} $(\vec{\alpha},\vec{\beta},s) \in \Omega\times (0,1)$, by adapting slightly the argument presented in Section \ref{Section: beta}, below equation \eqref{Explicit diff identity}. We fix $\vec{\alpha}$ and $\vec{\beta}$ such that $(\vec{\alpha},\vec{\beta})\in \Omega$, and we define $H(s) = D_{n}(\vec{\alpha},\vec{\beta},V_{s},0)\exp(-\int_{0}^{s}F_{2,n}(\vec{\alpha},\vec{\beta},\widetilde{s})d\widetilde{s})$. The equation \eqref{lol 38} is the statement that $H^{\prime}(s) = 0$ for all $s \in (0,1)$ such that $(\vec{\alpha},\vec{\beta},s) \in (\Omega\times(0,1))\setminus \widetilde{\Omega}_{2}^{(n)}$. Since the large $n$ asymptotics of $F_{2,n}(\vec{\alpha},\vec{\beta},s)$ given by the r.h.s. of \eqref{lol 38} hold uniformly for $s \in [0,1]$ and continuously (even analytically) for $s \in (0,1)$, $H$ is continuously differentiable for all $s\in (0,1)$ (for $n \geq n_{\star}$). Thus, by continuity, one has $H^{\prime}(s) = 0$ for all $s \in (\epsilon,1-\epsilon)$. Since $\epsilon > 0$ can be chosen arbitrarily small, one has in fact $H^{\prime}(s) = 0$ for all $s \in (0,1)$ (for $n \geq n_{\star}$). Also, since $D_{n}(\vec{\alpha},\vec{\beta},V_{s},0)$ is a continuous function of $s \in [0,1]$ (this can be proved by applying Lebesgue's dominated convergence theorem on the associated moments $w_{j}(\vec{\alpha},\vec{\beta},V_{s},0)$), $H(s)$ is continuous (and thus constant) for $s \in [0,1]$. Moreover, as $H(0) = D_{n}(\vec{\alpha},\vec{\beta},2x^{2},0) \neq 0$ (if $n$ is chosen sufficiently large, see \eqref{partition GUE}, \eqref{asymptotics alpha Krasovsky} and Proposition \ref{prop:beta}), the determinant $D_{n}(\vec{\alpha},\vec{\beta},V_{s},0)$ is never zero. Hence, the identity \eqref{lol 38} holds for \textit{all} $s \in (0,1)$. By integrating \eqref{lol 38} in $s$ from $0$ to $1$, and by using the above calculations for the quantities $\int_{0}^{1}I_{1,s}ds$,..., this finishes the proof of Proposition \ref{prop:V} (since $\Omega$ is arbitrary).

\section{Integration in $W$}\label{Section: Tau}
In this section, we will use the RH analysis done in Section \ref{Section: steepest descent}, where the potential $V$ is one-cut regular and independent of any parameter, and where $W$ is replaced by $W_{t}$ defined in \eqref{deformation parameter Tau}. The weight depends on $t$ and is denoted by
\begin{equation}\label{lol 13}
w_{t}(x) = e^{-nV(x)}e^{W_{t}(x)}\omega(x).
\end{equation}
We will use the differential identity 
\begin{equation}\label{lol 41}
\partial_{t} \log D_{n}(\vec{\alpha},\vec{\beta},V,W_{t}) =  \frac{1}{2\pi i}\int_{\mathbb{R}}[Y^{-1}(x)Y^{\prime}(x)]_{21}\partial_{t}w_{t}(x)dx,
\end{equation}
which was obtained in \eqref{diff identity t}, and which is valid only for $(\vec{\alpha},\vec{\beta},t) \in \widetilde{\mathcal{P}}_{3}^{(n)}$ (see the discussion at the beginning of Section \ref{Section:diff identities}). The main result of this section is the following.
\begin{proposition} As $n \to \infty$, 
\begin{align*}
& \log \frac{D_{n}(\vec{\alpha},\vec{\beta},V,W)}{D_{n}(\vec{\alpha},\vec{\beta},V,0)} = n\int_{-1}^{1}\psi(x)\sqrt{1-x^{2}}W(x)dx -\frac{1}{4\pi^{2}}\int_{-1}^{1}  \frac{W(y)}{\sqrt{1-y^{2}}} \bigg(\dashint_{-1}^{1} \frac{W^{\prime}(x) \sqrt{1-x^{2}}}{x-y}dx\bigg)dy \\
&  + \frac{\mathcal{A}}{2\pi}\int_{-1}^{1} \frac{W(x)}{\sqrt{1-x^{2}}}dx - \sum_{j=1}^{m} \frac{\alpha_{j}}{2}W(t_{j}) +  \sum_{j=1}^{m} \frac{i \beta_{j}}{\pi} \sqrt{1-t_{j}^{2}} \dashint_{-1}^{1} \frac{W(x)}{\sqrt{1-x^{2}}(t_{j}-x)}dx + \bigO\big(n^{-1+2\beta_{\max}}\big).
\end{align*}
\end{proposition}
By a very similar calculation to the one done at the beginning of section \ref{Section: integrating V}, we obtain exactly the same equation as \eqref{lol 4}, except that $s$ is replaced by $t$ and $\partial_{s}\log w_{s}(z)$ is replaced by $\partial_{t} \log w_{t}(z) = \partial_{t} W_{t}(z)$:
\begin{equation}\label{lol 25}
\begin{array}{r c l}
\displaystyle \partial_{t} \log D_{n}(\vec{\alpha},\vec{\beta},V,W_{t}) & = & \displaystyle I_{1,t} + I_{2,t} + I_{3,t} + \bigO(e^{-cn}), \\[0.3cm]
\displaystyle I_{1,t} & = & \displaystyle \frac{-n}{2\pi i} \int_{\mathcal{C}} g^{\prime}(z) \partial_{t} \log w_{t}(z) dz, \\[0.3cm]
\displaystyle I_{2,t} & = & \displaystyle \frac{-1}{2\pi i} \int_{\mathcal{C}} [P^{(\infty)}(z)^{-1}P^{(\infty)}(z)^{\prime}]_{11} \partial_{t} \log w_{t}(z) dz, \\[0.3cm]
\displaystyle I_{3,t} & = & \displaystyle \frac{-1}{2\pi i} \int_{\mathcal{C}} [P^{(\infty)}(z)^{-1}R^{-1}(z)R^{\prime}(z)P^{(\infty)}(z)]_{11} \partial_{t} \log w_{t}(z) dz.
\end{array}
\end{equation}
We thus have, as $n \to \infty$,
\begin{equation}
I_{1,t} = \bigO(n), \qquad I_{2,t} = \bigO(1), \qquad I_{3,t} = \bigO(n^{-1+2\beta_{\max}}),
\end{equation}
and therefore
\begin{equation}\label{lol 42}
\partial_{t} \log D_{n}(\vec{\alpha},\vec{\beta},V,W_{t}) = I_{1,t} + I_{2,t} + \bigO(n^{-1+2\beta_{\max}}).
\end{equation}
Let us denote $F_{3,n}(\vec{\alpha},\vec{\beta},t)$ for the r.h.s. of \eqref{lol 41}. We recall that (see the beginning of Section \ref{Section:diff identities} and Subsection \ref{Subsection: small-norm RH problem}), if $\Omega$ is a compact subset of $\mathcal{P}_{\alpha}\times \mathcal{P}_{\beta}^{(\frac{1}{4})}$, there exists $n_{\star} = n_{\star}(\Omega)$ such that $F_{3,n}(\vec{\alpha},\vec{\beta},t)$ exists for \textit{all} $(\vec{\alpha},\vec{\beta},t) \in \Omega\times[0,1]$. Furthermore, the large $n$ asymptotics of $F_{3,n}(\vec{\alpha},\vec{\beta},t)$ given by the r.h.s. of \eqref{lol 42} (more explicit expressions for $I_{1,t}$ and $I_{2,t}$ will be computed below) are valid uniformly for \textit{all} $(\vec{\alpha},\vec{\beta},t) \in \Omega\times [0,1]$, but the identity \eqref{lol 42} itself is valid only for $(\vec{\alpha},\vec{\beta},t) \in (\Omega\times[0,1])\setminus \widetilde{\Omega}_{3}^{(n)}$, where $\widetilde{\Omega}_{3}^{(n)}$ consists of at most a finite number of points. For sufficiently large $n$, we can extend the identity \eqref{lol 42} from $(\Omega\times[0,1])\setminus \widetilde{\Omega}_{3}^{(n)}$ to $\Omega\times[0,1]$. The argument is similar to the one presented in Section \ref{Section: beta}, below equation \eqref{Explicit diff identity}, and we omit the discussion here. In the rest of this section we compute more explicitly $I_{1,t}$ and $I_{2,t}$. From \eqref{g+ - g- 2}, we have
\begin{equation}
I_{1,t} = n\int_{-1}^{1}\psi(x)\sqrt{1-x^{2}}\partial_{t} W_{t}(x)dx,
\end{equation}
and therefore
\begin{equation}
\int_{0}^{1} I_{1,t} dt = n\int_{-1}^{1}\psi(x)\sqrt{1-x^{2}}W(x)dx.
\end{equation}
Also, from \eqref{Pinf Region 1}, we have $[P^{(\infty)}(z)^{-1}P^{(\infty)}(z)^{\prime}]_{11} = - \partial_{z} \log D(z)$, and thus
\begin{align*}
& I_{2,t} = I_{2,t,W} + I_{2,t,\alpha} + I_{2,t,\beta}, \hspace{0.3cm} \qquad \mbox{ where } \qquad I_{2,t,W} = \frac{1}{2\pi i} \int_{\mathcal{C}} \partial_{z} \log D_{W_{t}}(z)\partial_{t}W_{t}(z)dz, \\
& I_{2,t,\alpha} = \frac{1}{2\pi i} \int_{\mathcal{C}} \partial_{z} \log D_{\alpha}(z)\partial_{t}W_{t}(z)dz, \qquad  \qquad  I_{2,t,\beta} = \frac{1}{2\pi i} \int_{\mathcal{C}} \partial_{z} \log D_{\beta}(z)\partial_{t}W_{t}(z)dz.
\end{align*}
From \cite[equation (5.14) and Lemma 5.4]{BerWebbWong}, we have
\begin{align}
& \int_{0}^{1} I_{2,t,W} dt = -\frac{1}{4\pi^{2}}\int_{-1}^{1}  \frac{W(y)}{\sqrt{1-y^{2}}} \bigg(\dashint_{-1}^{1} \frac{W^{\prime}(x) \sqrt{1-x^{2}}}{x-y}dx\bigg)dy, \\
& \int_{0}^{1} I_{2,t,\alpha} dt = \frac{\mathcal{A}}{2\pi}\int_{-1}^{1} \frac{W(x)}{\sqrt{1-x^{2}}}dx - \sum_{j=1}^{m} \frac{\alpha_{j}}{2}W(t_{j}).
\end{align}
What remains to be done is to evaluate more explicitly $I_{2,t,\beta}$ and integrate it in $t$. This can be achieved from \eqref{lol 5} and a contour deformation of $\mathcal{C}$, and gives
\begin{equation}
\int_{0}^{1} I_{2,t,\beta} dt = \sum_{j=1}^{m} \frac{i \beta_{j}}{\pi} \sqrt{1-t_{j}^{2}} \dashint_{-1}^{1} \frac{W(x)}{\sqrt{1-x^{2}}(t_{j}-x)}dx.
\end{equation}
Since the $\bigO$ terms in \eqref{asymptotics for R} are uniform for $t \in [0,1]$ (see Remark \ref{Remark: R} (b)), we have $\int_{0}^{1} I_{3,t} dt = \bigO\big(n^{-1+2\beta_{\max}}\big)$ as $n \to \infty$, which finishes the proof.

\section{Appendix}\label{Section:Appendix}
In this section we recall some well-known model RH problems: 1) the Airy model RH problem, whose solution is denoted $\Phi_{\mathrm{Ai}}$ and 2) the confluent hypergeometric model RH problem, whose solution is denoted $\Phi_{\mathrm{HG}}(z) = \Phi_{\mathrm{HG}}(z;\alpha,\beta)$. We will only consider the cases when $\alpha$ and $\beta$ are such that $\Re \alpha >-1$ and $\Re \beta \in (-\frac{1}{4},\frac{1}{4})$.

\subsection{Airy model RH problem}
\begin{itemize}
\item[(a)] $\Phi_{\mathrm{Ai}} : \mathbb{C} \setminus \Sigma_{A} \rightarrow \mathbb{C}^{2 \times 2}$ is analytic, and $\Sigma_{A}$ is shown in Figure \ref{figAiry}.
\item[(b)] $\Phi_{\mathrm{Ai}}$ has the jump relations
\begin{equation}\label{jumps P3}
\begin{array}{l l}
\Phi_{\mathrm{Ai},+}(z) = \Phi_{\mathrm{Ai},-}(z) \begin{pmatrix}
0 & 1 \\ -1 & 0
\end{pmatrix}, & \mbox{ on } \mathbb{R}^{-}, \\

\Phi_{\mathrm{Ai},+}(z) = \Phi_{\mathrm{Ai},-}(z) \begin{pmatrix}
 1 & 1 \\
 0 & 1
\end{pmatrix}, & \mbox{ on } \mathbb{R}^{+}, \\

\Phi_{\mathrm{Ai},+}(z) = \Phi_{\mathrm{Ai},-}(z) \begin{pmatrix}
 1 & 0  \\ 1 & 1
\end{pmatrix}, & \mbox{ on } e^{ \frac{2\pi i}{3} }  \mathbb{R}^{+} , \\

\Phi_{\mathrm{Ai},+}(z) = \Phi_{\mathrm{Ai},-}(z) \begin{pmatrix}
 1 & 0  \\ 1 & 1
\end{pmatrix}, & \mbox{ on }e^{ -\frac{2\pi i}{3} }\mathbb{R}^{+} . \\
\end{array}
\end{equation}
\item[(c)] As $z \to \infty$, $z \notin \Sigma_{A}$, we have
\begin{equation}\label{Asymptotics Airy}
\Phi_{\mathrm{Ai}}(z) = z^{-\frac{\sigma_{3}}{4}}N \left( I + \sum_{k=1}^{\infty} \frac{\Phi_{\mathrm{Ai,k}}}{z^{3k/2}} \right) e^{-\frac{2}{3}z^{3/2}\sigma_{3}},
\end{equation}
where $N = \frac{1}{\sqrt{2}}\begin{pmatrix}
1 & i \\ i & 1
\end{pmatrix}$ and $\Phi_{\mathrm{Ai,1}} = \frac{1}{8}\begin{pmatrix}
\frac{1}{6} & i \\ i & -\frac{1}{6}
\end{pmatrix}$.

As $z \to 0$, we have
\begin{equation}
\Phi_{\mathrm{Ai}}(z) = \bigO(1).
\end{equation} 
\end{itemize}
The Airy model RH problem was introduced and solved in \cite{DKMVZ1} (see also \cite[equation (7.30)]{DKMVZ1}, where explicit forms for the constant matrices $\Phi_{\mathrm{Ai,k}}$ can be found). We have
\begin{figure}[t]
    \begin{center}
    \setlength{\unitlength}{1truemm}
    \begin{picture}(100,55)(-5,10)
        \put(50,40){\line(1,0){30}}
        \put(50,40){\line(-1,0){30}}
        \put(50,39.8){\thicklines\circle*{1.2}}
        \put(50,40){\line(-0.5,0.866){15}}
        \put(50,40){\line(-0.5,-0.866){15}}
        \qbezier(53,40)(52,43)(48.5,42.598)
        \put(53,43){$\frac{2\pi}{3}$}
        \put(50.3,36.8){$0$}
        \put(65,39.9){\thicklines\vector(1,0){.0001}}
        \put(35,39.9){\thicklines\vector(1,0){.0001}}
        \put(41,55.588){\thicklines\vector(0.5,-0.866){.0001}}
        \put(41,24.412){\thicklines\vector(0.5,0.866){.0001}}
    \end{picture}
    \caption{\label{figAiry}The jump contour $\Sigma_{A}$ for $\Phi_{\mathrm{Ai}}$.}
\end{center}
\end{figure}
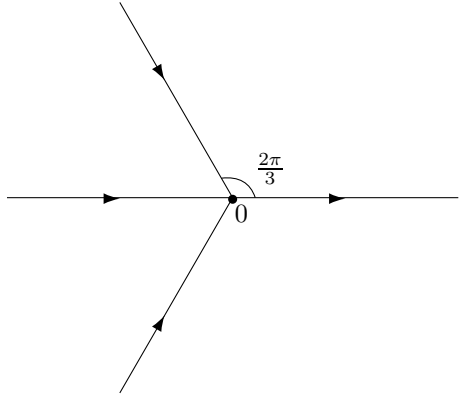
\begin{equation}
\Phi_{\mathrm{Ai}}(z) := M_{A} \times \left\{ \begin{array}{l l}
\begin{pmatrix}
\mbox{Ai}(z) & \mbox{Ai}(\omega^{2}z) \\
\mbox{Ai}^{\prime}(z) & \omega^{2}\mbox{Ai}^{\prime}(\omega^{2}z)
\end{pmatrix}e^{-\frac{\pi i}{6}\sigma_{3}}, & \mbox{for } 0 < \arg z < \frac{2\pi}{3}, \\
\begin{pmatrix}
\mbox{Ai}(z) & \mbox{Ai}(\omega^{2}z) \\
\mbox{Ai}^{\prime}(z) & \omega^{2}\mbox{Ai}^{\prime}(\omega^{2}z)
\end{pmatrix}e^{-\frac{\pi i}{6}\sigma_{3}}\begin{pmatrix}
1 & 0 \\ -1 & 1
\end{pmatrix}, & \mbox{for } \frac{2\pi}{3} < \arg z < \pi, \\
\begin{pmatrix}
\mbox{Ai}(z) & - \omega^{2}\mbox{Ai}(\omega z) \\
\mbox{Ai}^{\prime}(z) & -\mbox{Ai}^{\prime}(\omega z)
\end{pmatrix}e^{-\frac{\pi i}{6}\sigma_{3}}\begin{pmatrix}
1 & 0 \\ 1 & 1
\end{pmatrix}, & \mbox{for } -\pi < \arg z < -\frac{2\pi}{3}, \\
\begin{pmatrix}
\mbox{Ai}(z) & - \omega^{2}\mbox{Ai}(\omega z) \\
\mbox{Ai}^{\prime}(z) & -\mbox{Ai}^{\prime}(\omega z)
\end{pmatrix}e^{-\frac{\pi i}{6}\sigma_{3}}, & \mbox{for } -\frac{2\pi}{3} < \arg z < 0, \\
\end{array} \right.
\end{equation}
with $\omega = e^{\frac{2\pi i}{3}}$, Ai the Airy function and
\begin{equation}
M_{A} := \sqrt{2 \pi} e^{\frac{\pi i}{6}} \begin{pmatrix}
1 & 0 \\ 0 & -i
\end{pmatrix}.
\end{equation}

\subsection{Confluent hypergeometric model RH problem}

\begin{itemize}
\item[(a)] $\Phi_{\mathrm{HG}} : \mathbb{C} \setminus \Sigma_{\mathrm{HG}} \rightarrow \mathbb{C}^{2 \times 2}$ is analytic, where $\Sigma_{\mathrm{HG}}$ is shown in Figure \ref{Fig:HG}.
\item[(b)] For $z \in \Gamma_{k}$ (see Figure \ref{Fig:HG}), $k = 1,...,8$, $\Phi_{\mathrm{HG}}$ has the jump relations
\begin{equation}\label{jumps PHG3}
\Phi_{\mathrm{HG},+}(z) = \Phi_{\mathrm{HG},-}(z)J_{k},
\end{equation}
where
\begin{align*}
& J_{1} = \begin{pmatrix}
0 & e^{-i\pi \beta} \\ -e^{i\pi\beta} & 0
\end{pmatrix}, \quad J_{5} = \begin{pmatrix}
0 & e^{i\pi\beta} \\ -e^{-i\pi\beta} & 0
\end{pmatrix},\quad J_{3} = J_{7} = \begin{pmatrix}
e^{\frac{i\pi\alpha}{2}} & 0 \\ 0 & e^{-\frac{i\pi\alpha}{2}}
\end{pmatrix}, \\
& J_{2} = \begin{pmatrix}
1 & 0 \\ e^{-i\pi\alpha}e^{i\pi\beta} & 1
\end{pmatrix}\hspace{-0.1cm}, \hspace{-0.3cm} \quad J_{4} = \begin{pmatrix}
1 & 0 \\ e^{i\pi\alpha}e^{-i\pi\beta} & 1
\end{pmatrix}\hspace{-0.1cm}, \hspace{-0.3cm} \quad J_{6} = \begin{pmatrix}
1 & 0 \\ e^{-i\pi\alpha}e^{-i\pi\beta} & 1
\end{pmatrix}\hspace{-0.1cm}, \hspace{-0.3cm} \quad J_{8} = \begin{pmatrix}
1 & 0 \\ e^{i\pi\alpha}e^{i\pi\beta} & 1
\end{pmatrix}.
\end{align*}
\item[(c)] As $z \to \infty$, $z \notin \Sigma_{\mathrm{HG}}$, we have
\begin{equation}\label{Asymptotics HG}
\Phi_{\mathrm{HG}}(z) = \left( I + \sum_{k=1}^{\infty} \frac{\Phi_{\mathrm{HG},k}}{z^{k}} \right) z^{-\beta\sigma_{3}}e^{-\frac{z}{2}\sigma_{3}}M^{-1}(z),
\end{equation}
where 
\begin{equation}\label{def of tau}
\Phi_{\mathrm{HG},1} = \Big(\beta^{2}-\frac{\alpha^{2}}{4}\Big) \begin{pmatrix}
-1 & \tau(\alpha,\beta) \\ - \tau(\alpha,-\beta) & 1
\end{pmatrix}, \qquad \tau(\alpha,\beta) = \frac{- \Gamma\left( \frac{\alpha}{2}-\beta \right)}{\Gamma\left( \frac{\alpha}{2}+\beta + 1 \right)},
\end{equation}
and
\begin{equation}
M(z) = \left\{ \begin{array}{l l}
\displaystyle e^{\frac{i\pi\alpha}{4} \sigma_{3}}e^{- i\pi\beta  \sigma_{3}}, & \displaystyle \frac{\pi}{2} < \arg z < \pi, \\
\displaystyle e^{-\frac{i\pi\alpha}{4} \sigma_{3}}e^{-i\pi\beta  \sigma_{3}}, & \displaystyle \pi < \arg z < \frac{3\pi}{2}, \\
e^{\frac{i\pi\alpha}{4}\sigma_{3}} \begin{pmatrix}
0 & 1 \\ -1 & 0
\end{pmatrix}, & -\frac{\pi}{2} < \arg z < 0, \\
e^{-\frac{i\pi\alpha}{4}\sigma_{3}} \begin{pmatrix}
0 & 1 \\ -1 & 0
\end{pmatrix}, & 0 < \arg z < \frac{\pi}{2}.
\end{array} \right.
\end{equation}
In \eqref{Asymptotics HG}, $z^{-\beta}$ has a cut along $i\mathbb{R}^{-}$, such that $z^{-\beta} \in \mathbb{R}$ as $z \in \mathbb{R}^{+}$.

As $z \to 0$, we have
\begin{equation}\label{lol 35}
\begin{array}{l l}
\displaystyle \Phi_{\mathrm{HG}}(z) = \left\{ \begin{array}{l l}
\begin{pmatrix}
\bigO(1) & \bigO(\log z) \\
\bigO(1) & \bigO(\log z)
\end{pmatrix}, & \mbox{if } z \in II \cup III \cup VI \cup VII, \\
\begin{pmatrix}
\bigO(\log z) & \bigO(\log z) \\
\bigO(\log z) & \bigO(\log z)
\end{pmatrix}, & \mbox{if } z \in I\cup IV \cup V \cup VIII,
\end{array} \right., & \displaystyle \mbox{ if } \Re \alpha = 0, \\[0.9cm]

\displaystyle \Phi_{\mathrm{HG}}(z) = \left\{ \begin{array}{l l}
\begin{pmatrix}
\bigO(z^{\frac{\Re\alpha}{2}}) & \bigO(z^{-\frac{\Re\alpha}{2}}) \\
\bigO(z^{\frac{\Re\alpha}{2}}) & \bigO(z^{-\frac{\Re\alpha}{2}})
\end{pmatrix}, & \mbox{if } z \in II \cup III \cup VI \cup VII, \\
\begin{pmatrix}
\bigO(z^{-\frac{\Re\alpha}{2}}) & \bigO(z^{-\frac{\Re\alpha}{2}}) \\
\bigO(z^{-\frac{\Re\alpha}{2}}) & \bigO(z^{-\frac{\Re\alpha}{2}})
\end{pmatrix}, & \mbox{if } z \in I\cup IV \cup V \cup VIII,
\end{array} \right. , & \displaystyle \mbox{ if } \Re \alpha > 0, \\[0.9cm]

\displaystyle \Phi_{\mathrm{HG}}(z) = \begin{pmatrix}
\bigO(z^{\frac{\Re\alpha}{2}}) & \bigO(z^{\frac{\Re\alpha}{2}}) \\
\bigO(z^{\frac{\Re\alpha}{2}}) & \bigO(z^{\frac{\Re\alpha}{2}}) 
\end{pmatrix}, & \displaystyle \mbox{ if } \Re \alpha < 0.
\end{array}
\end{equation}
\end{itemize}
\begin{figure}[t!]
    \begin{center}
    \setlength{\unitlength}{1truemm}
    \begin{picture}(100,55)(-5,10)
        \put(50,40){\line(-1,0){26}}
        \put(50,40){\line(1,0){26}}        
        \put(50,39.8){\thicklines\circle*{1.2}}
        \put(50,40){\line(-0.5,0.5){21}}
        \put(50,40){\line(-0.5,-0.5){21}}
        \put(50,40){\line(0.5,0.5){21}}
        \put(50,40){\line(0.5,-0.5){21}}
        \put(50,40){\line(0,1){25}}
        \put(50,40){\line(0,-1){25}}
        \put(50.3,35){$0$}
                \put(76.5,39){$\Gamma_7$}        
                \put(71,62){$\Gamma_8$}        
                \put(49,66){$\Gamma_1$}        
                \put(25.8,62.3){$\Gamma_2$}        
                \put(19.7,39){$\Gamma_3$}        
                \put(24.5,17.5){$\Gamma_4$}
                \put(49,11.5){$\Gamma_5$}
                \put(71,17){$\Gamma_6$}        
        \put(30,39.9){\thicklines\vector(-1,0){.0001}}
        \put(68,39.9){\thicklines\vector(1,0){.0001}}
        \put(32,58){\thicklines\vector(-0.5,0.5){.0001}}
        \put(35,25){\thicklines\vector(0.5,0.5){.0001}}
        \put(68,58){\thicklines\vector(0.5,0.5){.0001}}
        \put(65,25){\thicklines\vector(-0.5,0.5){.0001}}
        \put(50,60){\thicklines\vector(0,1){.0001}}
        \put(50,25){\thicklines\vector(0,1){.0001}}
        \put(40,57){$I$}
        \put(30,49){$II$}
        \put(30,32){$III$}
        \put(40,23){$IV$}
        \put(56,23){$V$}
        \put(66,34){$VI$}
        \put(66,46){$VII$}
        \put(55,57){$VIII$}
    \end{picture}
    \caption{\label{Fig:HG}The jump contour $\Sigma_{\mathrm{HG}}$ for $\Phi_{\mathrm{HG}}(z)$. The ray $\Gamma_{k}$ is oriented from $0$ to $\infty$, and forms an angle with $\mathbb{R}^{+}$ which is a multiple of $\frac{\pi}{4}$.}
\end{center}
\end{figure}
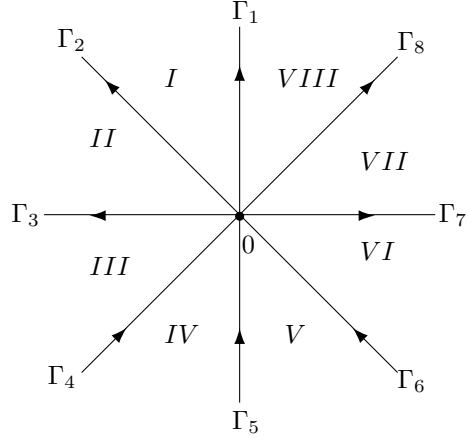
This model RH problem was first introduced and solved explicitly in \cite{ItsKrasovsky} for the case $\alpha = 0$, and then in \cite{FouMarSou} and \cite{DIK} for the general case. The constant matrices $\Phi_{\mathrm{HG},k}$ depend analytically on $\alpha$ and $\beta$ (they can be found explicitly e.g. in \cite[equation (56)]{FouMarSou}). Consider the matrix
\begin{equation}\label{phi_HG}
\widehat{\Phi}_{\mathrm{HG}}(z) = \begin{pmatrix}
\frac{\Gamma(1 + \frac{\alpha}{2}-\beta)}{\Gamma(1+\alpha)}G(\frac{\alpha}{2}+\beta, \alpha; z)e^{-\frac{i\pi\alpha}{2}} & -\frac{\Gamma(1 + \frac{\alpha}{2}-\beta)}{\Gamma(\frac{\alpha}{2}+\beta)}H(1+\frac{\alpha}{2}-\beta,\alpha;ze^{-i\pi }) \\
\frac{\Gamma(1 + \frac{\alpha}{2}+\beta)}{\Gamma(1+\alpha)}G(1+\frac{\alpha}{2}+\beta,\alpha;z)e^{-\frac{i\pi\alpha}{2}} & H(\frac{\alpha}{2}-\beta,\alpha;ze^{-i\pi })
\end{pmatrix} e^{-\frac{i\pi\alpha}{4}\sigma_{3}},
\end{equation}
where $G$ and $H$ are related to the Whittaker functions:
\begin{equation}\label{relation between G and H and Whittaker}
G(a,\alpha;z) = \frac{M_{\kappa,\mu}(z)}{\sqrt{z}}, \quad H(a,\alpha;z) = \frac{W_{\kappa,\mu}(z)}{\sqrt{z}}, \quad \mu = \frac{\alpha}{2}, \quad \kappa = \frac{1}{2}+\frac{\alpha}{2}-a.
\end{equation}
The solution $\Phi_{\mathrm{HG}}$ is given by
\begin{equation}
\Phi_{\mathrm{HG}}(z) = \left\{ \begin{array}{l l}
\widehat{\Phi}_{\mathrm{HG}}(z)J_{2}^{-1}, & \mbox{ for } z \in I, \\
\widehat{\Phi}_{\mathrm{HG}}(z), & \mbox{ for } z \in II, \\
\widehat{\Phi}_{\mathrm{HG}}(z)J_{3}, & \mbox{ for } z \in III, \\
\widehat{\Phi}_{\mathrm{HG}}(z)J_{3}J_{4}^{-1}, & \mbox{ for } z \in IV, \\
\widehat{\Phi}_{\mathrm{HG}}(z)J_{2}^{-1}J_{1}^{-1}J_{8}^{-1}J_{7}^{-1}J_{6}, & \mbox{ for } z \in V, \\
\widehat{\Phi}_{\mathrm{HG}}(z)J_{2}^{-1}J_{1}^{-1}J_{8}^{-1}J_{7}^{-1}, & \mbox{ for } z \in VI, \\
\widehat{\Phi}_{\mathrm{HG}}(z)J_{2}^{-1}J_{1}^{-1}J_{8}^{-1}, & \mbox{ for } z \in VII, \\
\widehat{\Phi}_{\mathrm{HG}}(z)J_{2}^{-1}J_{1}^{-1}, & \mbox{ for } z \in VIII. \\
\end{array} \right.
\end{equation}

\section*{Acknowledgements}
CC was supported by the European Research Council under the European Union's Seventh Framework Programme (FP/2007/2013)/ ERC Grant Agreement n.\, 307074. The author acknowledges both referees for their careful reading and useful remarks that have helped to improve the present paper.

\end{document}